\newif\ifarxiv\arxivtrue%

\documentclass[a4paper,UKenglish,cleveref,numberwithinsect]{lipics-v2019}

\usepackage{silence}

\ifarxiv\nolinenumbers\else%
\fi%

\usepackage{amsmath}
\usepackage{amssymb}
\usepackage{arydshln}
\usepackage{amsthm}

\usepackage{stmaryrd}
\SetSymbolFont{stmry}{bold}{U}{stmry}{m}{n}

\usepackage{mathpartir}
\usepackage{upgreek}
\usepackage{microtype}
\usepackage{bm}

\usepackage{tikz}
\usetikzlibrary{positioning}
\usetikzlibrary{calc}

\usepackage{aliascnt}
\usepackage{hyperref}
\usepackage{paralist}
\usepackage{xspace}

\usepackage{thm-restate}
\input{fix-restatable.tex} 
\usepackage{cleveref}

\usepackage[colorinlistoftodos]{todonotes}

\usepackage{scalerel}

\usepackage{enumitem}

\newcommand{\alphabet}{\Sigma}

\newcommand{\thealgebra}{\emph{observation} algebra\xspace}
\newcommand{\acronym}{OA}
\newcommand{\PL}{\ensuremath{\mathsf{\acronym}}\xspace}

\newcommand{\BKA}{\ensuremath{\mathsf{BKA}}\xspace}

\newcommand{\val}{\ensuremath{\normalfont\textsc{Val}}\xspace}
\newcommand{\var}{\ensuremath{\normalfont\textsc{Var}}\xspace}

\newcommand{\termsbka}[1][]{\mathcal{E}_{#1}}
\newcommand{\termspl}{\mathcal{O}}
\newcommand{\termspocka}{\mathcal{T}}

\newcommand{\equivt}[1]{\equiv_{\scriptscriptstyle#1}}
\newcommand{\equivpl}{\equivt{}}
\newcommand{\equivplexpl}{\equivt{\PL}}
\newcommand{\equivpocka}{\equivt{}}
\newcommand{\equivbka}{\equivt{}}
\newcommand{\equivbkaexpl}{\equivt{\BKA}}

\newcommand{\leqqt}[1]{\leqq_{\scriptscriptstyle#1}}
\newcommand{\leqqpl}{\leqqt{}}
\newcommand{\leqqplexpl}{\leqqt{\PL}}
\newcommand{\leqqpocka}{\leqqt{}}
\newcommand{\leqqbka}{\leqqt{}}
\newcommand{\leqqbkaexpl}{\leqqt{\BKA}}

\newcommand{\pnegop}{\overline{\vphantom{i}\hspace{0.5mm}\cdot\hspace{0.5mm}}}

\newcommand\closure[2][H]{{#2}{\downarrow^{#1}}}
\newcommand\closurep[2][H]{\left(#2\right){\downarrow^{#1}}}

\newcommand{\sem}[1]{{\left\llbracket#1\right\rrbracket}}
\newcommand{\semt}[2]{\sem{#2}_{\scriptscriptstyle#1}}

\newcommand{\sembka}{\sem}
\newcommand{\sempl}{\sem}
\newcommand{\semplexpl}{\semt{\PL}}

\newcommand{\rsem}[1]{\stretchleftright{\llparenthesis}{\displaystyle\makebox[0pt][c]{\color{white} $\beta$}{#1}}{\rrparenthesis}}
\newcommand{\sempocka}[1]{\rsem{#1}{\downarrow}}
\newcommand{\semwpocka}{\rsem}

\newcommand{\angl}[1]{\left\langle#1\right\rangle}
\newcommand{\pipe}{\;\;|\;\;}

\newcommand{\naturals}{\mathbb{N}}

\newcommand{\lp}[1]{\mathbf{#1}}
\newcommand{\ltr}[1]{\mathtt{#1}}

\newcommand{\Pom}{\ensuremath{\mathsf{Pom}}}
\newcommand{\SP}{\ensuremath{\mathsf{SP}}}

\newcommand{\PC}{\mathsf{PC}}

\newcommand{\hexch}{\mathsf{exch}}

\newcommand{\hcontr}{\mathsf{contr}}

\newcommand{\htop}{\mathsf{top}}

\newcommand{\Act}{\mathsf{Act}}
\newcommand{\Obs}{\mathsf{Obs}}
\newcommand{\State}{\mathsf{State}}

\newcommand{\dom}{\mathsf{dom}}
\newcommand{\G}{\mathcal{G}}

\usepackage{thm-restate}

\newif\iffull\fullfalse%

\theoremstyle{plain}

\crefname{lem}{Lemma}{Lemmas}
\crefname{prop}{Proposition}{Propositions}

\title{Partially Observable Concurrent Kleene Algebra}
\author{Jana Wagemaker}{Radboud University, Nijmegen}{j.wagemaker@cs.ru.nl}{http://orcid.org/0000-0002-8616-3905}{} 
\author{Paul Brunet}{University College London}{}{https://orcid.org/0000-0002-9762-6872}{EPSRC grant IRIS (EP/R006865/1)} 
\author{Simon Docherty}{University College London}{}{http://orcid.org/0000-0001-7523-6630}{EPSRC grant (EP/S013008/1)} 
\author{Tobias Kapp\'{e}}{University College London}{}{https://orcid.org/0000-0002-6068-880X}{} 
\author{Jurriaan Rot}{Radboud University, Nijmegen}{}{}{}
\author{Alexandra Silva}{University College London}{}{https://orcid.org/0000-0001-5014-9784}{} 

\funding{ERC Starting Grant ProFoundNet (679127)}

\relatedversion{
A version with detailed proofs
is available at~\url{https://arxiv.org/abs/2007.07593}.
}

\authorrunning{J. Wagemaker, P. Brunet, S. Docherty, T. Kapp\'{e}, J. Rot, and A. Silva}

\Copyright{J. Wagemaker, P. Brunet, S. Docherty, T. Kapp\'{e}, J. Rot, and A. Silva}


\acknowledgements{}

\usepackage{etoolbox}
\makeatletter
\patchcmd{\@addmarginpar}{\ifodd\c@page}{\ifodd\c@page\@tempcnta\m@ne}{}{}
\makeatother
\reversemarginpar%

\begin{document}

\maketitle

\begin{abstract}
We introduce partially observable concurrent Kleene algebra (POCKA), an algebraic framework to reason about concurrent programs with variables as well as control structures, such as conditionals and loops, that depend on those variables.
We illustrate the use of POCKA through concrete examples.
We prove that POCKA is a sound and complete axiomatisation of a model of partial observations, and show the semantics passes an important check for sequential consistency.
\end{abstract}

\section{Introduction}

Kleene Algebra (KA) was originally proposed as the algebra of regular languages~\cite{salomaa-1966,conway-1971,krob-1990,kozen-1994}, but its well-developed meta-theory facilitates applications in the analysis and verification of sequential programs.
Many extensions of KA were studied in the last decades, notably Kleene Algebra with Tests (KAT)~\cite{kozen-1996}, which enables reasoning about control structures such as $\mathsf{if}$-statements and $\mathsf{while}$-loops.
Orthogonally, Concurrent Kleene Algebra (CKA) was proposed as an extension of KA to analyse concurrent program behaviour~\cite{hoare-moeller-struth-wehrman-2009}.

It is a natural question whether concurrent Kleene algebra can be extended with tests as in KAT\@.
This question was studied by Jipsen~\cite{jipsen-moshier-2016} and later by Kapp\'{e} et al.~\cite{kappe-brunet-rot-silva-wagemaker-zanasi-2019,fossacs2020},
who proposed Concurrent Kleene Algebra with Observations (CKAO). Observations are tests in a concurrent setting, and they are governed by different axioms than tests, hence justifying their name change.
It was illustrated that extending CKA with tests in a naive way results in an algebraic framework that is unusable in program verification.
In a nutshell, the interactions of parallel threads are lost if we identify the conjunction of observations with their sequential composition, as is done in KAT\@.
Instead, an algebra where conjunction and sequential composition are kept distinct is essential to capture concurrent interaction between conditionals in different threads --- this distinguishes tests from observations.

In this paper we demonstrate how this class of techniques can be used for a more fine-grained analysis of concurrent programs.
We focus our development around the issue of \emph{sequential consistency}, i.e., whether programs behave as if memory accesses taking place were interleaved and executed sequentially~\cite{lampie}.
A standard way of testing this property is the so-called \emph{store buffering litmus test}~\cite{litmus}.
Consider the following program with two threads:
\[\begin{array}{ll||ll} 
  \textsf{T0:}& x\leftarrow 1; & \textsf{T1:}&  y\leftarrow 1;\\
              & r_0\leftarrow y;&&r_1\leftarrow x;
\end{array}\]
A sequentially consistent implementation should satisfy the following property: if initially both registers $r_0$ and $r_1$ are set to $0$, after running the program one of them should have value $1$.
Therefore, we can detect failures of sequential consistency by observing behaviour that deviates from this specification.
This test can be encoded algebraically~\cite{kozen-2000} as:
\begin{equation}%
\label{equation:sequential-consistency}
\Big((r_0=0\wedge r_1=0) ; (\textsf{T0} \parallel \textsf{T1}) ; \neg(r_0=1\vee r_1=1)\Big) \equiv 0.
\tag{\textdagger}
\end{equation}
That is, the program that asserts that $r_0$ and $r_1$ are both $0$, executes $\mathsf{T0}$ and $\mathsf{T1}$ in parallel, and then asserts that neither $r_0$ nor $r_1$ is $1$, is equivalent to the program $0$, which has no valid behaviour.
To reason in this fashion, our algebraic framework should include \emph{observations} of the shape $v=n$ as well as \emph{assignments} $v \leftarrow n$ and $v \leftarrow v'$, where $v,v'$ and $n$ range over some fixed sets of variables and values.
To that end, we propose \emph{Partially Observable Concurrent Kleene Algebra} (\emph{POCKA}), an algebraic theory built on top of CKA that allows for an analysis of concurrent programs manipulating memory, such as the simple program above.
POCKA has a natural interpretation in terms of pomset languages over assignments
and memory states, encoded as partial functions, similarly to separation logic~\cite{reynolds}, which describe the behaviour of concurrent programs that can access variables and values (Section~\ref{sec:pocka}).
We prove soundness and completeness with respect to this interpretation (Section~\ref{sec:complete}).

POCKA deviates from KAT and CKAO by using \emph{partial} observations in its semantics.
These are crucial in a concurrent setting, where a single thread may have only a partial view of the memory.
Whilst memory as a whole depends on the combined action of all threads, these partial views may be analysed on a thread-by-thread basis.
This shift from total to partial observations thus allows for a richer compositional semantic model.
Formally, this means that we move from a Boolean algebra of observations, as in CKAO or KAT, to a pseudocomplemented  distributive lattice (PCDL)~\cite{blyth-2005}, as proposed by Jipsen and Moshier~\cite{jipsen-moshier-2016}.

To ensure compositionality, semantics of concurrent programs should capture not only isolated program behaviour, but also all \emph{possible} behaviours of the program when run in parallel with another program.
For example, take the program $P = (x=1); (x=2)$, which asserts that $x$ has value $1$ and then value $2$, and the program $Q = (x \leftarrow 2)$, which assigns the value $2$ to $x$.
In an interpretation that captures isolated program behaviour, the semantics of $P$ would be empty, as $x$ cannot change between the tests.
In contrast, the program $P \parallel Q$ (i.e., $P$ and $Q$ in parallel) \emph{does} have behaviour, because the assignment may be interleaved between the two observations.
Hence, the isolated semantics of $P$ is not sufficient.

Thus, the semantics of a POCKA term accommodates possible interference by an outside context.
As a result, the test (\ref{equation:sequential-consistency}) fails at this stage, meaning this semantics is not sequentially consistent.
This raises the question of how to study the isolated program behaviour.
To this end, we identify a subset of the semantics that captures isolated program behaviour, and show that this fragment coincides with guarded pomsets~\cite{jipsen-moshier-2016} (Section~\ref{section:causal}).
This turns out to fix the defect in sequential consistency we observe earlier, as we show in Section~\ref{section:litmus}.


\section{Preliminaries}

Throughout this section we fix a finite alphabet $\alphabet$.
We recall pomsets~\cite{gischer-1988,grabowski-1981}, a generalisation of words that model concurrent traces.
First, a \emph{labelled poset} over $\alphabet$ is a tuple $\lp{u} = \angl{S_\lp{u}, \leq_\lp{u}, \lambda_\lp{u}}$, where $S_\lp{u}$ is a finite set (the \emph{carrier} of $\lp{u}$), $\leq_\lp{u}$
is a partial order on $S_\lp{u}$ (the \emph{order} of $\lp{u}$), and $\lambda_\lp{u} \colon S_\lp{u} \to \alphabet$ is a function (the \emph{labelling} of $\lp{u}$).
Pomsets are labelled posets up to isomorphism:

\begin{definition}[Poset isomorphism, pomset]%
\label{definition:lp-isomorphism-pomset}
Let $\lp{u}, \lp{v}$ be labelled posets over $\alphabet$.
We say $\lp{u}$ is \emph{isomorphic} to $\lp{v}$, denoted $\lp{u} \cong \lp{v}$, if there exists a bijection $h\colon S_\lp{u} \to S_\lp{v}$ that preserves labels, and preserves and reflects ordering.
More precisely, we require that $\lambda_\lp{v} \circ h = \lambda_\lp{u}$, and $s \leq_\lp{u} s'$ if and only if $h(s) \leq_\lp{v} h(s')$.
A \emph{pomset} over $\alphabet$ is an isomorphism class of labelled posets over $\alphabet$, i.e., the class
$[\lp{v}] = \{ \lp{u} \mid \lp{u} \cong \lp{v} \}$ for some labelled poset $\lp{v}$.
\end{definition}

When two pomsets are in scope, we tacitly assume that they are represented by labelled posets with disjoint carriers.
We write $\Pom(\alphabet)$ for the set of pomsets over $\alphabet$, and $1$ for the empty pomset.
When $\ltr{a} \in \alphabet$, we write $\ltr{a}$ for the pomset represented by the labelled poset whose sole element is labelled by $\ltr{a}$.
Pomsets can be composed in sequence and in parallel:

\begin{definition}[Pomset composition]%
\label{definition:pomset-composition}
Let $U = [\lp{u}]$ and $V = [\lp{v}]$ be pomsets over $\alphabet$.

We write $U \parallel V$ for the \emph{parallel composition} of $U$ and $V$, which is the pomset over $\alphabet$ represented by the labelled poset $\lp{u} \parallel \lp{v}$, where
    $S_{\lp{u} \parallel \lp{v}} = S_\lp{u} \cup S_\lp{v}$,
    ${\leq_{\lp{u} \parallel \lp{v}}} = {\leq_\lp{u}} \cup {\leq_\lp{v}}$
    and for $x\in S_\lp{u}$ we have $\lambda_{\lp{u} \parallel \lp{v}}(x)= \lambda_\lp{u}(x)$
    and for $x\in S_\lp{v}$ we let $\lambda_{\lp{u} \parallel \lp{v}}(x)=\lambda_\lp{v}(x)$.

We write $U \cdot V$ for the \emph{sequential composition} of $U$ and $V$, that is, the pomset represented by the labelled poset $\lp{u} \cdot \lp{v}$, where
    $S_{\lp{u} \cdot \lp{v}} = S_{\lp{u} \parallel \lp{v}}$,
    ${\leq_{\lp{u} \cdot \lp{v}}} = {\leq_\lp{u}} \cup {\leq_\lp{v}} \cup (S_\lp{u} \times S_\lp{v})$ and
    $\lambda_{\lp{u} \cdot \lp{v}} = \lambda_{\lp{u} \parallel \lp{v}}$.
\end{definition}

The pomsets that we use can be built using sequential and parallel composition.

\begin{definition}[Series-parallel pomsets]%
\label{definition:pomset-sp}
The set of \emph{series-parallel pomsets (sp-pomsets)} over $\alphabet$, denoted $\SP(\alphabet)$, is the smallest subset of $\Pom(\alphabet)$ such that $1 \in \SP(\alphabet)$ and $\ltr{a} \in \SP(\alphabet)$ for every $\ltr{a} \in \alphabet$, and is furthermore closed under parallel and sequential composition.
\end{definition}

One way of comparing pomsets is to see whether they have the same events and labels, except that one is ``more sequential'' in the sense that more events are ordered.
This is captured by the notion of \emph{subsumption}~\cite{gischer-1988}, defined as follows.

\begin{definition}[Subsumption]%
\label{definition:subsumption}
Let $U = [\lp{u}]$ and $V = [\lp{v}]$.
We say $U$ \emph{is subsumed by} $V$, written $U \sqsubseteq V$, if there exists a label- and order-preserving bijection $h\colon S_\lp{v} \to S_\lp{u}$.
That is, $h$ is a bijection such that $\lambda_\lp{u} \circ h = \lambda_\lp{v}$ and if $s \leq_\lp{v} s'$, then $h(s) \leq_\lp{u} h(s')$.
\end{definition}
In the rest of this paper we only consider the relation $\sqsubseteq$ restricted to series-parallel pomsets.
We will also need the notion of \emph{pomset contexts}~\cite{fossacs2020}.

\begin{definition}
Let $*$ be a symbol not in $\alphabet$.
The set of \emph{pomset contexts}, denoted $\PC(\alphabet)$, is the smallest subset of $\SP(\alphabet \cup \{*\})$ satisfying
\begin{mathpar}
\inferrule{~}{%
    * \in \PC(\alphabet)
}
\and
\inferrule{%
    X \in \SP(\alphabet \cup \{*\}) \\
    C \in \PC(\alphabet)
}{%
    X \cdot C \in \PC(\alphabet) \\
    C \cdot X \in \PC(\alphabet)
}
\and
\inferrule{%
    X \in \SP(\alphabet \cup \{*\}) \\
    C \in \PC(\alphabet)
}{%
    X \parallel C \in \PC(\alphabet)
}
\end{mathpar}
Alternatively, $\PC(\alphabet)$ consists of the sp-pomsets over $\alphabet \cup \{ * \}$ with exactly one occurrence of $*$.
\end{definition}

One can think of $*$ as a gap where another pomset can be inserted: given $C \in \PC$ and $U \in \Pom$, we can insert $U$ into the gap in $C$ to obtain $C[U]$.
More precisely, we define
\begin{mathpar}
*[U] = U
\and
(C \cdot X)[U] = C[U] \cdot X
\and
(X \cdot C)[U] = X \cdot C[U]
\and
(X \parallel C)[U] = X \parallel C[U]
\end{mathpar}
This insertion is well-defined, and can in fact be extended to pomsets in general~\cite{fossacs2020}.
We extend the notation to a set of pomsets $L \subseteq \Pom$ by $C[L] = \left\{C[U] \mid U\in L\right\}$.

\subparagraph*{Bi-Kleene Algebra (BKA): syntax and semantics}
\emph{Bi-Kleene Algebra}~\cite{laurence-struth-2014} adds a binary operator, denoted $\parallel$, to KA, which satisfies a few basic axioms but does not interact with the other KA operators.
\emph{BKA-terms} over $\alphabet$, denoted $\termsbka[\alphabet]$ (the subscript is omitted if it is clear from the context), also called series-rational expressions~\cite{lodaya-weil-2000}, are generated by the grammar
\[
    e, f::= 0 \pipe 1 \pipe \ltr{a} \in \alphabet \pipe e + f \pipe e \cdot f \pipe e \parallel f \pipe e^*
\]

The semantics of a BKA-term is a \emph{pomset language}, i.e., an element of $2^\SP$.
Formally, the BKA-semantics is a function $\sembka{-}\colon \termsbka \to 2^\SP$ defined inductively, as follows:
\begin{align*}
\sembka{0} &= \emptyset
    & \sembka{1} &= \{ 1 \}&
 \sembka{e + f} &= \sembka{e} \cup \sembka{f}
    & \sembka{e \cdot f} &= \sembka{e} \cdot \sembka{f}    \\
    \sembka{e^*} &= \sembka{e}^*    &
\sembka{\ltr{a}} &= \{ \ltr{a} \}
    & \sembka{e \parallel f} &= \sembka{e} \parallel \sembka{f}
\end{align*}
In this definition we use the pointwise lifting of sequential and parallel composition from pomsets to pomset languages, e.g., $L \cdot K = \{U \cdot V \mid U \in L,\ V \in K\}$.
The Kleene star of a pomset language $L$ is defined as $L^* = \bigcup_{n \in \naturals} L^n$, where $L^0 = \{ 1 \}$ and $L^{n+1} = L^n \cdot L$.

We write $\equivbkaexpl$ or simply $\equivbka$ for the smallest congruence on $\termsbka$
generated by the Kleene algebra axiom together with the additional bi-Kleene algebra axioms, which govern the parallel operator $\parallel$; it is associative, commutative, has a unit and distributes over $+$ (\cref{fig:pocka}).
Soundness and completeness of $\equivbkaexpl$ w.r.t.\ the pomset language semantics was proved in~\cite{laurence-struth-2014}:
\begin{theorem}[Soundness and Completeness BKA]%
\label{theorem:bka-completeness}
Let $e, f \in \termsbka$.
Then $e \equivbka f \Leftrightarrow \sembka{e} = \sembka{f}$.
\end{theorem}

Given alphabets $\alphabet$ and $\Gamma$, a function $h\colon\alphabet\to\termsbka[\Gamma]$ extends inductively to a map $\hat h \colon \termsbka[\alphabet]\to\termsbka[\Gamma]$ (e.g., $\hat{h}(e + f) = \hat{h}(e) + \hat{h}(f)$) which we refer to as the \emph{homomorphism generated by $h$}.

\subparagraph*{Concurrent Kleene Algebra with Hypotheses (CKAH)}
\emph{Concurrent Kleene algebra with Hypotheses}~\cite{fossacs2020} (see also~\cite{DoumaneKPP19} for the case of KA), allows for a set of additional axioms, called \emph{hypotheses}, to be added to the axioms of BKA\@.
Based on these hypotheses, one can then derive a sound model.
This facilitates a modular completeness proof of POCKA based on the completeness of BKA, as POCKA extends BKA with additional axioms.
\begin{definition}
A \emph{hypothesis} is an inequation $e \leq f$ where $e, f \in \termsbka$.
When $H$ is a set of hypotheses, we write $\equivbka^H$ for the smallest congruence on $\termsbka$ generated by the hypotheses in $H$ as well as the axioms and implications that build the equational theory of BKA\@.
More concretely, whenever $e \leq f \in H$, also $e \leqqbka^H f$.%
\end{definition}

\begin{definition}\label{def:closure}
Let $L \subseteq \Pom$.
We define the $H$-closure of $L$, written $\closure L$, as the smallest language containing $L$ such that for all $e \leq f \in H$ and $C \in \PC$, if $C[\sembka f]\subseteq \closure L$, then $C[\sembka e]\subseteq\closure L$.
We stress here the use of the BKA-semantics for defining the $H$-closure of any language.
Formally, $\closure L$ may be described as the smallest language satisfying:
\begin{mathpar}
    \inferrule{~}{%
        L\subseteq \closure L
    }
    \and
    \inferrule{%
      e\leq f \in H\\
      C\in\PC\\
      C[\sembka f]\subseteq \closure L
    }{%
      C[\sembka e]\subseteq \closure L
    }
\end{mathpar}
\end{definition}

The $H$-closure adds those pomsets that are needed to ensure soundness of the axioms generated by $H$. This yields a sound model for BKA with the set of hypotheses $H$~\cite{fossacs2020}:
\begin{restatable}[Soundness]{lem}{soundness}%
\label{lemma:soundness}
If $e \equivbka^H f$, then $\closure {\sembka{e}} = \closure {\sembka{f}}$.
\end{restatable}

An axiom often added to BKA is the \emph{exchange law}, and together with BKA it axiomatises Concurrent Kleene Algebra (CKA).
It can be added in the form of a set of hypotheses~\cite{fossacs2020}:
\begin{definition}
We write $\hexch$ for the set
$
    \{ (e \parallel f) \cdot (g \parallel h) \leq (e \cdot g) \parallel (f \cdot h) \pipe e, f, g, h \in \termsbka \}
$.
\end{definition}
These hypotheses encode the interleavings of a program: when $e\cdot g$ runs in parallel with $f\cdot h$, one possible behaviour is that $e$ first runs in parallel with $f$, followed by $g$ in parallel with $h$.

The $\hexch$-closure coincides with the downwards closure w.r.t.\ the subsumption order~\cite{fossacs2020}.

\begin{lemma}%
\label{lemma:exch-closure-vs-subsumption}
Let $L \subseteq \SP$ and $U \in \SP$.
$U\in\closure[\hexch]L$ $\Leftrightarrow$ there exists a $V \in L$ s.t.\ $U \sqsubseteq V$.
\end{lemma}

\begin{definition}%
\label{definition:reduction}
  A map $c\colon \termsbka \to \termsbka$ is a \emph{syntactic closure} for $H$ when for all $e \in \termsbka$ it holds that $e \equivbka^H c(e)$ and $\closure[H]{\sembka e}=\sembka {c(e)}$.
\end{definition}
Syntactic closures are used in modular constructions of completeness proofs: their existence implies a completeness result for $H$, by reducing it to completeness of \BKA, i.e. \Cref{theorem:bka-completeness}.

\section{Partially Observable Concurrent Kleene Algebra}%
\label{sec:pocka}

In this section we define \emph{partially observable concurrent Kleene algebra} (\emph{POCKA}).
The syntax of POCKA is given by BKA terms over an alphabet tailor-made to reason about programs that can access variables and values.
Specifically, this alphabet holds \emph{assignments} of the form $(v\leftarrow n)$ and $(v\leftarrow v')$, and \emph{observations} of the form $(v=n)$.
We say $(v\leftarrow n)$ assigns the value $n$ to variable $v$, $(v\leftarrow v')$ copies the value of variable $v'$ to $v$, and $(v=n)$ asserts that $v$ must have value $n$.
Formally, we define the alphabets
\begin{mathpar}
\Act=\{(v\leftarrow n), (v\leftarrow v') \pipe v,v'\in \var, n\in \val\}
\and
\Obs=\{(v=n) \pipe v\in \var, n\in \val\}
\end{mathpar}
where $\var$ and $\val$ are finite sets of variables and values, respectively (see \cref{rmk:finiteness} for a discussion on the finiteness assumption).
An example POCKA term would be $(x=1)\cdot(x\leftarrow 2)\cdot (x=2)$, which asserts that $x$ must start with value $1$, assigns the value $2$ to $x$, and then asserts that $x$ holds the value $2$.

We will later give semantics to POCKA terms using program states, which are partial functions from $\var$ to $\val$: $\State=\{\alpha \pipe \alpha\colon \var\rightharpoonup \val\}$.
The domain of a state $\alpha$ is denoted $\dom(\alpha)$. $\State$ carries a partial order $\leq$, where $\alpha \leq \beta$ iff $\dom(\beta)\subseteq \dom(\alpha)$ and for all $x \in \dom(\beta)$ we have $\alpha(x) = \beta(x)$, which we will use to generate the algebra of observations.

\subsection{Observation algebra: axiomatisation and semantics}
To obtain POCKA, we define the \thealgebra (\acronym) that will be added to CKA as the algebraic structure of observations.
This is similar to how a Boolean algebra enriches Kleene algebra into Kleene algebra with tests.
In contrast with KAT, the observation algebra of POCKA is a pseudocomplemented distributive lattice, which is a generalisation of Boolean algebra in which the law of excluded middle does not necessarily hold.

\begin{table*}[t]\small\centering\noindent%
  \fbox{\noindent%
    \begin{minipage}[t]{.46\textwidth}
      \noindent%
      \textbf{Kleene Algebra Axioms}
      \vspace{-.8em}
      \begin{align*}
        e + (f + g) & \equivbka  (f + g) + h \\
        e + f  & \equivbka  f + e   \\
        e + 0                             & \equivbka  e\\
        e + e                             & \equivbka  e\\
        e \cdot (f \cdot g)               & \equivbka  (e \cdot f) \cdot g     \\
        e \cdot 1                         & \equivbka  e \equivbka  1 \cdot e
        \\
        e \cdot 0 & \equivbka  0  \equivbka  0 \cdot e\\
        e \cdot (f + g) & \equivbka  e \cdot f + e \cdot h \\
        (e + f) \cdot g & \equivbka  e \cdot g + f \cdot g\\
        e^* &\equiv 1 + ee^* \\
        e + f \cdot g \leqqbka f & \Rightarrow  e \cdot g^* \leqqbka f\\
        e^* &\equiv 1 + e^*e \\
        e + f \cdot g \leqqbka g  &\Rightarrow  f^* \cdot e \leqqbka g
      \end{align*}
      \vspace{-1.5em}
      \hrule
      \vspace{.5em}
      \noindent%
      \textbf{Additional Bi-Kleene Algebra Axioms}
      \vspace{-.8em}
      \begin{align*}
        e \parallel 1 & \equivbka  e  \\
        e \parallel (f \parallel g) & \equivbka   (e \parallel f) \parallel g \\
        e \parallel 0 & \equivbka  0 \\
        e \parallel (f + g) & \equivbka  e \parallel f + e \parallel g \\
        e \parallel f & \equivbka  f \parallel e
      \end{align*}
      \vspace{-1.5em}
      \hrule
      \vspace{.5em}
      \noindent%
      \textbf{Exchange law}
      \vspace{-.8em}
      \begin{equation*}
        (e \parallel f) \cdot (g \parallel h) \leqqpocka (e \cdot g) \parallel (f \cdot h)
      \end{equation*}
    \end{minipage}\hspace{.01\linewidth}\vrule\hspace{.01\linewidth}%
    \begin{minipage}[t]{.49\linewidth}
      \noindent%
      \textbf{Bounded Distributive Lattice Axioms}
      \vspace{-.8em}
      \begin{align*}
        p \vee \bot  &\equivpl  p \equivpl
                       p \wedge \top \\
        p \vee q & \equivpl  q \vee p \\
        p \wedge q & \equivpl  q \wedge p \\
        p \wedge (q \wedge r) & \equivpl  (p \wedge q) \wedge r \\
        p \vee (q \vee r) & \equivpl  (p \vee q) \vee r \\
        p \vee (p \wedge q) & \equivpl  p  \equivpl
                              p \wedge (p \vee q) \\
        p \vee (q \wedge r) & \equivpl  (p \vee q) \wedge (p \vee r) \\
        p \wedge (q \vee r) & \equivpl  (p \wedge q) \vee (p \wedge r)
      \end{align*}
      \vspace{-1.5em}
      \hrule
      \vspace{.5em}
      \noindent%
      \textbf{Pseudocomplement}
      \vspace{-.8em}
      \begin{equation*}
        p\leqqpl \overline{q} \Leftrightarrow p\wedge q\equivpl \bot
      \end{equation*}
      \vspace{-1.5em}
      \hrule
      \vspace{.5em}
      \noindent%
      \textbf{Observation Axioms}
      \vspace{-.8em}
      \begin{align*}
        v = n &\wedge v = m \equivpl  \bot \tag{$n\neq m$} \\ 
        \overline{v=n} & \leqqpl  \bigvee\limits_{n\neq m}v=m \\
        \overline{\bigwedge_{i} v_i=n_i} &\leqqpl \bigvee_{i} \overline{v_i=n_i} \tag{$\forall i\neq j. v_i \neq v_j$} 
      \end{align*}
      \vspace{-1em}
      \hrule
      \vspace{.5em}
      \noindent%
      \textbf{Interface Axioms}
      \vspace{-.8em}
      \begin{align*}
        p \wedge q & \leqqpocka  p \cdot q \\
        p \vee q & \equivpocka  p + q \\
        0 & \equivpocka  \bot \\
        \top \cdot p & \leqqpocka   p &p \cdot \top &\leqqpocka  p\tag{$p\in\termspl$}  \\
        \top \cdot a & \leqqpocka   a &a \cdot \top &\leqqpocka a\tag{$a\in\Act$}
      \end{align*}
    \end{minipage}
  }
  \caption{%
    Axioms of POCKA, built over an alphabet of actions $\Act$ and observations $\Obs$.
    The left column contains the axioms of Concurrent Kleene Algebra.
    The right column axiomatises the partial observations: they form a pseudocomplemented distributive lattice, subject to constraints on the interface axioms that connect the lattice operators to the Kleene algebra ones.
    The last group of axioms applies to the observation alphabet $\Obs$. We write $ e\leqqpl f $ as a shorthand for $e+f \equivpl f$.
  }%
  \label{fig:pocka}
\end{table*}
\begin{definition}[Pseudocomplemented Distributive Lattice]
A \emph{pseudocomplemented distributive lattice} (PCDL) is a tuple $(A, \land, \lor, \pnegop, \top, \bot)$ such that $(A, \land, \lor, \top, \bot)$ is a bounded distributive lattice and $\pnegop\colon A \to A$ is such that for $p, q \in A$ we have $p \land q = \bot$ iff $p \leq \overline{q}$.
\end{definition}

For a poset $(X, \leq)$ and a set $S \subseteq X$, define the \emph{downwards-closure} of $S$ by $S_{\leq} ::= \{x \pipe \exists y\in S \text{ s.t. } x\leq y\}$ and $P_{\leq}(X) ::= \{ Y \subseteq X \mid Y = Y_{\leq} \}$.
It is well-known that $P_{\leq}(X)$ carries the structure of a bounded distributive lattice, with intersection as meet, union as join, $X$ as top and $\emptyset$ as bottom.
Further, if $(X, \leq)$ is finite, the lattice is itself finite and thus carries a (necessarily unique) pseudocomplement defined by $\overline{Y} ::= \bigcup \{ Z \in P_{\leq}(X) \mid Y \cap Z = \emptyset\}$.
This simply reifies that the pseudocomplement of an element is the largest element incompatible with it, which is guaranteed to exist in any complete lattice with bottom.

\begin{definition}[Observation Algebra]
The \emph{Observation Algebra} is the PCDL $\mathrm{\acronym} ::= (P_{\leq}(\State), \cap, \cup, \pnegop, \State, \emptyset)$ generated by $(\State, \leq)$.
\end{definition}

Taking $\Obs$ as our set of propositions, we generate a term language $\termspl$ over the signature of PCDLs as follows:
\[ p, q ::= \bot \pipe \top \pipe o \in \Obs \pipe p \vee q \pipe p \wedge q \pipe \overline{p}. \]

\noindent This language is interpreted in $\acronym$ by the homomorphic extension of the assignment
 \[ \sempl{v=n} ::= {\{ \{ v \mapsto n \} \}}_{\leq} = \{ \alpha \in \State \mid \alpha(v) = n \}. \]

Intuitively, the behaviour of an observation $p$ consists of all partial functions that agree with $p$. This is captured algebraically below, in \cref{axiomaticsum}.
For instance, $\sempl{v=n}$ is the set containing all partial functions assigning $n$ to $v$, and this is downwards closed because any partial function with a larger domain that also assigns $n$ to $v$ is included in this set.

If threads have only partial information about the machine state, an observation should
be satisfied only if there is \emph{positive evidence} for it. Hence, $\overline{v = n}$ should be satisfied only
when $v$ has a value that is different from $n$. To see why a Boolean algebra does not capture
the intended meaning of $\overline{v = n}$, consider using a BA over sets of partial functions, with
negation as set-complement. This entails that $\sempl{\overline{v=n}}$ will
include all partial functions where $v$ either gets a different value than $n$ or no value at all. In
the latter case, were we to obtain more information about the machine state we may discover
that the actual value of $v$ is in fact $n$, and it was therefore incorrect to assert $\overline{v = n}$. Our
pseudocomplement provides a notion of negation that correctly excludes states for which this
error could manifest, and this motivates our use of a PCDL rather than a Boolean algebra.
This can be calculated directly:

\begin{example}
Consider the semantics of $\overline{v=n}$:
\begin{align*}
\sempl{\overline{v=n}} &= \bigcup \{ Z \in P_{\leq}(\State) \mid \sempl{v=n} \cap Z = \emptyset\} \\
& = \{\alpha \pipe \alpha\in Z \text{ and } \{\beta \pipe \beta(v)=n\}\cap Z =\emptyset \} \\
& = \{\alpha \pipe \alpha(v)=m \text{ and } m\neq n\}
\end{align*}
In the last step we use that $Z$ is downwards closed: if $\alpha(v)$ were undefined, then the partial function $\alpha'$ which is the same as $\alpha$ except that $\alpha'(v)=n$ would also occur in $Z$, making the intersection with $\sempl{v=n}$ non-empty.
Thus $\alpha \in \sempl{\overline{v=n}}$ only if $\alpha(v)$ is defined, and evaluates to a value distinct from $n$. This witnesses the failure of the law of excluded middle.
\end{example}

\begin{definition}[Axiomatisation]
$\equivplexpl$, or simply $\equivpl$, is the smallest congruence on $\termspl$ generated by the distributive lattice, pseudocomplement and observation axioms in \cref{fig:pocka}.
\end{definition}

This axiomatisation supplements a standard axiomatisation of PCDLs with domain-specific axioms to capture the propositional theory of observations.
For instance, the axiom $(v = n \wedge v = m \equivpl \bot)$ states that a variable cannot have two different values at the
same time.
The axiom $(\overline{v=n}\leqqpl \bigvee_{n\neq m}v=m)$ tells us that the pseudocomplement of a variable having a value $n$ is the assertion that the variable holds \emph{some} distinct value $m$ (the axiom is an inequality, but the other way around also holds).
The last domain-specific axiom enforces specific instances of a De Morgan law that does not hold generally hold in arbitrary PCDLs\@.

Soundness of this axiomatisation follows straightforwardly from the fact that $\mathrm{\acronym}$ is a PCDL, together with basic consequences of the definition of the poset $(\State, \leq)$.

\begin{restatable}[Soundness OA]{lem}{soundnesspcdl}%
  \label{soundnesspcdl}
  For all $p,q\in\termspl$, if $p\equivpl q$ then $\sempl{p}=\sempl{q}$.
\end{restatable}

Let $\pi_{\alpha} ::= \bigwedge_{\alpha(v)=n}v=n$.
Note that if $\alpha$ is the empty function, then $\pi_{\alpha}=\bigwedge \emptyset = \top$.
\begin{restatable}{lem}{lemmabasicfacts}\label{lemma:basicfacts}
For all $\alpha,\beta\in \State$: $\alpha \in \sempl{\pi_\beta}$ iff $\alpha \leq \beta$ iff $\pi_\alpha \leqqpl \pi_{\beta}$.
\end{restatable}
In the following sections, we will silently assume that $\State\subseteq\termspl$.
This is possible because $\pi_{-}$ provides us with a sound way of injecting $\State$ inside $\termspl$.
In order to prove completeness for $\sempl{-}$ w.r.t. $\equivpl$, we need an intermediary result, which allows us to syntactically rewrite any \acronym-expression in terms of elements of $\State$.
\begin{restatable}{lem}{axiomaticsum}\label{axiomaticsum}
For all $p\in\termspl$, we have $p \equivpl \bigvee \{ \alpha \in \State \pipe \pi_{\alpha}\leqqpl p\}$.
\end{restatable}

With this result, we can then prove completeness of $\sempl{-}$ w.r.t.\ $\equivpl$ on terms from $\termspl$. In short,
from $\sempl{p} = \sempl{q}$, $\pi_{\alpha} \leqqpl p$ iff $\pi_{\alpha} \leqqpl q$ can be established, from which $p \equivpl q$ follows.

\begin{restatable}[Completeness OA]{thrm}{soundnesscompletenesspcdl}%
  \label{soundnesscompletenesspcdl}
  For all $p,q\in\termspl$, we have $p\equivpl q$ if and only if $\sempl{p}=\sempl{q}$.
\end{restatable}

\subsection{POCKA:\texorpdfstring{\@}{} axiomatisation and semantics}%
\label{section:syntaxpocka}
\begin{definition}
The POCKA-terms, denoted $\termspocka$, are formed by the following grammar:
\[
    e, f ::= 0 \pipe 1 \pipe a \in \Act \pipe p \in \termspl \pipe e + f \pipe e \cdot f \pipe e \parallel f \pipe e^* \ .
\]
\end{definition}
Note that $\termspocka=\termsbka[\Act\cup\termspl]$.

The language model for POCKA consists of pomset languages over $\Act\cup\State$.
When using pomsets to reason about behaviours of programs, we would like actions and states to alternate, because the states allow one
to take stock of the configuration of the machine in between actions.
However, imposing such an alternation
in the semantics can be problematic with the exchange law~\cite{kappe-brunet-rot-silva-wagemaker-zanasi-2019}.
Imagine the program $(\alpha\cdot\beta)\parallel \ltr{a}$, where $\alpha,\beta\in\State$ and $\ltr{a}\in\Act$. We can derive the following:
\begin{align*}
 \alpha\cdot\ltr{a}\cdot\beta & \equivpocka (\alpha\parallel 1)\cdot (1\parallel\ltr{a})\cdot (\beta \parallel 1) \tag{Unit axiom} \\
 & \leqq (\alpha\parallel 1)\cdot ((1\cdot \beta) \parallel (\ltr{a}\cdot 1)) \tag{Exchange Law} \\
 & \equivpocka (\alpha\parallel 1)\cdot (\beta \parallel \ltr{a}) \tag{Unit Axiom} \\
 & \leqq (\alpha\cdot \beta) \parallel (1\cdot \ltr{a}) \equivpocka (\alpha\cdot\beta)\parallel \ltr{a} \tag{Exchange Law, Unit Axiom}
\end{align*}
However, if the semantics $\semwpocka{-}$ contain only pomsets with alternating assignments and states, $\semwpocka{\alpha\cdot\beta}$ would have to consist of one state, and hence be empty if $\alpha\neq\beta$.
This would make $\semwpocka{(\alpha\cdot\beta)\parallel \ltr{a}}$ empty as well.
As $\semwpocka{\alpha\cdot\ltr{a}\cdot\beta}$ should not be empty, the exchange law is unsound.
Thus, the POCKA-semantics is not restricted to pomsets with alternating states and actions.

\begin{definition}[Semantics]
Let $\semwpocka{-}\colon\termspocka \to 2^{\SP}$, where $2^{\SP}$ are pomset languages over $\Act\cup\State$.
 For $p\in\termspl$, $(v\leftarrow n),(v\leftarrow v')\in\Act$ and
 $e,f\in\termspocka$ we have:
\begin{align*}
\semwpocka{v\leftarrow n} &= \State^*\cdot \{v\leftarrow n \}\cdot \State^*
    & \semwpocka{e + f} &= \semwpocka{e} \cup \semwpocka{f}
    & \semwpocka{e^*} &= {\semwpocka{e}}^* \\
\semwpocka{v\leftarrow v'} &= \State^*\cdot \{v\leftarrow v' \}\cdot \State^*
    & \semwpocka{e \cdot f} &= \semwpocka{e} \cdot \semwpocka{f}
    & \semwpocka{0} &= \emptyset \\
\semwpocka{p} &= \State^*\cdot \semplexpl{p}\cdot \State^*
    & \semwpocka{e \parallel f} &= \semwpocka{e} \parallel \semwpocka{f}
    & \semwpocka{1} &= \{ 1 \}
\end{align*}
We define the POCKA-semantics of $e \in \termspocka$ as $\sempocka{e}= \closure[\hexch\cup\hcontr] {\semwpocka{e}}$, where we use the closure definition from \cref{def:closure}, and $\hcontr=\{ \alpha \leq \alpha \cdot \alpha \pipe \alpha\in\State \}$,
referred to as \emph{contraction}.
\end{definition}

We briefly explain closure under $\hexch$ and $\hcontr$. These closures are not part of the axiomatisation, but exist to ensure soundness of some of the axioms.
The set of hypotheses $\hexch$ closes the POCKA-semantics under subsumption and ensures soundness for the exchange law familiar from CKA\@.
The set $\hcontr$ encodes that
one way of observing $\alpha$ twice is to make both observations on the same state.
This provides soundness for the axiom $p\wedge q\leqqpocka p\cdot q$, which was introduced in~\cite{kappe-brunet-rot-silva-wagemaker-zanasi-2019}.
This axiom captures that if $p$ and $q$ hold simultaneously in some state, it is possible to observe $p$ and $q$ in sequence (the converse should not hold as some action could happen in between the two obervations in a parallel thread).

\begin{remark}
 The assignment $(v\leftarrow v')$ cannot be simulated.
 In a sequential setting, we could express $(v\leftarrow v')$ as $\sum_{n\in \val}((v'=n)\cdot(v\leftarrow n))$.
 However, in a parallel setting this does not work, since some action can change the value of $v'$ in between the observation that $(v'=n)$ and the assignment $(v\leftarrow n)$, meaning that $v$ does not get assigned the value of $v'$.
\end{remark}

\begin{remark}\label{rmk:finiteness}
  In this paper we assume the set of variables $\var$ and the set of values $\val$ are both finite, in keeping with other verification frameworks, e.g.\ in model-checking.

  The restriction on $\var$ could be lifted, since the finite set of variables that appear syntactically in a term completely determine its semantics. However, this is not the case for the set of values: for instance, the term $\overline{v=0}$ evaluates to $\emptyset$ if the set of values is $\val=\{0\}$, but contains the partial function $[v\mapsto n]$ if $\val$ contains some value $n\neq 0$.

  It is possible that a more sophisticated reduction could still work: indeed it seems unlikely that our finite terms would be able to manipulate non-trivially infinitely many values. For now though, this question is left open for future investigations.
\end{remark}

The POCKA-semantics of a program $e$ contains the possible behaviours of $e$ in any possible context, where the context refers to any expression that could be put in parallel with $e$.
For instance,  $\sempocka{(v\leftarrow n)}$ contains pomsets that consist of a string of possible states of the machine, where the state of the machine can have been influenced by other parallel threads, followed at some point by the assignment $(v\leftarrow n)$, followed by another string of states.
In \cref{section:causal}, we will show how to reason about programs in isolation, i.e., under the hypothesis that there is no outside context to prompt state-modifying actions.
 \begin{example}\label{example:semantics-litmus}
 Let $t = (r_0=0\wedge r_1=0) \cdot (\textsf{T0} \parallel \textsf{T1}) \cdot \overline{(r_0=1\vee r_1=1)}$ as in~\eqref{equation:sequential-consistency} be our litmus test.
A pomset in $\semwpocka{t}$ may look as follows, where we depict a pomset graphically with nodes labelled by actions or observations and their ordering with arrows.

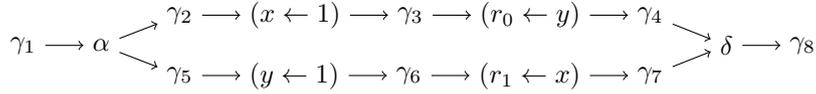
\begin{tikzpicture}[node distance=0.5cm]
    \begin{scope}[every node/.style={anchor=center}]
    \node (gamma1) {$\gamma_1$};
    \node[right=of gamma1] (alpha) {$\alpha$};
    \node[yshift=4mm,right=of alpha] (gamma2) {$\gamma_2$};
    \node[yshift=-4mm,right=of alpha] (gamma5) {$\gamma_5$};
    \node[right=of gamma2] (x) {$(x\leftarrow 1)$};
    \node[right=of gamma5] (y) {$(y\leftarrow 1)$};
    \node[right=of x] (gamma3) {$\gamma_3$};
    \node[right=of y] (gamma6) {$\gamma_6$};
    \node[right=of gamma3] (r0) {$(r_0\leftarrow y)$};
    \node[right=of gamma6] (r1) {$(r_1\leftarrow x)$};
    \node[right=of r0] (gamma4) {$\gamma_4$};
    \node[right=of r1] (gamma7) {$\gamma_7$};
    \node[yshift=-4mm, right=of gamma4] (delta) {$\delta$};
    \node[right=of delta] (gamma8) {$\gamma_8$};
    \end{scope}
    \draw[->] (gamma1) edge (alpha);
    \draw[->] (alpha) edge (gamma2);
    \draw[->] (alpha) edge (gamma5);
    \draw[->] (gamma2) edge (x);
    \draw[->] (gamma5) edge (y);
    \draw[->] (x) edge (gamma3);
    \draw[->] (y) edge (gamma6);
    \draw[->] (gamma3) edge (r0);
    \draw[->] (gamma6) edge (r1);
    \draw[->] (r0) edge (gamma4);
    \draw[->] (r1) edge (gamma7);
    \draw[->] (gamma4) edge (delta);
    \draw[->] (gamma7) edge (delta);
    \draw[->] (delta) edge (gamma8);
\end{tikzpicture}

\noindent Here, $\gamma_i\in\State$, $\alpha(r_0)=0=\alpha(r_1)$ and $\delta(r_0)=0=\delta(r_1)$.
However, as stated in the introduction, if POCKA is sequentially consistent, this litmus test should pass,
which means that the semantics of
$t$
should instead be empty. The reason it is not empty is that our semantics gives the behaviour of a program in any possible context, and indeed, if we put the litmus test in parallel with a program such as $(r_0\leftarrow 0)\cdot (r_1\leftarrow 0)$, the final assertion becomes satisfiable.
In \cref{section:causal,section:litmus}, we look at how to execute the litmus test in isolation.
\end{example}

We also have axioms to algebraically describe equivalence between POCKA-terms, including some domain-specific axioms tailored to the alphabet.
We define $\equivpocka$ as the smallest congruence on $\termspocka$ generated by the axioms in \cref{fig:pocka}.

\begin{restatable}[Soundness POCKA]{thrm}{soundnesspocka}\label{soundnesspocka}
  For all $e,f\in\termspocka$, if $e\equivpocka f$ then $\sempocka{e}=\sempocka{f}$.
\end{restatable}

\section{Completeness}\label{sec:complete}
In this section we prove completeness of the POCKA-semantics w.r.t.\ the axioms provided in \cref{section:syntaxpocka}.
First, we show that POCKA terms can be normalised to a simpler form, where the only observations that appear are states.
Next, we show that the resulting POCKA-terms can be used to describe the POCKA-semantics using BKA-semantics and closure.
We then use the techniques from~\cite{fossacs2020} to obtain a completeness result with respect to this semantics.
Finally, we put all of the above together to obtain a completeness result for POCKA proper.

In order to normalise POCKA terms, we replace every observation by the summation of states to which it corresponds.
This is done using the homomorphism $\hat{r}$ generated by
\[
    r(a) =
    \begin{cases}
      \sum_{\alpha \leqqplexpl a} \alpha & a \in \termspl \\
      a & a \in \Act
    \end{cases}
\]
As a straightforward consequence of \cref{axiomaticsum} and the interface axioms, we then obtain:
\begin{lemma}\label{lemma:reification}
For all $e \in \termspocka$, it holds that $e \equiv \hat{r}(e)$.
\end{lemma}
\begin{proof}
This can be proven by induction on the structure of $e$. If $e=\ltr{a} \in \Act$, then $\hat r(\ltr{a}) = \ltr{a} \equiv \ltr{a}$ immediately.
Otherwise, if $p \in \termspl$, then we derive $\hat r(p) = \sum_{\alpha \leqqplexpl p} \alpha \equiv \bigvee_{\alpha \leqqplexpl p} \alpha \equiv p$, where we apply \cref{axiomaticsum} in the last step.
The inductive step follows trivially.
\end{proof}

The effect of $\hat r$ is to bridge the gap between the semantics of BKA and that of observation algebra: indeed for an observation $p\in\termspl$, we have $\sembka{\hat r (p)}=\semplexpl{p}$.
However, this does not bring us fully to the unclosed POCKA-semantics $\semwpocka{-}$, as the latter inserts state-nodes in between actions and observations.
For instance, $\semwpocka{v \leftarrow n}$ includes pomsets like $\alpha \cdot (v \leftarrow n) \cdot \beta$, while $\sembka{v \leftarrow n} = \{ v \leftarrow n \}$.
We cover for this by means of the following set of hypotheses:
\[
    \htop = \{ \alpha \cdot c \leq c, c \cdot \alpha \leq c \mid \alpha \in \State, c \in \Act \cup \State \}
\]

The hypotheses in $\htop$ allow us to connect the unclosed POCKA semantics to the BKA-semantics, by filling in surrounding or preceding state-labelled nodes as necessary.
\begin{restatable}{lem}{lemmaunclosedvstop}\label{lemma:unclosed-vs-top}
  For all $e\in\termspocka$, we have $\semwpocka{e}=\closure[\htop]{\sembka{\hat{r}(e)}}$.
\end{restatable}
\begin{proof}[Sketch]
Proceed by induction on the construction of $e$; the case where $e \in \Act \cup \termspl$ is fairly straightforward.
The inductive case follows from the fact that closure w.r.t. $\htop$ is compatible with pomset language composition, i.e., that $\closurep[\htop]{L \cup K} = \closure[\htop]{L} \cup \closure[\htop]{K}$ as well as $\closurep[\htop]{L \cdot K} = \closure[\htop]{L} \cdot \closure[\htop]{K}$, and similarly for the other operators defining $\sembka{-}$.
\end{proof}

The next step is to provide syntactic closures for the sets of hypotheses involved.
First, we note that there exists a syntactic closure for $\hexch \cup \hcontr$, as shown in~\cite[Theorem~5.6]{fossacs2020}.
\begin{lemma}\label{lemma:completeness-exch-contr}
There exists a syntactic closure $k$ for $\hexch \cup \hcontr$.
\end{lemma}

For the set $\htop$ we still need to provide a syntactic closure.
To this end, we simply take every action or observation in a term and surround it by a sequence of states of arbitrary length, which can be done using the homomorphism generated by
\[
  s(a) = {\Bigl(\sum_{\alpha\in\State} \alpha\Bigr)}^* \cdot a  \cdot {\Bigl(\sum_{\alpha\in\State} \alpha\Bigr)}^*
\]
It is fairly straightforward to show that this gives rise to a syntactic closure.
\begin{restatable}{lem}{completenessalt}\label{lemma:completeness-alt}
  The homomorphism generated by $s$ is a syntactic closure for $\htop$.
\end{restatable}
\begin{proof}[Sketch]
The proof proceeds by induction on the construction of a term $e$.
In the base, we can show for $a \in \Act \cup \termspl$ that $s(a) \equiv^\htop a$ and $\closure[\htop]{\sembka{a}} = \sembka{s(a)}$.
The inductive step follows by an argument similar to that in \cref{lemma:unclosed-vs-top}.
\end{proof}

The final step needed for the completeness proof of POCKA is a relation between the axioms that generate $\equiv$ and the hypotheses found in $\htop$ and $\hcontr$.
\begin{restatable}{lem}{implicationaxioms}\label{implicationaxioms}
  For all $e,f \in \termspocka$, if $e\leq f \in \htop\cup\hcontr$, then $e\leqq f$.
\end{restatable}

We now have all the ingredients in place for the desired completeness proof.

\begin{restatable}{thrm}{theoremcompletenessthree}\label{theorem:completeness-three-parts}
For all $e,f \in \termspocka$, we have $e \equiv f$ if and only if $\sempocka{e} = \sempocka{f}$.
\end{restatable}
\begin{proof}
  The direction from left to right was already established in \cref{soundnesspocka}.
  For the other direction, suppose that $e, f \in \termspocka$ such that $\sempocka{e} = \sempocka{f}$.
  We can then derive that
  \begin{align*}
    \sempocka{e}
      &= \closure[\hexch \cup \hcontr]{\semwpocka{e}}
          \tag{def. $\sempocka{-}$} \\
      &= \closurep[\hexch \cup \hcontr]{\closure[\htop]{\sembka{\hat{r}(e)}}}
          \tag{\cref{lemma:unclosed-vs-top}} \\
      &= \closure[\hexch \cup \hcontr]{\sembka{\hat{s} \circ \hat{r}(e)}}
          \tag{\cref{lemma:completeness-alt}} \\
      &= \sembka{k \circ \hat{s} \circ \hat{r}(e)}
          \tag{\cref{lemma:completeness-exch-contr}}
  \intertext{%
  Similarly, $\sempocka{f} = \sembka{k \circ \hat{s} \circ \hat{r}(e)}$.
  Since $\sempocka{e} = \sempocka{f}$, also $\sembka{k \circ \hat{s} \circ \hat{r}(e)} = \sembka{k \circ \hat{s} \circ \hat{r}(f)}$; by \cref{theorem:bka-completeness}, it then follows that $k \circ \hat{s} \circ \hat{r}(e) \equivbkaexpl k \circ \hat{s} \circ \hat{r}(e)$.
  We then derive that
  }
    e
      &\equiv \hat{r}(e)
        \tag{\cref{lemma:reification}} \\
      &\equivbkaexpl^\htop \hat{s} \circ \hat{r}(e)
        \tag{\cref{lemma:completeness-alt}} \\
      &\equivbkaexpl^{\hexch \cup \hcontr} k \circ \hat{s} \circ \hat{r}(e)
        \tag{\cref{lemma:completeness-exch-contr}} \\
      &\equivbkaexpl k \circ \hat{s} \circ \hat{r}(f)
        \tag{Observation above} \\
      &\equivbkaexpl^{\hexch \cup \hcontr} \hat{s} \circ \hat{r}(f)
        \tag{\cref{lemma:completeness-exch-contr}} \\
      &\equivbkaexpl^\htop \hat{r}(f)
        \tag{\cref{lemma:completeness-alt}} \\
      &\equiv f
        \tag{\cref{lemma:reification}}
  \end{align*}
  By~\cref{implicationaxioms}, $\equivbkaexpl^\htop$ and $\equivbkaexpl^{\hexch \cup \hcontr}$ are contained in $\equiv$; we conclude that $e \equiv f$.
\end{proof}

\section{Guarded Pomsets}\label{section:causal}
We now identify a fragment of the semantics that we use in the analysis of the litmus test from~\eqref{equation:sequential-consistency}, namely the \emph{guarded} pomsets. This term comes from Jipsen and Moshier~\cite{jipsen-moshier-2016}, and was meant to define guarded pomsets in analogy to guarded strings in KAT\@.

We need two pieces of notation. First, we define the result of a state
after updating it for one value. Let $\ltr{a}\in\Act$ and $\alpha \in \State$.
We say
that $\alpha[\ltr{a}]$ \emph{exists} if $\ltr{a}=v\leftarrow n$
for some $n \in \val$ or $\ltr{a} = v\leftarrow v'$ and $v' \in \dom(\alpha)$.
If $\alpha[\ltr{a}]$ exists, we define it for all $w\in \var$ as follows:
\begin{mathpar}
\alpha[v\leftarrow n](w) =
\begin{cases}
n & \text{if } w=v \\
\alpha(w) & \text{otherwise}
\end{cases}
\and
\alpha[v\leftarrow v'](w) =
\begin{cases}
\alpha(v') & \text{if } w=v \\
\alpha(w) & \text{otherwise}
\end{cases}
\end{mathpar}
\noindent Second, we define a binary operator $\oplus$ on $\State$ to combine states. For $\alpha,\beta\in \State$:
\[
\alpha\oplus\beta =
\begin{cases}
  \alpha\cup\beta & \text{ if $\alpha(v)=\beta(v)$ for all }v \in \dom(\alpha) \cap \dom(\beta)\\
  \text{undefined} & \text{otherwise}
\end{cases}
\]

\begin{definition}
The set of \emph{guarded pomsets}, denoted $\G$, is the smallest set satisfying:
\begin{mathpar}
\inferrule{%
    \alpha \in \State
}{%
    \alpha \in \G
}
\and
\inferrule{%
    \alpha \in \State \\
    \ltr{a} \in \Act \\
    \alpha[\ltr{a}] \text{ exists}
}{%
    \alpha \cdot \ltr{a} \cdot \alpha[\ltr{a}] \in \G
}
\and
\inferrule{%
    U \cdot \alpha, \alpha \cdot V \in \G \\
    \alpha \in \State
}{%
    U \cdot \alpha \cdot V \in \G
}
\and
\inferrule{%
    \alpha \cdot U \cdot \beta \\
    \gamma \cdot V \cdot \delta \in \G \\
    \alpha\oplus\gamma \text{ defined } \\
    \beta\oplus\delta \text{ defined } \\
    \alpha, \beta, \gamma, \delta \in \State
}{%
    \alpha\oplus\gamma \cdot (U \parallel V) \cdot \beta\oplus\delta \in \G
}
\end{mathpar}
\end{definition}

This definition is close to~\cite{jipsen-moshier-2016}.
The definition in op.\ cit.\ is not catered to a specific alphabet, and the operator $\oplus$ to combine states does not allow for
the two states to have any shared variable in their domains.
We deliberately deviate from this, allowing threads to share variables as long as they do so in a consistent manner.

\begin{example}
  Consider the guarded pomsets $(x=1) \cdot(x \leftarrow 2) \cdot (x=2)$ and $(x=1 \wedge y=3) \cdot (x \leftarrow 2) \cdot (x=2 \wedge y=3) \cdot (y \leftarrow x) \cdot (x=2 \wedge y=2)$.
  The final rule for the construction of $\G$ guarantees that the following pomset is again guarded:
  \[
  \begin{tikzpicture}[node distance=0.5cm]
    \begin{scope}[every node/.style={anchor=center}]
    \node (n1) {$(x=1 \wedge y=3)$};
    \node[yshift=-4mm,right=of n1] (n3) {$(x \leftarrow 2)$};
    \node[right=of n3] (n4) {$(x=2 \wedge y=3)$};
    \node[right=of n4] (n5) {$(y \leftarrow x)$};
    \node[right=of n5,yshift=4mm] (n6) {$(x=2 \wedge y=2)$};

    \node[above=4mm of n4] (n2) {$(x \leftarrow 2)$};
    \end{scope}
    \path ($(n1.north east) + (0,-2mm)$) edge[->] (n2.west);
    \path ($(n1.south east) + (0,2mm)$) edge[->] (n3.west);
    \path (n2.east) edge[->] ($(n6.north west) + (0,-2mm)$);
    \path (n3) edge[->] (n4);
    \path (n4) edge[->] (n5);
    \path (n5) edge[->] ($(n6.south west) + (0,2mm)$);
  \end{tikzpicture}
  \]
  Note how $(x = 2) \oplus (x=2 \wedge y=3)$ is defined, because both states agree on the value of $x$.
\end{example}

In the execution of parallel threads in pomsets, no interaction between the threads takes place: the threads execute ``truly'' concurrently.
To account for interactions, we consider the interleavings that result from closure w.r.t.\ the exchange law (c.f.\ \cref{lemma:exch-closure-vs-subsumption}).

\begin{example}
Consider the a slightly adjusted version of the litmus test $t$ discussed earlier:
\[
    t' = (r_0=0 \wedge r_1=0) ; (\textsf{T0} \parallel \textsf{T1}) ; (r_0=1 \vee r_1=1)
\]
The \emph{unclosed} semantics of $t'$ includes (but is not limited to) the pomset below on the left, for all $\alpha, \beta, \gamma \in \State$, where $\alpha(r_0) = \alpha(r_1) = 0$, and $\gamma(r_0) = 1$ or $\gamma(r_1) = 1$.
As a result of the exchange law, the \emph{closed} semantics includes the pomset below on the right.
\begin{mathpar}
  \begin{tikzpicture}[node distance=0.5cm]
    \begin{scope}[every node/.style={anchor=center}]
      \node (atom1) {$\alpha$};
      \node[right=of atom1] (write1) {$(x \leftarrow 1)$};
      \node[below=1mm of write1] (write2) {$(y \leftarrow 1)$};
      \node[right=of write1] (atom3) {$\beta$};
      \node[right=of atom3] (copy1) {$(r_0 \leftarrow y)$};
      \node[below=1mm of copy1] (copy2) {$(r_1 \leftarrow x)$};
      \node[right=of copy1] (atom7) {$\gamma$};
    \end{scope}
    \path (atom1) edge[->] (write1);
    \path (atom1) edge[->] (write2.west);
    \path (write1) edge[->] (atom3);
    \path (write2) edge[->] (copy2);
    \path (atom3) edge[->] (copy1);
    \path (copy1) edge[->] (atom7);
    \path (copy2.east) edge[->] (atom7);
  \end{tikzpicture}
  \and
  \begin{tikzpicture}[node distance=0.5cm]
    \begin{scope}[every node/.style={anchor=center}]
      \node (atom1) {$\alpha$};
      \node[right=of atom1] (write1) {$(x \leftarrow 1)$};
      \node[below=1mm of write1] (write2) {$(y \leftarrow 1)$};
      \node[right=of write1] (atom3) {$\beta$};
      \node[right=of atom3] (copy1) {$(r_0 \leftarrow y)$};
      \node[below=1mm of copy1] (copy2) {$(r_1 \leftarrow x)$};
      \node[right=of copy1] (atom7) {$\gamma$};
    \end{scope}
    \path (atom1) edge[->] (write1);
    \path (atom1) edge[->] (write2.west);
    \path (write1) edge[->] (atom3);
    \path (write2.east) edge[->] (atom3);
    \path (atom3) edge[->] (copy2.west);
    \path (atom3) edge[->] (copy1);
    \path (copy1) edge[->] (atom7);
    \path (copy2.east) edge[->] (atom7);
  \end{tikzpicture}
\end{mathpar}
In the special case where $\alpha = \{ r_0 \mapsto 0, r_1 \mapsto 0 \}$, $\beta = \{ r_0 \mapsto 0, r_1 \mapsto 0, x \mapsto 1, y \mapsto 1 \}$, $\gamma = \{ r_0 \mapsto 1, r_1 \mapsto 1, x \mapsto 1, y \mapsto 1 \}$, the latter is a guarded pomset.
\end{example}

Guardedness in pomsets can be characterised by the conjunction of seven properties, which we will discuss now.
On the one hand, these properties have an intuitive explanation as characteristics of behaviours of (possibly concurrent) programs running in isolation.
Hence, if a pomset represents some execution of an isolated program, it is guarded.
On the other hand, the characterisation in terms of these properties provides a proof method to show that a pomset is \emph{not} guarded, by demonstrating the failure of one such property.

We start by observing that guarded pomsets alternate states and actions.
Formally, we can capture this in three properties.
Let $U = [\lp{u}] \in \SP(\Act\cup\State)$. We say that $s' \in S_\lp{u}$ is a \emph{predecessor} of $s$ if it is the latest node ordered strictly before $s$---i.e., $s$ $s' <_\lp{u} s$ and for all $s'' \in S_\lp{u}$ such that $s'' <_\lp{u} s$ it holds that $s'' \leq_\lp{u} s'$.
The notion of \emph{successor} is defined dually.
A node is a \emph{state-node} if it is labelled by an element of $\State$, and an \emph{action-node} otherwise.
\begin{enumerate}[label={(A\arabic*)},leftmargin=1cm]
  \item\label{item:minmax}
  $U$ admits a unique minimum and maximum, $*_{\min}, *_{\max} \in S_\lp{u}$, labelled by states.

  \item\label{item:alternating} Every two related state-nodes are separated by an action-node.

  \item\label{item:succpred} Action-nodes have unique state-nodes as neighbours (their predecessor and successor).
\end{enumerate}

The next property formalises the idea
that two related observations cannot contradict each other, such as in the program $(x=1)\cdot (x=2)$. To this end, we need the notion of a \emph{path}. A path for a variable $v$ from a state-node $u$ to another state-node $s$ is a chain such that the changes in the value of $v$ between $u$ and $s$ are explained by the actions between them and recorded in all the states between $u$ and $s$.

\begin{definition}[Path]\label{def:path}
Let $U = [\lp{u}] \in \Pom(\Act\cup \State)$ and $u_1,u_2\in S_\lp{u}$ such that $u_1\leq_\lp{u} u_2$ and $\lambda_\lp{u}(u_1),\lambda_\lp{u}(u_1)\in \State$.
We say a \emph{path} $p_v$ from $u_1$ to $u_2$ for variable $v\in \var$ is a sequence of nodes $q_1, a_1, \dots, a_n, q_{n+1} \in S_\lp{u}$ that satisfy the following conditions:
\begin{enumerate}[label={(P\arabic*)},leftmargin=1cm]
  \item\label{item:acties}
  For all $1 \leq i \leq n$, we have $\lambda_\lp{u}(a_i)\in \Act$ and $u_1\leq_\lp{u} a_i \leq_\lp{u} u_2$ for all $i$.
  Additionally we require that $a_i\leq_\lp{u}a_{i+1}$ for $1\leq i <n$.

  \item\label{item:guards}
  For all $1 \leq i \leq n + 1$ it holds that $\lambda_\lp{u}(q_i)\in \State$, and for all $1 \leq i \leq n$, the predecessor of $a_i$ is $q_{i}$, and the successor of $a_i$ is $q_{i+1}$.
  Additionally we have that $\lambda_\lp{u}(q_1)=\lambda_\lp{u}(u_1)$, $v\in\dom(\lambda_\lp{u}(u_1))$ and
  $\lambda_\lp{u}(q_{n+1})=\lambda_\lp{u}(u_2)$.
  Lastly, for $1\leq  i \leq n$ we have:
  \[
  \lambda_\lp{u}(q_{i+1})(v) =
  \begin{cases}
  n & \lambda_\lp{u}(a_{i})= v\leftarrow n\text{ for some }n\in \val\\
  \lambda_\lp{u}(q_{i})(v') & \lambda_\lp{u}(a_{i})= v\leftarrow v'\text{ for some }v'\in \var, v'\in\dom(\lambda_\lp{u}(q_{i}))\\
  \lambda_\lp{u}(q_{i})(v) & \text{otherwise }
  \end{cases}
  \]
\end{enumerate}
\end{definition}

\begin{example}
The following is a path for $x$:
\[(x=1)\cdot(y\leftarrow 3)\cdot (x=1)\cdot (x\leftarrow 2)\cdot(x=2\wedge y=3)\cdot (x\leftarrow y)\cdot(x=3)\]
Note that this is not a path for $y$, because it is not assigned a value by the final atom.
\end{example}

We can now formulate another criterion for a pomset executing in isolation: for every variable in the domain of a state-node, there is a path explaining the changes in value of that variable between the state-node and the maximum node of the pomset.
\begin{enumerate}[label={(A\arabic*)},leftmargin=1cm,resume]
      \item\label{item:path} For all state-nodes $u\in S_\lp{u}$ and $v\in\dom(\lambda_\lp{u}(u))$, there is a path for $v$ from $u$ to $*_{\max}$.
\end{enumerate}

\begin{example}
  The first pomset below satisfies~\ref{item:path}, and the second pomset does not, as there is no path from beginning to end for $x$: the value of $x$ in the second observation is not in accordance to the previous assignment.

  \begin{tikzpicture}[node distance=0.5cm]
      \begin{scope}[every node/.style={anchor=center}]
      \node (test1) {$(x=2\wedge y=2)$};
      \node[yshift=4mm,right=of test1] (act1) {$(x\leftarrow 4)$};
      \node[yshift=-4mm,right=of test1] (act2) {$(y\leftarrow 3)$};
      \node[yshift=-4mm, right=of act1] (test2) {$(x=4\wedge y=3)$};
      \end{scope}
      \draw[->] (test1) edge (act1);
      \draw[->] (test1) edge (act2);
      \draw[->] (act1) edge (test2);
      \draw[->] (act2) edge (test2);
  \end{tikzpicture}

  \begin{tikzpicture}[node distance=0.5cm]
      \begin{scope}[every node/.style={anchor=center}]
      \node (test1) {$(x=2\wedge y=4)$};
      \node[right=of test1] (act1) {$(x\leftarrow 4)$};
      \node[right=of act1] (test2) {$(x=5\wedge y=4)$};
      \node[right=of test2] (act3) {$(y\leftarrow 2)$};
      \node[right=of act3] (test5) {$(x=4\wedge y=2)$};
      \end{scope}
      \draw[->] (test1) edge (act1);
      \draw[->] (act1) edge (test2);
      \draw[->] (test2) edge (act3);
      \draw[->] (act3) edge (test5);
  \end{tikzpicture}
\end{example}

If a pomset represents an isolated program, an action has an effect on its successor.
If that action is of the form $v \leftarrow n$, then the sucessor should assign $n$ to $v$; likewise, if the action is of the form $v \leftarrow v'$, then the successor should assign the value of $v'$ to $v$, but the predecessor should also be aware of a value for $v'$.
\begin{enumerate}[label={(A\arabic*)},leftmargin=1cm,resume]
  \item\label{item:effect} If $u\in S_\lp{u}$ such that $\lambda_\lp{u}(u)=v\leftarrow n$ for some $v\in \var$ and $n\in \val$, we require that the successor of $u$ is $s$ s.t. $\lambda_\lp{u}(s)(v)=n$.
  \item\label{item:annoyingassignments}
  Let $u\in S_\lp{u}$ s.t. $\lambda_\lp{u}(u)=v\leftarrow v'$ for some $v,v'\in \var$ and let $p$ and $s$ be the predecessor, resp.\ successor, of $u$. Then $v'\in\dom(\lambda_{\lp{u}}(p))$ and
  $\lambda_\lp{u}(s)(v)=\lambda_\lp{u}(s)(v')=\lambda_\lp{u}(p)(v')$.
\end{enumerate}

\begin{example}
The pomset below on the left violates~\ref{item:effect}, because the successor of $(x \leftarrow 1)$ does not assign $2$ to $x$.
On the other hand, the pomset on the right satisfies~\ref{item:annoyingassignments}, because the predecessor of $(x \leftarrow y)$ has a value for $y$, and that value is assigned to $x$ in the successor.
\vspace{-0.7em}
\[
  \begin{tikzpicture}[node distance=0.5cm]
      \begin{scope}[every node/.style={anchor=center}]
      \node (test1) {$(x=0)$};
      \node[right=of test1] (act1) {$(x\leftarrow 1)$};
      \node[right=of act1] (test2) {$(x=2)$};
      \end{scope}
      \draw[->] (test1) edge (act1);
      \draw[->] (act1) edge (test2);
  \end{tikzpicture}
  \quad
  \begin{tikzpicture}[node distance=0.5cm]
      \begin{scope}[every node/.style={anchor=center}]
      \node (test1) {$(x=1\wedge y=2)$};
      \node[right=of test1] (act1) {$(x\leftarrow y)$};
      \node[right=of act1] (test2) {$(x=2\wedge y=2)$};
      \end{scope}
      \draw[->] (test1) edge (act1);
      \draw[->] (act1) edge (test2);
  \end{tikzpicture}
\]
\end{example}

Finally, isolated programs cannot observe variables that have not been assigned a value anywhere in the program.
On the pomset-level, this translates to:
\begin{enumerate}[label={(A\arabic*)},leftmargin=1cm,resume]
    \item\label{item:prepath} Let $u\in S_\lp{u}$ be a state-node. Then for all $v\in\dom(\lambda_\lp{u}(u))$, there exists a path for $v$ from $s\in S_\lp{u}$ to $u$
    such that either $v\in \dom(\lambda_\lp{u}(s))$ and $s=*_{\min}$ or $s$ is the successor of an assignment-node with label $v\leftarrow k$ with $k\in \var\cup \val$.
\end{enumerate}

\begin{example}
  The pomset on the left does not satisfy~\ref{item:prepath}, but the one on the right does.
\vspace{-0.7em}
\[
  \begin{tikzpicture}[node distance=0.5cm]
      \begin{scope}[every node/.style={anchor=center}]
      \node (test1) {$(x=1)$};
      \node[right=of test1] (act1) {$(x\leftarrow 2)$};
      \node[right=of act1] (test2) {$(x=2\wedge y=2)$};
      \end{scope}
      \draw[->] (test1) edge (act1);
      \draw[->] (act1) edge (test2);
  \end{tikzpicture}
  \quad
  \begin{tikzpicture}[node distance=0.5cm]
      \begin{scope}[every node/.style={anchor=center}]
      \node (test1) {$(x=1)$};
      \node[right=of test1] (act1) {$(y\leftarrow 2)$};
      \node[right=of act1] (test2) {$(x=1\wedge y=2)$};
      \end{scope}
      \draw[->] (test1) edge (act1);
      \draw[->] (act1) edge (test2);
  \end{tikzpicture}
\]
\end{example}

Guarded pomsets satisfy~\ref{item:minmax}--\ref{item:prepath}. In fact, there exists an equivalence:

\begin{restatable}{thrm}{fromguardedtocausal}\label{lemma:fromguardedtocausal}
For $U\in\SP$, $U$ is guarded if and only if $U$ satisfies~\ref{item:minmax}--\ref{item:prepath}.
\end{restatable}
\begin{proof}[Sketch]
The forward implication is proved by induction on the construction of $\G$.
For the other direction, we perform induction on the size of $U$ (which is possible because $U$ is series-parallel and therefore finite).
The induction hypothesis then states that whenever $V$ is strictly smaller than $U$, and $V$ satisfies the seven properties, then $V$ is guarded.
Since $U$ satisfies~\ref{item:minmax}, we know that either it consists of one node labelled by a state, in which case $U$ is immediately guarded, or $U=\alpha\cdot V\cdot \beta$ for $\alpha,\beta\in\State$ and $V \in \SP$.
This gives us four cases to consider: $V=1$, $V=\ltr{a}$ for some $\ltr{a}\in\Act\cup\State$, $V=V_0\cdot V_1$ or $V=V_0\parallel V_1$.
The first case can be disregarded as $U$ would then violate~\ref{item:alternating}.
In the second case,~\ref{item:path}--\ref{item:prepath} can be used to show that $\beta=\alpha[\ltr{a}]$.
In the latter two cases we show that $U$ is built out of two strictly smaller pomsets that satisfy~\ref{item:minmax}--\ref{item:prepath}, making them guarded by the induction hypothesis.
When these two pomsets are combined to form $U$, this is done according to the rules of guarded pomsets, making $U$ guarded as well.
\ifarxiv\else%
The details can be found in the full version of this paper~\cite{fullversion}.
\fi%
\end{proof}

\section{Litmus Test}\label{section:litmus}
The POCKA-semantics of a program captures all \emph{possible} behaviours
of the program, including all behaviours that could arise when it is put in
parallel with other programs. In this section we look at the behaviour
of the litmus test when it is executed in isolation. In the previous section we saw that if a pomset represents an execution of a program in isolation, it is guarded, and hence it is sufficient to look at the guarded pomsets. We demonstrate that there are in fact no guarded pomsets in the semantics of the litmus test, which shows that it passes. This suggests the guarded fragment of the POCKA-semantics is sequentially consistent: the programs behave as if memory accesses performed concurrently are interleaved and executed sequentially and writes to memory are broadcasted to all threads instantaneously.

Recall the litmus test $t$ we considered before, with $\var=\{x,y,r_0,r_1\}$ and $\val=\{0,1\}$:
\[t:=(r_0=0\wedge r_1=0)\cdot ((x\leftarrow 1\cdot r_0\leftarrow y) \parallel (y\leftarrow 1\cdot r_1\leftarrow x)) \cdot \overline{(r_0=1\vee r_1=1)} \]

Our strategy for showing that the semantics of $t$ does not contain guarded pomsets, is to first show that all pomsets in the semantics of $t$ have certain property. We then claim that if a pomset has this property, then it is not guarded, using~\ref{item:minmax}--\ref{item:prepath} from \cref{section:causal}.

\begin{definition}[Litmus Pomsets]\label{def:propq}
    Let $x,y,r_0,r_1,w\in \var$ be distinct and $0,1\in \val$.
    A pomset $U=[\lp{u}]$ has property $P$, denoted $P(U)$, if there exists $u_1,u_2,v_1,v_2,w\in
    S_\lp{u}$ s.t.
    \begin{enumerate}
      \setlength\itemsep{0em}
      \item\label{item:existence-in-u} the following conditions hold:
      \begin{mathpar}\lambda_\lp{u}(u_1)=(x\leftarrow
    1) \and \lambda_\lp{u}(u_2)=(y\leftarrow 1) \and \lambda_\lp{u}(v_1)=(r_0\leftarrow y) \and \lambda_\lp{u}(v_2)=(r_1\leftarrow x) \and
      \lambda_\lp{u}(w)(r_0)=0=\lambda_\lp{u}(w)(r_1)
      \and
  u_1 \leq_\lp{u} v_1 \leq_\lp{u} w
        \and u_2 \leq_\lp{u} v_2 \leq_\lp{u} w
      \end{mathpar}
      Graphically, we can represent these conditions as the following diagram:

      \hfill
      \begin{tikzpicture}[xscale=1.4,yscale=.4]
        \node(u1)at (0,2) {$u_1:x\leftarrow 1$};
        \node(u2)at (0,0) {$u_2:y\leftarrow 1$};
        \node(v1)at (2,2) {$v_1:r_0\leftarrow x$};
        \node(v2)at (2,0) {$v_2:r_1\leftarrow y$};
        \node(w)at (5,1) {$w:\left[
            \begin{array}{c@{\mapsto}c}
              r_0&0\\
              r_1&0
            \end{array}
          \right]$};
        \draw[->] (u1) to (v1);
        \draw[->] (u2) to (v2);
        \draw[->] (v1) to (w);
        \draw[->] (v2) to (w);
      \end{tikzpicture}
      \hfill$ $

      \item\label{item:relative-existence} For other assignment-nodes in $U$, we have the following conditions. Let $k \in\val\cup\var$.
      \begin{mathpar}
        \forall z. \lambda_\lp{u}(z)=(x\leftarrow k) \Rightarrow z\leq_\lp{u} u_1
      \and
      \forall z.
      \lambda_\lp{u}(z)=(y\leftarrow k)
      \Rightarrow z\leq_\lp{u} u_2
      \and
      \forall z.
      \lambda_\lp{u}(z)=(r_0\leftarrow k)
      \Rightarrow z\leq_\lp{u} v_1
      \and
      \forall z. \lambda_\lp{u}(z)=(r_1\leftarrow k) \Rightarrow z\leq_\lp{u} v_2
    \end{mathpar}
   \end{enumerate}
\end{definition}

The property $P$ describes the actions and observations found in the litmus test, and their relative ordering.
For instance, $\forall z. \lambda_\lp{u}(z)=(x\leftarrow n) \Rightarrow z\leq_\lp{u} u_1$ states that all action-nodes that change the value of $x$, occur before node $u_1$
Hence, the maximal node that alters the value of $x$, changes $x$ to $1$. The other requirements are explained similarly.

\begin{restatable}{lem}{ifpnotguarded}\label{lemma:ifpnotguarded}
    Let $U=[\lp{u}]\in\SP$. If $P(U)$ then $U$ is not guarded.
\end{restatable}

We show that $P$ is an invariant under closure w.r.t.\ $\hexch$ and $\hcontr$.
To this end, it is useful to study the effect of the contraction order on the level of pomsets; we introduce the following partial order relation on pomsets, analogous to the subsumption order.

\begin{definition}[Contraction Order]%
  \label{definition:contraction-pomsets}
  Let $U = [\lp{u}]$ and $V = [\lp{v}]$ be pomsets over $\Act \cup \State$.
  We write $U \preceq V$ holds iff there exists a surjection $h\colon S_\lp{v} \to S_\lp{u}$ satisfying: (i) $\lambda_{\lp{u}} \circ h = \lambda_{\lp{v}}$;
(ii) $v \leq_{\lp{v}} v'$ implies $h(v) \leq_{\lp{u}} h(v')$;
(iii) if $h(v) \leq_{\lp{u}} h(v')$, then $\lambda_{\lp{v}}(v),
\lambda_{\lp{v}}(v') \in \State$ implies $v \leq_{\lp{v}} v'$ or $v' \leq_{\lp{v}} v$, and $\lambda_{\lp{v}}(v)$ or
$\lambda_{\lp{v}}(v') \not\in \State$ implies $v \leq_{\lp{v}} v'$.
\end{definition}

We then prove the analogue of \cref{lemma:exch-closure-vs-subsumption}, relating $\preceq$ to closure w.r.t. $\hcontr$ as follows.

\begin{restatable}{lem}{lemmaclosurevscontraction}%
\label{lemma:closure-vs-contraction}
Let $L \subseteq \SP$ and $U \in \SP$.
Now $U\in\closure[\hcontr]L$ iff $U \preceq V$ for some $V \in L$.
\end{restatable}

With this characterisation in hand, we can prove that $P$ is invariant under closure.

\begin{restatable}{lem}{pclosure}\label{lemma:p-closure}
Let $e\in\termspocka$. If $\forall U\in\semwpocka{e}$ we have $P(U)$, then $\forall V\in \sempocka{e}$ it holds that $P(V)$.
\end{restatable}
\begin{proof}[Sketch]
By~\cite[Lemma~5.4]{fossacs2020} we know that $\sempocka{e} = \closurep[\hcontr]{\closure[\hexch]{\semwpocka{e}}}$.
It then follows, by \cref{lemma:exch-closure-vs-subsumption,lemma:closure-vs-contraction}, that if $V \in \sempocka{e}$, then there must exist $W, X \in \SP$ with $X \in \semwpocka{e}$ and $V \preceq W \sqsubseteq X$.
We then show that $P$ is preserved by both of these orders.
\end{proof}

\begin{corollary}
The semantics of the litmus test contains no guarded pomsets: $\sempocka{t}\cap \G=\emptyset$.
\end{corollary}
\begin{proof}
All pomsets in $\semwpocka{t}$ have property $P$ if we pick for $u_1$ the node with label $(x\leftarrow 1)$, for $v_1$ the node with label $(r_0\leftarrow y)$, and same for $u_2$ and $v_2$ (see \cref{example:semantics-litmus}). Lastly, we pick for $w$ the node with label $\delta$.
By \cref{lemma:p-closure} we can conclude that all pomsets in $\sempocka{t}$ have property $P$, and by \cref{lemma:ifpnotguarded} we infer that $t$ has no guarded pomsets in its semantics.
\end{proof}

We showed that we can correctly analyse the litmus test in our algebraic framework. In the next example we show that addition of one extra axiom, which is a commonly made assumption in programming languages, makes the litmus test fail on the guarded semantics.

\begin{example}
We add the following axiom, which states that assignments to different variables can be swapped as long as the assigned values are none of the involved variables:
\[
v\leftarrow k \cdot v'\leftarrow k' \equivpocka v'\leftarrow k' \cdot v\leftarrow k \text{ for } v,v'\in\var,\ k,k'\in\var\cup\val,\ k'\neq v\neq v',\ k\neq v\neq v'
\]
We show that with this assumption, which is commonly made in programming languages, we get guarded pomsets in the semantics of the litmus program. We can derive:
\begin{align*}
((r_0\leftarrow y)\cdot (r_1\leftarrow x)) &\equivpocka ((r_0\leftarrow y)\parallel 1)\cdot (1\parallel (r_1\leftarrow x)) \tag{Unit axiom} \\
&\leqqpocka ((r_0\leftarrow y) \cdot 1)\parallel (1\cdot (r_1\leftarrow x)) \tag{Exchange Law} \\
&\equivpocka (r_0\leftarrow y) \parallel (r_1\leftarrow x) \tag{Unit axiom}
\end{align*}
Similarly, we can derive that $(x\leftarrow 1)\cdot (y\leftarrow 1)\leqqpocka (x\leftarrow 1)\parallel (y\leftarrow 1)$. Hence, we have
\begin{align*}
& ((r_0\leftarrow y)\cdot (r_1\leftarrow x))\cdot ((x\leftarrow 1)\cdot (y\leftarrow 1))
\leqqpocka ((r_0\leftarrow y) \parallel (r_1\leftarrow x))\cdot ((x\leftarrow 1)\parallel (y\leftarrow 1)) \\
& \leqqpocka    ((r_0\leftarrow y) \cdot (x\leftarrow 1))\parallel ((r_1\leftarrow x)\cdot (y\leftarrow 1)) \tag{Exchange law} \\
& \equivpocka  ((x\leftarrow 1) \cdot (r_0\leftarrow y))\parallel ((y\leftarrow 1)\cdot (r_1\leftarrow x)) \tag{New axiom}
\end{align*}
Let $e= ((r_0\leftarrow y)\cdot (r_1\leftarrow x))\cdot ((x\leftarrow 1)\cdot (y\leftarrow 1))$. We can conclude that
\[
(r_0=0\wedge r_1=0)\cdot e \cdot \overline{(r_0=1\vee r_1=1)} \leqqpocka t
\]
From soundness, we infer that $\sempocka{(r_0=0\wedge r_1=0)\cdot e \cdot \overline{(r_0=1\vee r_1=1)} }\subseteq \sempocka{t}$. In the left set we find at least one guarded pomset. Let $\alpha=(r_0=0\wedge r_1=0\wedge x=0\wedge y=0)$, $\beta=(r_0=0\wedge r_1=0\wedge x=1\wedge y=0)$ and $\gamma = (r_0=0\wedge r_1=0\wedge x=1\wedge y=1)$.
\begin{tikzpicture}[node distance=0.5cm]
    \begin{scope}[every node/.style={anchor=center}]
    \node (alpha) {$\alpha$};
    \node[right=of alpha] (gamma1) {$(r_0\leftarrow y)$};
    \node[right=of gamma1] (alpha2) {$\alpha$};
    \node[right=of alpha2] (gamma2) {$(r_1\leftarrow x)$};
    \node[right=of gamma2] (alpha3) {$\alpha$};
    \node[right=of alpha3] (gamma3) {$(x\leftarrow 1)$};
    \node[right=of gamma3] (alpha4) {$\beta$};
    \node[right=of alpha4] (gamma4) {$(y\leftarrow 1)$};
    \node[right=of gamma4] (alpha5) {$\gamma$};
    \end{scope}
    \draw[->] (alpha) edge (gamma1);
    \draw[->] (gamma1) edge (alpha2);
    \draw[->] (alpha2) edge (gamma2);
    \draw[->] (gamma2) edge (alpha3);
    \draw[->] (alpha3) edge (gamma3);
    \draw[->] (gamma3) edge (alpha4);
    \draw[->] (alpha4) edge (gamma4);
    \draw[->] (gamma4) edge (alpha5);
\end{tikzpicture}

It is easy to show that this pomset is guarded by observing that $\alpha\cdot(r_0\leftarrow y)\cdot\alpha$, $\alpha\cdot (r_1\leftarrow x)\cdot \alpha$, $\alpha\cdot(x\leftarrow 1)\cdot\beta$ and $\beta\cdot(y\leftarrow 1)\cdot \gamma$ are all guarded. Hence, by adding this one extra axiom, we find guarded pomsets in the semantics of the litmus test, meaning that this axiom breaks sequential consistency.
\end{example}

\section{Discussion}\label{sec:discussion}
We presented POCKA, a sound and complete algebraic framework that can be used to analyse concurrent programs that manipulate variables. We identified the guarded fragment of the semantics, and showed this fragment captures the behaviour of programs executing in isolation. We demonstrated reasoning in POCKA by analysing a litmus test, also suggesting that the guarded fragment of the POCKA-semantics is sequentially consistent.

This work is built on Kleene algebra and extensions thereof. It is closest to Concurrent Kleene algebra with Observations~\cite{kappe-brunet-rot-silva-wagemaker-zanasi-2019,fossacs2020}, which was proposed to integrate concurrency with a form of tests (i.e., observations). We deviate from CKAO by using partial observations and accordingly changing the algebraic structure of observations (a PCDL instead of a Boolean algebra), and by incorporating explicit assignments and tests to manipulate variables. Programs such as the litmus test that we analyse in POCKA are outside the scope of CKAO\@.

The idea of using a PCDL and partial functions in the semantics comes from Jipsen and Moshier~\cite{jipsen-moshier-2016}. In the current paper we establish completeness w.r.t.\ the partial function model, which is missing in \emph{loc.\ cit.}
A further contrast is that POCKA includes as basic syntax atomic programs and assertions pertaining to variable assignment, as occur in the litmus test.
The definition of guarded pomset that we used is close to the one proposed in~\cite{jipsen-moshier-2016}. We provided an extensive analysis of guarded pomsets and showed how they can be used to study concrete program behaviour: our new characterisation in terms of concrete properties of pomsets (Theorem~\ref{lemma:fromguardedtocausal}) is essential for the analysis of the litmus test in Section~\ref{section:litmus}.

We suggest three avenues for future research. Firstly, the concrete observations and assignments that we have used are reminiscent of NetKAT~\cite{netkat,netkat2}, an algebraic framework based on Kleene algebra with tests that allows for reasoning about networks. 
POCKA is thus suggestive of a concurrent version of NetKAT, in which algebraic reasoning about concurrent networks could be studied. While NetKAT arises as a particular instance of KAT, POCKA is \emph{not} an instance of its closest relative in the Kleene algebra family, CKAO, due to the aforementioned move from an \emph{arbitrary} Boolean algebra of observations to a \emph{concrete} PCDL\@. It would therefore be of interest to formulate the necessary metatheory for the analogous framework of CKA with partial observations (where partial observations are given by an arbitrary PCDL), and situate POCKA within it.

This naturally leads to a third line of research. We have used the CKAH framework to obtain a completeness proof, and it turned out that the proof technique was perfectly amenable to a replacement of the Boolean algebra structure of observations with our observation algebra. This raises the question: which conditions are necessary on the algebraic structure of observations to be able to prove completeness in a similar manner? In particular, what conditions are needed for a result similar to \cref{lemma:reification} to hold? Our conjecture is that the observation algebra needs to be such that all elements can be written as a finite sum of join-irreducible elements of the algebra (cf. \cref{axiomaticsum}).

\bibliographystyle{plainurl}
\bibliography{bibliography}

\clearpage%
\appendix%
\section{Proofs about observation algebra}%
\label{appendix:proofs}

\soundnesspcdl*
\begin{proof}
The fact that $\mathrm{\acronym}$ is a PCDL and the assignment $\sempl{-}
$ is well-defined establishes the soundness of the PCDL axioms. We thus only have the domain-specific axioms left to verify.
\begin{itemize}
  \item
  If $n \neq m$, it is immediate that
  \[
    \sempl{v=n\wedge v=m}
        = \sempl{v=n} \cap \sempl{v=m}
        = \{\alpha\pipe\alpha(v)=n\} \cap \{\alpha\pipe\alpha(v)=m\}
        = \emptyset
        = \sempl{\bot}
  \]

  \item Next we show $\sempl{\overline{v=n}} \subseteq \sempl{\bigvee_{m \neq n} v=m}$.
	Assume
    \[
        \alpha \in \sempl{\overline{v=n}} = \bigcup\{B\in P_{\leq}(\State) \mid B \cap \sempl{v=n}=\emptyset\}.
    \]
	We have some downwards-closed $B \subseteq \State$ such that $\alpha \in B$ and $B \cap \sempl{v=n} = \emptyset$.
    Suppose towards a contradiction that $\alpha(v)$ is undefined, or that $\alpha(v) = n$.
    We can then choose $\alpha' \in \State$ to be $n$ on $v$, and identical to $\alpha$ elsewhere; in that case, $\alpha' \leq \alpha$, which means that $\alpha' \in B$.
    But then, since $\alpha' \in \sempl{v = n}$ by construction, we have a contradiction with the fact that $B \cap \sempl{v = n} = \emptyset$.
    Thus, there exists an $m \in \val$ such that $\alpha(v) = m$ and $m \neq n$.
    It then follows that $\alpha \in \sempl{\bigvee_{m \neq n} v = m}$.

  \item Next we prove that $\sempl{\overline{\bigwedge_{i}v_i=n_i}}\subseteq \sempl{\bigvee_{i}\overline{v_i=n_i}}$, if the $v_i$ are distinct.
  Let $\alpha \in B$ such that $B\cap \sempl{\bigwedge_{i}v_i=n_i} =\emptyset$.
We claim that for some $i$, $\alpha(v_i) = m \neq n_i$. Suppose otherwise: then for each $v_i$ either $\alpha(v_i) = n_i$ or $\alpha(v_i)$ is undefined. Define:
\[ \alpha'(v) ::= \begin{cases}
			n_i & \text{ if } v = v_i \text{ and } \alpha(v_i) \text{ undefined;} \\
			\alpha(v) & \text{ otherwise.}
			\end{cases} \]

This is well-defined by the assumption that the $v_i$ are all distinct. By construction, $\alpha'\leq \alpha$ and
  $\alpha'\in \bigcap_{i}\{\beta \pipe \beta(v_i)=n_i \}=\sempl{\bigwedge_{i}v_i=n_i}$. As $B$ is downwards-closed, we get $\alpha'\in B$, contradicting $B\cap \sempl{\bigwedge_{i}v_i=n_i} =\emptyset$.  Hence for some $i$, $\alpha(v_i) = m \neq n_i$. Hence $\alpha \in \sempl{\overline{v_i = n_i}} \subseteq \sempl{\bigvee_i \overline{v_i = n_i}}$ as required.
\end{itemize}

For the inductive step, we verify that the closure rules for congruence preserve soundness. This is all immediate from the definition of $\sempl{-}$. For instance, if $e = e_0 \lor e_1$, $f = f_0 \lor f_1$, $e_0 \equivpl f_0$ and $e_1 \equivpl f_1$, then $\sempl{e}=\sempl{e_0}\cup \sempl{e_1}=\sempl{f_0}\cup \sempl{f_1}=\sempl{f}$,
  where we have used that $\sempl{e_0}=\sempl{f_0}$ and $\sempl{e_1}=\sempl{f_1}$ by the induction hypothesis. \qedhere
\end{proof}

\lemmabasicfacts*
\begin{proof}
Assume $\alpha \in \sempl{\pi_\beta}$. Then for all $v \in \dom(\beta)$, $\alpha(v)$ is defined and $\alpha(v) =\beta(v)$, hence $\alpha \leq \beta$. Assume $\alpha \leq \beta$. Then
$\pi_{\alpha}\leqqpl \pi_{\beta}$ is established from $\pi_{\alpha}\wedge \pi_{\beta}\equivpl \pi_{\alpha}$: by the assumption, every conjunct in $\bigwedge_{\beta(v)=n}v=n$ appears as a conjunct in $\bigwedge_{\alpha(v)=n}v=n$, so by idempotence $\pi_\alpha \land \pi_\beta \equivpl \Bigl( \bigwedge_{\alpha(v)=n}v=n \Bigr) \wedge \Bigl(\bigwedge_{\beta(v)=n}v=n \Bigr) \equivpl \bigwedge_{\alpha(v)=n}v=n \equivpl \pi_\alpha$.
Finally, assume $\pi_\alpha \leqqpl \pi_{\beta}$. By soundness $\sempl{\pi_\alpha} \subseteq \sempl{\pi_\beta}$, and it is trivial to
establish $\alpha \in \sempl{\pi_\alpha}$.
\end{proof}

\axiomaticsum*
\begin{proof}
	Noting that $\bigvee \{\alpha \in \State \mid \alpha\leqqpl p\} \leqqpl p$ by definition, we focus on the other inequality,
	proceeding by induction on $p$. For the base cases, $\bot \leqqpl \bigvee \{\alpha \in \State \mid \alpha\leqqpl \bot \}$ by definition; $\top \leqqpl \bigvee \{\alpha \in \State \mid \alpha\leqqpl \top \}$ as
	$\top \equivpl \pi_{\emptyset} \in \{ \alpha \in \State \mid \alpha \leqqpl \top \}$; and
	$v=n \leqqpl \bigvee \{ \alpha \in \State \mid \alpha \leqqpl v = n \}$ as
	$v=n \equivpl \pi_{\{v \mapsto n\}} \in \{ \alpha \in \State \mid \alpha \leqqpl v = n \}$.

In the induction step we have three cases.
\begin{itemize}
  \item If $p = p_0\wedge p_1$, by the inductive hypothesis and distributivity we obtain
	\[p_0 \land p_1 \leqqpl \bigvee\{\alpha \land  \beta \mid \alpha \leqqpl p_0, \beta\leqqpl p_1\}.\]
	We claim that
	$\{\alpha \land  \beta \mid \alpha \leqqpl p_0, \beta\leqqpl p_1\} \subseteq
	\{\alpha \mid \alpha \leqqpl p_0 \land p_1\} \cup \{ \bot \}$. Call $\alpha, \beta \in \State$ \emph{compatible} if,
	for any $v \in \dom(\alpha) \cap \dom(\beta)$, $\alpha(v) = \beta(v)$. Take $\alpha \leqqpl p_0$ and
	$\beta \leqqpl p_1$. There are two cases. In the first, $\alpha$ and $\beta$ are compatible. Then
	define $\gamma$ by
	\[ \gamma(v) ::= \begin{cases}
				\alpha(v) & \text{ if } \alpha(v) \text{ defined;} \\
				\beta(v) & \text{ if } \beta(v) \text{ defined;} \\
				\text{undefined} & \text{ otherwise.} \\
				\end{cases} \]
	This is well-defined by compatability. Then $\gamma \leq \alpha$ as well as $\gamma \leq \beta$, and hence by \cref{lemma:basicfacts} we find $\gamma \leqqpl \alpha \land \beta \leqqpl p_0 \land p_1$.
	In the other case, $\alpha$ and $\beta$ are not compatible: hence for some distinct $n$ and $m$, $v=n$ and $v=m$ are among the conjuncts of $\alpha \land \beta$.
	By the axiom $v=n \land v=m \equivpl \bot$, it then follows that $\alpha \land \beta \equivpl \bot$.
	We obtain
	\[ \bigvee \{\alpha \land \beta \mid \alpha \leqqpl p_0, \beta\leqqpl p_1\}
		\leqqpl \bigvee (\{\alpha \mid \alpha \leqqpl p_0 \land p_1\} \cup \{ \bot \})
		\equivpl \bigvee \{\alpha \mid \alpha \leqqpl p_0 \land p_1\}. \]

  \item If $p = p_0\vee p_1$, we derive
  \begin{align*}
    p_0\vee p_1 & \leqqpl \bigvee\{\alpha \in \State \mid \alpha\leqqpl p_0\} \vee \bigvee\{\beta \in \State \mid \beta\leqqpl p_1\} \tag{IH} \\
          & \leqqpl \bigvee\{\alpha \in \State \mid\alpha\leqqpl p_1\vee p_2\} \tag{$\alpha\leqqpl p_0\leqqpl p_0\vee p_1$, similar for $\beta$} 
  \end{align*}
  \item If $p = \overline{p_0}$, we derive
  \begin{align*}
    \overline{p_0} & \equivpl \overline{\bigvee\{\alpha\mid \alpha\leqqpl p_0\}} \tag{IH} \\
          & \equivpl \bigwedge\{\overline{\alpha}\mid \alpha\leqqpl p_0\} \tag{De Morgan} \\
          & \equivpl \bigwedge\{ \overline{\bigwedge_{\alpha(v)=n} v = n}\mid \alpha\leqqpl p_0\} \tag{Definition of $\pi_\alpha = \alpha$} \\ 
          & \leqqpl \bigwedge\{\bigvee_{\alpha(v)=n} \overline{v = n} \mid \alpha\leqqpl p_0\}  \tag{De Morgan-like domain-specific axiom} \\
          & \leqqpl \bigwedge\{\bigvee_{\substack{\alpha(v)=n \\ m \neq n}} v=m \mid \alpha\leqqpl p_0\} \tag{Pseudocomplement domain-specific axiom}
  \end{align*}
Note that the De Morgan law applied in the second step is indeed satisfied by PCDLs~\cite{blyth-2005}.
Now, define $K ::= \{ \alpha \in \State \mid \alpha \leqqpl p_0\}$, $J_\alpha ::= \{ (v, m) \mid \alpha(v) = n \neq m \}$, $J ::= \bigcup_{\alpha \in K} J_\alpha$ and $F ::= \{ f: K \rightarrow J \mid \forall \alpha \in K, f(\alpha) \in J_{\alpha} \}$. Further, let
\[ p_{\alpha, (v, m)} ::= \begin{cases}
					v=m & \text{ if } \alpha(v) = n \neq m; \\
					\bot & \text{ otherwise}
				\end{cases} \]
Then \[\bigwedge\{\bigvee_{\substack{\alpha(v)=n \\ m \neq n}} v=m \mid \alpha\leqqpl p_0\}
\equivpl \bigwedge_{\alpha \in K} \bigvee_{(v, m) \in J_{\alpha}} p_{\alpha, (v,m)} \equivpl
\bigvee_{f \in F} \bigwedge_{\alpha \in K} p_{\alpha, f(\alpha)}\]
by distributivity. For each $f \in F$, if the $p_{\alpha, f(\alpha)}$ are compatible,
$\bigwedge_{\alpha \in K} p_{\alpha, f(\alpha)} \equivpl \beta_f$, for a $\beta_f$ with the property that for every
$\alpha \leqqpl p_0$, $\beta_f \land \alpha = \bot$ (as by definition, for each such $\alpha$, $\beta_f$
has some $v = m$ as a conjunct, where $\alpha$ has a conjunct $v = n$ for $n \neq m$). If they are incompatible,
$\bigwedge_{\alpha \in K} p_{\alpha, f(\alpha)} \equivpl \bot$. Hence
\[ \bigvee_{f \in F} \bigwedge_{\alpha \in K} p_{\alpha, f(\alpha)} \leqqpl
	\bigvee \{ \beta \mid \text{for all } \alpha, \alpha \leqqpl p_0 \text{ implies } \alpha \land \beta \equivpl \bot\}. \]
For any $\beta$ satisfying the property that $\text{for all } \alpha, \alpha \leqqpl p_0 \text{ implies } \alpha \land \beta \equivpl \bot$, we have
$\beta \land p_0 \equivpl \beta \land \bigvee \{ \alpha \mid \alpha \leqqpl p_0 \} \equivpl \bigvee \{ \alpha \land \beta \mid \alpha \leqqpl
p_0 \} \equivpl
 \bot$, so $\beta \leqqpl \overline{p_0}$ and
\[ \bigvee \{ \beta \mid \text{for all } \alpha \leqqpl p_0 \text{ implies } \alpha \land \beta \equivpl \bot\} \leqqpl \bigvee \{ \beta \mid \beta \leqqpl \overline{p_0} \}, \]
completing the proof. \qedhere
\end{itemize}
\end{proof}

\soundnesscompletenesspcdl*
\begin{proof}
The left-to-right direction follows from \cref{soundnesspcdl}.
For the right-to-left direction, 
suppose that $\sempl{p}=\sempl{q}$. By \cref{axiomaticsum}, we obtain that
\[
\sempl{p}=\sempl{\bigvee\{\alpha \in \State \mid \alpha\leqqpl p\}}=\sempl{\bigvee\{\beta \in \State \mid \beta\leqqpl q\}}=\sempl{q}.
\]

We prove $\alpha\leqqpl p$ if and only if $\alpha \leqqpl q$. Take $\alpha\leqqpl p$.
Then $\alpha \in \sempl{\alpha} \subseteq \sempl{p} = \bigcup_{\beta \leqqpl q} \sempl{\beta}$ by \cref{soundnesspcdl,lemma:basicfacts}.
Hence for some $\beta \leqqpl q$, $\alpha \in \sempl{\beta}$, and by \cref{lemma:basicfacts} once more,
$\alpha \leqqpl \beta \leqqpl q$. The other direction is symmetric. It follows that
\[ p \equivpl \bigvee \{ \alpha \in \State \mid \alpha \leqqpl p \} \equivpl \bigvee \{ \alpha \in \State \mid \alpha \leqqpl q\} \equivpl q, \]
as required.
\end{proof}

\section{Proofs towards completeness}

The following three results are all needed in the proofs that follow, and come from~\cite{fossacs2020}.
First of all, we can prove the following useful properties about the interaction between closure and other operators on pomset languages:
\begin{lemma}%
\label{lemma:composition-vs-closure}
Let $L, K \subseteq \Pom$ and $C \in \PC$.
The following hold.

\smallskip
\begin{minipage}{0.40\textwidth}
\begin{enumerate}
    \item\label{property:hypothesis-closure}
    $L \subseteq {\closure K}$ iff ${\closure L} \subseteq {\closure K}$.

    \item\label{property:hypothesis-monotone}
    If $L \subseteq K$, then ${\closure L} \subseteq {\closure K}$.

    \item\label{property:hypothesis-union}
    ${\closurep {L \cup K}} = {\closurep {{\closure L} \cup {\closure K}}}$

    \item\label{property:hypothesis-concat}
    ${\closurep {L \cdot K}} = {\closurep {{\closure L} \cdot {\closure K}}}$
\end{enumerate}
\end{minipage}
\begin{minipage}{0.54\textwidth}
\begin{enumerate}
 \setcounter{enumi}{4}
    \item\label{property:hypothesis-parallel}
    ${\closurep {L \parallel K}} = {\closurep {{\closure L} \parallel {\closure K}}}$

    \item\label{property:hypothesis-star}
    ${\closurep {L^*}} = \closure {({\left(\closure L\right)}^*)}$

    \item\label{property:hypothesis-context}
    If $\closure{L} \subseteq \closure{K}$, then $\closure{C[L]} \subseteq \closure{C[K]}$.

    \item\label{property:hypothesis-sp}
    If $L \subseteq \SP$, then $\closure{L} \subseteq \SP$.
\end{enumerate}
\end{minipage}
\end{lemma}

Second, we can note the following result about the interaction between $\hexch$ and $\hcontr$:

\begin{lemma}%
\label{lemma:factorise-exch}
For any $L\in 2^{\SP}$, we have $\closure[\hcontr \cup
\hexch]{L}=\closure[\hcontr]{(\closure[\hexch]{L})}$.
\end{lemma}

We can then prove soundness of the POCKA semantics w.r.t.\ the axioms.

\soundnesspocka*
\begin{proof}
  By construction it is immediate that $\sempocka{-}$ is closed under closure with respect to $\hexch\cup\hcontr$.
  We then proceed by induction on $\equivpocka$. For all the pairs from $\equivbkaexpl$, it follows from \cref{theorem:bka-completeness} that $\semwpocka{e}=\semwpocka{f}$. Then immediately their POCKA-semantics also coincide. For all
  the pairs from $\equivplexpl$,
  we make use of \cref{soundnesscompletenesspcdl}. Note that $\sempocka{-}$ almost coincides with $\semplexpl{-}$ on $\termspl$, so the proof is very straightforward.
  For instance, take $p\vee(p\wedge q)\equivplexpl p$. Then $\semwpocka{p\vee(p\wedge q)}= \State^* \cdot \semplexpl{p\vee(p\wedge q)} \cdot \State^*$.
  From \cref{soundnesscompletenesspcdl}, we know that $\semplexpl{p\vee(p\wedge q)}=\semplexpl{p}$, and thus we obtain
  $\semwpocka{p\vee(p\wedge q)}= \State^* \cdot \semplexpl{p} \cdot \State^* = \semwpocka {p}$. Then from this we can conclude that
 $\sempocka{p\vee(p\wedge q)}=\closure[\hexch\cup\hcontr]{\semwpocka{p\vee(p\wedge q)}}=\closure[\hexch\cup\hcontr]{\semwpocka{p}}=\sempocka{p}$.
 We can prove soundness of the other observation algebra axioms analogously.

  The next axiom is the exchange law. We show that
  $\sempocka{(e\parallel f)\cdot (g\parallel h)}\subseteq\sempocka{(e\cdot
  g)\parallel (f\cdot h)}$.
  By \cref{lemma:composition-vs-closure}\eqref{property:hypothesis-closure}, it
  suffices to prove that
  $\semwpocka{(e\parallel f)\cdot (g\parallel h)}\subseteq\sempocka{(e\cdot
  g)\parallel (f\cdot h)}$. Take an element in $\semwpocka{(e\parallel f)\cdot
  (g\parallel h)}$. This is thus a pomset of the form $(X\parallel Y)\cdot
  (V\parallel W)$ for $X\in \semwpocka{e}$, $Y\in\semwpocka{f}$,
  $V\in\semwpocka{g}$ and $W\in\semwpocka{h}$. Thus we immediately obtain that
  $(X\cdot V)\parallel (Y\cdot W)\in \semwpocka{(e\cdot
  g)\parallel (f\cdot h)}$. From \cref{lemma:factorise-exch}, we know that
  $\sempocka{-}=\closure[\hcontr]{(\closure[\hexch]{\semwpocka{-}})}$. We know that $(X\parallel Y)\cdot
  (V\parallel W)\sqsubseteq(X\cdot V)\parallel (Y\cdot W) $ and that $(X\cdot V)\parallel (Y\cdot W)\in \closure[\hcontr]{(\closure[\hexch]{\semwpocka{(e\cdot
  g)\parallel (f\cdot h)}})}$.
  Then we can apply \cref{lemma:factorise-exch,lemma:exch-closure-vs-subsumption} to obtain that $(X\parallel Y)\cdot
  (V\parallel W)\in\sempocka{(e\cdot
  g)\parallel (f\cdot h)}$.

Left to verify are the interface axioms.
  To check $p\wedge q\leqqpocka p\cdot q$ it again suffices to prove that
  $\semwpocka{p\wedge q}\subseteq\sempocka{p\cdot q}$ by
  \cref{lemma:composition-vs-closure}\eqref{property:hypothesis-closure}. We take an element in
  $\semwpocka{p\wedge q}$. This a pomset of the form $U\cdot \alpha \cdot V$ such
  that $U,V\in\State^*$ and $\alpha\in\semplexpl{p}\cap\semplexpl{q}$. We can establish that $U\cdot
  \alpha\cdot\alpha\cdot V\in\semwpocka{p\cdot q}$. Now take the pomset context $C=U\cdot *\cdot V$.
  We have that $C[\sembka{\alpha\cdot\alpha}]=\{U\cdot \alpha\cdot \alpha\cdot V\}\subseteq\semwpocka{p\cdot q}\subseteq \sempocka{p\cdot q} $. Then by closure we find $C[\sembka{\alpha}]=\{U\cdot \alpha\cdot V\}\subseteq \sempocka{p\cdot q} $.

  Since $\semwpocka{\bot} = \emptyset = \semwpocka{0}$, it follows that $\sempocka{\bot}=\emptyset=\sempocka{0}$.
  Similarly, since $\semwpocka{p + q} = \semwpocka{p} \cup \semwpocka{q} = \semwpocka{p \vee q}$, we also have that $\sempocka{p + q} = \sempocka{p \vee q}$

  Now we have four axioms left. The first one we verify is $\top\cdot p\leqqpocka p$ for $p\in\termspl$. We immediately obtain that $\semwpocka{\top\cdot p}=\State\cdot\State^*\cdot\semplexpl{p}\cdot \State^*\subseteq \State^*\cdot\semplexpl{p}\cdot \State^*=\semwpocka{p}\subseteq \sempocka{p}$.
  The axioms $\top\cdot p\leqqpocka p$, $a\cdot \top \leqqpocka a$, and $\top\cdot a \leqqpocka a$ for $a\in\Act$ are all verified in a similar manner.

  In the inductive step we need to check whether the closure rules for congruence have been preserved. We distinguish four cases.
  \begin{itemize}
    \item
    If $e = e_0 + e_1$ and $f = f_0 + f_1$ with $e_0 \equivpocka f_0$ and $e_1 \equivpocka f_1$, then by induction we know that $\sempocka{e_0} = \sempocka{f_0}$ and $\sempocka{e_1} = \sempocka{f_1}$.
    By \cref{lemma:composition-vs-closure}\eqref{property:hypothesis-union}, we can then derive that
    \begin{align*}
      \sempocka{e} &= \closure[\hexch\cup\hcontr] {(\semwpocka{e_0} \cup \semwpocka{e_1})} \\
      & = \closure[\hexch\cup\hcontr] {(\closure[\hexch\cup\hcontr] {\semwpocka{e_0}} \cup
     \closure[\hexch\cup\hcontr]{\semwpocka{e_1}})} \\
     & =
     \closure [\hexch\cup\hcontr]{(\closure [\hexch\cup\hcontr]{\semwpocka{f_0}} \cup
      \closure[\hexch\cup\hcontr]{\semwpocka{f_1}})} \\
     &= \closure[\hexch\cup\hcontr] {(\semwpocka{f_0} \cup \semwpocka{f_1})}
     \\
     & = \closure[\hexch\cup\hcontr]{\semwpocka{f}}
     \end{align*}
    \item
    The cases for $\cdot$, $\parallel$ and $*$ are argued similarly.\qedhere
\end{itemize}
\end{proof}

\ifarxiv%

Next, we need~\cite[Lemma~3.4]{fossacs2020}, which goes as follows.
\begin{lemma}%
\label{lemma:tinycontextparallel}
Let $C \in \PC$, let $V, W \in \Pom$, and $\ltr{a} \in \alphabet$.
If $C[U] = V \parallel W$, then there exists a $C' \in \PC$ such that $C = C' \parallel W$ and $C'[U] = V$, or $C = V \parallel C'$ and $C'[U] = W$.
\end{lemma}

Analogously to the previous lemma, we can prove the following.
\begin{lemma}\label{lemma:tinycontextsequential}
  Let $C \in \PC$, let $V, W \in \Pom$, and $\ltr{a} \in \alphabet$.
  If $C[\ltr{a}] = V \cdot W$, then there exists a $C' \in \PC$ such that $C = C' \cdot W$ and $C'[\ltr{a}] = V$, or $C = V \cdot C'$ and $C'[\ltr{a}] = W$.
\end{lemma}

The following properties of $\htop$ will also be useful.

\begin{lemma}%
\label{lemma:composition-vs-top}
Let $L, K \subseteq \SP$.
The following hold:
\begin{mathpar}
\closure[\htop]{L} \cup \closure[\htop]{K} = \closurep[\htop]{L \cup K}
\and
\closure[\htop]{L} \cdot \closure[\htop]{K} = \closurep[\htop]{L \cdot K}
\and
\closure[\htop]{L} \parallel \closure[\htop]{K} = \closurep[\htop]{L \parallel K}
\and
{\left( \closure[\htop]{L} \right)}^* = \closurep[\htop]{L^*}
\end{mathpar}
\end{lemma}
\begin{proof}
The inclusions from left to right are a consequence of \cref{lemma:composition-vs-closure}.

For the other inclusions, it is useful to note that, since for all $e \leq f \in \htop$ we have that $\sembka{e}$ as well as $\sembka{f}$ are singletons, the language $\closure[\htop]{L}$ can be more simply described as the smallest set containing $L$ and satifying the rule
\begin{mathpar}
\inferrule{%
    c \in \Act \cup \State \\
    \alpha \in \State \\
    C \in \PC \\
    C[c] \in \closure[\htop]{L}
}{%
    C[\alpha \cdot c] \in \closure[\htop]{L} \\
    C[c \cdot \alpha] \in \closure[\htop]{L} \\
}
\end{mathpar}

We now treat the other inclusions on a case-by-case basis.
for the first one, we show more generally that if ${(L_n)}_{n \in \naturals}$ is a family of pomset languages, then $\closure[\htop]{\Bigl( \bigcup_{n \in \naturals} L_n \Bigr)} \subseteq \bigcup_{n \in \naturals} \closure[\htop]{L_n}$.
We do this by induction on the construction of $U \in \closure[\htop]{\Bigl( \bigcup_{n \in \naturals} L_n \Bigr)}$ as induced by the characterisation of $\closurep[\htop]{-}$ above.
In the base, $U \in \bigcup_{n \in \naturals} L_n$, in which case the claim follows immediately.
Otherwise, if there exist $c \in \State \cup \Act$, $\alpha \in \State$ and $C \in \PC$ such that $C[c] \in \closure[\htop]{\Bigl( \bigcup_{n \in \naturals} L_n \Bigr)}$ and either $U = C[\alpha \cdot c]$ or $U = C[c \cdot \alpha]$, then by induction we know that $C[c] \in \closure[\htop]{L_n}$ for some $n \in \naturals$.
If $U = C[\alpha \cdot c]$, then $U = C[\alpha \cdot c] \in \closure[\htop]{L_n} \subseteq \bigcup_{n \in \naturals} \closure[\htop]{L_n}$.
The case where $U = C[c \cdot \alpha]$ can be treated similarly.

Next, we show that if $U \in \closurep[\htop]{L \cdot K}$, then $U \in \closure[\htop]{L} \cdot \closure[\htop]{K}$.
In the base, $U \in L \cdot K \subseteq \closure[\htop]{L} \cdot \closure[\htop]{K}$.
Otherwise, if there exist $c \in \State \cup \Act$, $\alpha \in \State$ and $C \in \PC$ such that $C[c] \in \closurep[\htop]{L \cdot K}$ and either $U = C[\alpha \cdot c]$ or $U = C[c \cdot \alpha]$, then by induction we have $C[c] \in \closure[\htop]{L} \cdot \closure[\htop]{K}$, i.e., $C[c] = V \cdot W$ such that $V \in \closure[\htop]{L}$ and $W \in \closure[\htop]{K}$.
By \cref{lemma:tinycontextsequential}, we find $C' \in \PC$ with either $C = C' \cdot W$ and $C'[c] = V$, or $C = V \cdot C'$ and $C'[c] = W$.
If $U = C[\alpha \cdot c]$, then in the former case $U = C'[\alpha \cdot c] \cdot W \in \closure[\htop]{L} \cdot \closure[\htop]{K}$; in the latter case, $U = V \cdot C'[\alpha \cdot c] \in \closure[\htop]{L} \cdot \closure[\htop]{K}$.
The case where $U = C[c \cdot \alpha]$ can be treated similarly.

An argument similar to the case for sequential composition shows that if $\closurep[\htop]{L \parallel K}$ is contained in $\closure[\htop]{L} \parallel \closure[\htop]{K}$, where this time we use \cref{lemma:tinycontextparallel}.

The final containment can be showed using the arguments above; after all, we have
\[
    \closurep[\htop]{L^*}
        = \closure[\htop]{\Bigl( \bigcup_{n \in \naturals} L^n \Bigr)}
        \subseteq \bigcup_{n \in \naturals} \closurep[\htop]{L^n}
        \subseteq \bigcup_{n \in \naturals} {\left( \closure[\htop]{L} \right)}^n
        = {\left( \closure[\htop]{L} \right)}^*
    \qedhere
\]
\end{proof}

\begin{lemma}\label{lemma:closed-alt}
  For all $e\in\termspocka$, we have that
  $\semwpocka{e}=\closure[\htop]{\semwpocka{e}}$.
\end{lemma}
\begin{proof}
  By \cref{lemma:composition-vs-top}, it suffices to prove the claim for all $a \in \Act \cup \termspl$; by definition of closure, we already know that $\semwpocka{a} \subseteq \closure[\htop]{\semwpocka{a}}$.
  For the converse inclusion, we proceed by induction on the construction of $\closure[\htop]{\semwpocka{a}}$ as characterised in the previous proof.
  Thus, in the base we have $U \in \closure[\htop]{\semwpocka{a}}$ because $U \in \semwpocka{a}$, in which case the claim holds vacuously.
  For the inductive step, we have $U \in \closure[\htop]{\semwpocka{a}}$ because there exist $c \in \State \cup \Act$, $\alpha \in \State$, and $C \in \PC$ such that $C[c] \in \closure[\htop]{\semwpocka{a}}$ and either $U = C[\alpha \cdot c]$ or $U = C[c \cdot \alpha]$.
  By induction, we find that $C[c] \in \semwpocka{a}$, and hence $C[c] = V \cdot W \cdot X$ for $V, X \in \State^*$ and $W \in \semplexpl{a}$ or $W = a$, depending on whether $a \in \termspl$ or $a \in \Act$ respectively.
  By applying \cref{lemma:tinycontextsequential} twice, we find that either $C = C' \cdot W \cdot X$ with $C'[c] = V$, or $C = V \cdot C' \cdot X$ with $C'[c] = W$, or $C = V \cdot W \cdot C'$ with $C'[c] = X$.
  In the first and the last case, we find that $C'[\alpha \cdot c], C'[c \cdot \alpha] \in \State^*$ as well, and hence $C[\alpha \cdot c], C[c \cdot \alpha] \in \semwpocka{a}$.
  In the second case, we note that $C'[\alpha \cdot c], C'[c \cdot \alpha] \subseteq \State^* \cdot L \cdot \State^*$ for $L = \semplexpl{a}$ or $L = \{ a \}$ (again, depending on whether $a \in \termspl$ or $a \in \Act$), and hence
  \[
    C[\alpha \cdot c], C[c \cdot \alpha]
        \in \State^* \cdot \State^* \cdot L \cdot \State^* \cdot \State^*
        = \State^* \cdot L \cdot \State^* = \semwpocka{a}.
        \qedhere
  \]
\end{proof}

\lemmaunclosedvstop*
\begin{proof}
  We proceed by induction on $e$.
  In the base, $e = a$ for some $a \in \Act \cup \termspl$.

  First, let $x \in \semwpocka{a}$; in that case, $x = u \cdot y \cdot v$ for $u, v \in \State^*$ and $y \in \State \cup \Act$.
  We claim that $u\cdot y$ is an element of $\closure[\htop]{\sembka{\hat r(a)}}$.
  We proceed by induction on the length of $u$.
  If $u=1$, then there are two cases to consider.
  \begin{itemize}
    \item
    If $a \in \Act$, then $y = a = \hat r(a)$; in that case, we find that $y \in \sembka{y} = \sembka{\hat r(a)}$.
    \item
    If $a \in \termspl$, then $y \in \State$ and $y \leqqplexpl a$.
    In that case $y \in \sembka{\hat r(a)}$ as well.
  \end{itemize}
  Thus, we have $u \cdot y = y \in \sembka{\hat r(a)}\subseteq \closure[\htop]{\sembka{\hat r(a)}}$.
  If $u$ has length $m+1$, we know that $u=u'\cdot \alpha$ for some $\alpha\in \State$ and $u'$ has length $m$.
  Our induction hypothesis tells us that $u'\cdot y \in \closure[\htop]{\sembka{\hat r(a)}}$.
  We can take $C=u'\cdot *$ to obtain from the induction hypothesis that $\{C[y]\}=C[\sembka{y}]\subseteq \closure[\htop]{\sembka{\hat r(a)}}$.
  As $\alpha\cdot y\leq y\in \htop$, we obtain that $C[\sembka{\alpha\cdot y}]\subseteq \closure[\htop]{\sembka{\hat r(a)}}$ from the definition of closure.
  Because $C[\sembka{\alpha\cdot y}]=\{u'\cdot\alpha\cdot y \}=\{u\cdot y\}$, we have reached the desired conclusion.
  Then we can show in a similar matter that $u\cdot y \cdot v\in \closure[\htop]{\sembka{\hat r(a)}}$.

  For the other direction, we start by proving that $\sembka{\hat{r}(a)}\subseteq \semwpocka{a}$.
  We have two cases.
  \begin{itemize}
    \item
    If $a \in \Act$, then $\sembka{\hat r(a)}=\{a\}\subseteq\State^*\cdot\{a\}\cdot \State^*=\semwpocka{a}$.
    \item
    Otherwise, $a = p$ with $p\in\termspl$.
    We have $\sembka{\hat r(p)}=\sembka{\sum_{\alpha\leqqplexpl p}\alpha}$.
    Take $x\in\sembka{\hat r(p)}$.
    Thus $x=\alpha$ for some $\alpha\leqqplexpl p$.
    Applying \cref{axiomaticsum,lemma:basicfacts}, we obtain that
    \[
        x\in\State^*\cdot\semplexpl{\bigvee_{\alpha\leqqplexpl p}\alpha}\cdot \State^*=\State^*\cdot\semplexpl{p}\cdot \State^*=\semwpocka{p}.
    \]
  \end{itemize}
From \cref{lemma:closed-alt} we know that $\semwpocka{a} = \closure[\htop]{\semwpocka{a}}$; it then follows that $\closure[\htop]{\sembka{\hat{r}(a)}} \subseteq \semwpocka{a}$.

  For the inductive step, we can rely on \cref{lemma:composition-vs-top}.
  For instance, if $e = e_0 + e_1$, then $\closure[\htop]{\sembka{\hat{r}(e)}} = \closurep[\htop]{\sembka{\hat{r}(e_0)} \cup \sembka{\hat{r}(e_1)}} = \closure[\htop]{\sembka{\hat{r}(e_0)}} \cup \closure[\htop]{\sembka{\hat{r}(e_1)}} = \semwpocka{e_0} \cup \semwpocka{e_1} = \semwpocka{e}$.
\end{proof}

\completenessalt*
\begin{proof}
  In the following, let $g = \sum_{\alpha \leq \top} \alpha$.
  We start by proving $e\equivbkaexpl^{\htop}\hat s(e)$.
  We know already that for $a\in\Act\cup\State$, we have $a \leqqbkaexpl^{\htop} g^* \cdot a  \cdot g^* = \hat{s}(a)$, since $1 \leqq g^*$.

  For the other direction, it suffices to show that $\hat s(a) \leqqbkaexpl^{\htop} a$ for $a \in \State \cup \Act$.
  To this end, let $a \in \State \cup \Act$.
  We are going to use the least fixpoint axioms of KA\@.
  We know that $g \cdot a \leqqbkaexpl^{\htop} a$ and $a \cdot g \leqqbkaexpl^{\htop} a$ follow immediately from the hypotheses in $\htop$ and distributivity.
  Hence, we have that $a + a \cdot g \leqqbkaexpl^{\htop} a$.
  Now we can apply one of the least fixpoint axioms ($e+(f\cdot h)\leqqbkaexpl f \Rightarrow e\cdot h^* \leqqbkaexpl f$) to obtain $a \cdot g^* \leqqbkaexpl^{\htop} a$.
  From this it follows that $a \cdot g^* + g \cdot a \leqqbkaexpl^{\htop} a$.
  To this we apply the other least-fixpoint axiom ($e+(f\cdot h)\leqqbkaexpl h \Rightarrow e f^* \cdot e \leqqbkaexpl h$), to conclude that $g^*\cdot a \cdot g^* \leqqbkaexpl^{\htop} a$.

  The next thing to prove is that $\closure[\htop]{\sembka{e}}=\sembka{\hat s(e)}$.
  This can be done by noting that $\semwpocka{e} = \sembka{\hat{s}(e)}$, and hence $\sembka{\hat{s}(e)} = \closure[\htop]{\sembka{\hat{s}(e)}}$ by \cref{lemma:closed-alt}.
  From \cref{lemma:soundness} and $e\equivbkaexpl^{\htop}\hat s(e)$ we then conclude that $\closure[\htop]{\sembka{e}}=\closure[\htop]{\sembka{\hat s(e)}}=\sembka{\hat s(e)}$.
\end{proof}

\fi%

\implicationaxioms*
\begin{proof}
  If $e \leq f \in \hcontr$, then $e = \alpha$ and $f = \alpha \cdot \alpha$ for some $\alpha \in \State$.
      We can then derive that $\alpha \equiv \alpha \wedge \alpha
      \leqq \alpha \cdot \alpha$, using the PCDL and the Interface axioms. Hence $e \leqq f$.
      If $e \leq f \in \htop$, then we distinguish two cases.
      \begin{enumerate}
        \item Let $e= \alpha \cdot c $ and $g=c$ for $c\in\Act$ and $\alpha\in\State$.
        Then we derive
        \[
        \alpha \cdot c \leqq \top \cdot c \leqq c \tag{PCDL and Interface axioms}
        \]
        The case where $e = c \cdot \alpha$ is similar.
        \item Let $e = \alpha \cdot \beta $ and $g=\beta$ for $\alpha,\beta\in\State$.
        Then we derive
        \[
        \alpha \cdot \beta \leqq\top \cdot \beta \leqq \beta \tag{PCDL and Interface axioms}
        \]
        The case where $e = \beta \cdot \alpha$ is similar. Hence, we can conclude that $e\leqq f$. \qedhere
      \end{enumerate}
\end{proof}

\ifarxiv%

\section{Proofs about guardedness}

We can establish the following three facts about paths.
\begin{lemma}\label{lemma:subpath}
  Let $U = [\lp{u}] \in \Pom(\Act\cup \State)$ and $u_1,u_2\in S_\lp{u}$.
  If $p_v$ is a path from $u_1$ to $u_2$ for $v\in \var$, then for any $z_1,z_2\in{q_1,\dots,q_{n+1}}$ such that $z_1\leq_\lp{u} z_2$, the action-nodes and state-nodes between $z_1$ and $z_2$ in $p_v$ form a path from $z_1$ to $z_2$ for $v$.
\end{lemma}
\begin{proof}
Let $a_k,\dots,a_m$ be the action-nodes between $z_1$ and $z_2$ on $p_v$. We know by construction that $z_1\leq_\lp{u}a_k\leq_\lp{u} a_{k+1}\leq_\lp{u}\cdots\leq_\lp{u} a_m\leq_\lp{u} z_2$. This verifies the first property of a path.
For the state-nodes we take the nodes $z_1$ and $z_2$ and the state-nodes between them in $p_v$. We denote these state-nodes with $q_k,\cdots q_{m+1}$. They satisfy~\ref{item:guards} by construction.
\end{proof}

The following notion is also useful.

\begin{definition}[Bottleneck]\label{def:bottleneck}
    Let $U = [\lp{u}] \in \Pom(\Act\cup \State)$ and $u_0,u_1,u_2\in S_\lp{u}$.
We say $u_1$ is a \emph{bottleneck between $u_0$ and $u_2$} if $u_0\leq_\lp{u} u_1 \leq_\lp{u} u_2$ and
for all $u_3\in S_\lp{u}$ s.t. $u_0\leq_\lp{u} u_3$ we have $u_1\leq_\lp{u} u_3$ or $u_3\leq_\lp{u} u_1$.
\end{definition}

\begin{lemma}\label{lemma:bottleneck}
    Let $U = [\lp{u}] \in \Pom(\Act\cup \State)$ and $u_1,u_2\in S_\lp{u}$ s.t.\ $u_1\leq_\lp{u} u_2$.
    If there exists a path $p_v$ from $u_1$ to $u_2$, and a bottleneck $u_3$ between them, then the bottleneck is on $p_v$.
\end{lemma}
\begin{proof}
Proof by contradiction. Suppose that $u_3$ is not on $p_v$. Take the biggest element of the path that is below $u_3$, denote it with $s_1$. If $\lambda_\lp{u}(s_1)\in \Act$, we know it has a unique successor state-node that is on $p_v$, denote that with $s_2$.
This means that $u_3<_\lp{u}s_2$ as $s_1$ was the biggest element below $u_3$ on $p_v$. This is a contradiction with the fact that $s_1\leq_\lp{u}u_3$ and $u_3$ not on $p_v$. Now suppose that $\lambda_\lp{u}(s_1)\in \State$.
If $u_2=s_1$, then we get that $u_3\leq_\lp{u} u_2 \leq_\lp{u} u_3$, which is a contradiction. Thus we know that on $p_v$ there must be a node with an action label after $s_1$, let us call it $s_3$ and its unique predecessor is $s_1$.
And thus we have $u_3 <_\lp{u} s_3$. From this we can conclude that $u_3\leq_\lp{u} s_1$ which can only be true if $u_3 = s_1$, which is a contradiction.
\end{proof}

\begin{lemma}\label{lemma:pathconcatenation}
    Let $U = [\lp{u}] \in \Pom(\Act\cup \State)$ and $u_0,u_1,u_2\in S_\lp{u}$ such that $u_0\leq_\lp{u} u_1\leq_\lp{u} u_2$.
    If there exists a path $p_v$ for $v$ from $u_0$ to $u_1$, and a path $s_v$ for $v$ from $u_1$ to $u_2$, then the path $t_v$ obtained by taking the union of $p_v$ and $s_v$ is a path for $v$ from $u_0$ to $u_2$.
\end{lemma}
\begin{proof}
  Let $a_1,\dots, a_n$ be the action-nodes in $p_v$ and
  $b_1,\dots, b_m$ the action nodes of $s_v$. The action-nodes of
  $t_v$ are then $a_1,\dots, a_n,b_1,\dots, b_m=c_1,\dots, c_{n+m}$. The first requirement of a path is then automatically
  satisfied for $t_v$, by construction of $\leq_\lp{u}$. Take
  $q_1,\dots, q_{n+1}$ to be the state-nodes in $p_v$ and
  $w_1,\dots, w_{m+1}$ the state-nodes in $s_v$ and
  let $q_1,\dots, q_{n+1},w_2,\dots, w_{m+1}$ be the state-nodes of $t_v$, which we will denote with $x_1,\dots, x_{n+m+1}$. Per definition we know that $q_{n+1}=w_1$.
  We need that the predecessor of $c_i$ is $x_{i}$ for $1\leq i\leq n+m$. If $1\leq i\leq n$, this condition is immediately satisfied. For $n+1<i\leq n+m$, we know that $c_i$ is $b_j$ for some $1\leq j\leq m$. Thus $i=j+n$.
  The predecessor of $b_j$ is $w_{j}$, which is $x_{j+n}=x_{i}$. For the case where $i=n+1$, we know that the predecessor of $c_{n+1}$ is the predecessor of $b_1$, which is $w_1=q_{n+1}$. We have that $x_{i}=x_{n+1}$, so we have obtained the desired result. We can give a similar argument why the successor of $c_i$ is $x_i$ for $1\leq i\leq n+m$. The other properties of a path can be verified in a similar manner.
\end{proof}

We introduce the following definition in order to aid in the proof of why~\ref{item:minmax}--\ref{item:prepath} hold in a pomset if and only if that pomset is \emph{guarded}.

\begin{definition}
Let $U = [\lp{u}]$ and $V = [\lp{v}]$ be pomsets.
We say that $U$ is a \emph{convex subpomset} of $V$ when there exists an injective function $h\colon S_\lp{u} \to S_\lp{v}$ such that all of the following hold:
\begin{enumerate}[label={(\roman*)},leftmargin=1cm]
    \item
    Labels are partially preserved. We have for $u\in S_\lp{u}$, if $\lambda_\lp{u}(u)\in\Act$ or $\lambda_\lp{v}(h(u))\in\Act$ then $\lambda_\lp{v} \circ h (u) = \lambda_\lp{u}$, and $\lambda_\lp{u}(u)\in\State$ if and only if $\lambda_\lp{v} \circ h(u)\in\State$.

    \item
    Order is preserved and reflected, i.e., for $u_0, u_1 \in S_\lp{u}$ we have $u_0 \leq_\lp{u} u_1$ if and only if $h(u_0) \leq_\lp{v} h(u_1)$.

    \item
    $h$ is convex, i.e., if $u_0, u_1 \in S_\lp{u}$ and $v \in S_\lp{v}$ with $h(u_0) \leq_\lp{v} v \leq_\lp{v} h(u_1)$, then $h(u) = v$ for some $u \in S_\lp{u}$.
\end{enumerate}
\end{definition}

\begin{lemma}\label{lemma:succandpred}
Let $U,V\in\Pom(\Act\cup\State)$ and $u_0,u_1\in S_\lp{u}$. Assume $U$ is a convex subpomset of $V$ such that $h$ maps the minimum of $U$ to the minimum of $V$ and same for the maximum. If $u_1$ is the successor or predecessor of $u_0$, then $h(u_1)$ is the successor, resp.\ predecessor, of $h(u_0)$ in $V$.
\end{lemma}
\begin{proof}
If $u_1$ is the successor of $u_0$ in $U$, we know that $u_0\leq_{\lp{u}} u_1$. By convexity, we obtain $h(u_0)\leq_{\lp{v}}h(u_1)$. Now suppose that there exists $u_2\in S_{\lp{v}}$ such that $h(u_0) \leq_{\lp{v}} u_2$.
As we know that $h(*_{\max})$ is the maximum of $V$, we know that $h(u_0) \leq_{\lp{v}} u_2 \leq_{\lp{v}} h(*_{\max})$. By convexity, we obtain $u_2'\in S_{\lp{u}}$ such that $h(u_2')=u_2$ and $u_0 \leq_\lp{u} u_2'$.
As $u_1$ is the successor of $u_0$, we get that $u_1\leq_{\lp{u}} u_2'$.
By convexity, this means that $h(u_1)\leq_{\lp{v}} h(u_2')$. As $h(u_2')=u_2$, we have obtained the desired result.
The case where $u_1$ is the predecessor of $u_0$ in $U$ is proved analogously.
\end{proof}

\begin{lemma}\label{lemma:succandpredtwo}
Let $U,V\in\Pom(\Act\cup\State)$ s.t. $U$ is a convex subpomset of $V$. If $h(u_1)$ is the successor, resp.\ predecessor, of $h(u_0)$ in $V$, then $u_1$ is the successor, resp.\ predecessor, of $u_0$ in $U$.
\end{lemma}
\begin{proof}
If $h(u_1)$ is the successor of $h(u_0)$ in $V$, we know that $h(u_0)\leq_{\lp{v}} h(u_1)$. Then, by convexity, we obtain $u_0\leq_{\lp{u}}u_1$. Now suppose that there exists $u_2\in S_{\lp{u}}$ such that $u_0 \leq_{\lp{u}} u_2$.
Thus we know that $h(u_0) \leq_{\lp{v}} h(u_2)$. As $h(u_1)$ is the successor of $h(u_0)$, we obtain $h(u_1)\leq_{\lp{v}} h(u_2)$ and subsequently that $u_1 \leq_\lp{u} u_2$.
For the case where $h(u_1)$ is the predecessor of $h(u_0)$ in $V$, the proof is analogous.
\end{proof}

\begin{lemma}%
\label{lemma:subpomset}
Let $U, V \in \Pom(\Act \cup \State)$.
If $U$ is a convex subpomset of $V$ with a unique state-labelled minimum and maximum, and $V$ obeys~\ref{item:minmax}--\ref{item:succpred}, then so does $U$.
\end{lemma}
\begin{proof}
Since $U$ has a unique state-labelled minimum and maximum, it suffices to verify~\ref{item:alternating} and~\ref{item:succpred}.
Let $U = [\lp{u}]$ and $V = [\lp{v}]$, and let $U$ be a convex subpomset of $V$ with $h\colon S_\lp{u} \to S_\lp{v}$.
\begin{enumerate}[label={(A\arabic*)},leftmargin=1cm]
    \setcounter{enumi}{1}
    \item
    Suppose $u_0, u_1 \in S_\lp{u}$ such that $u_0 <_\lp{u} u_1$ and $\lambda_\lp{u} (u_0),\lambda_\lp{u} (u_1)\in\State$; in that case, we know that $h(u_0) <_\lp{v} h(u_0)$, since $h$ is order-preserving and injective.
    Since $V$ is has the property of~\ref{item:alternating}, we obtain $v \in S_\lp{v}$ such that $h(u_0) <_\lp{v} v <_\lp{v} h(u_1)$ and $\lambda_\lp{v}(v) \in \Act$.
    Because $h$ is convex, we obtain $u \in S_\lp{u}$ such that $h(u) = v$, which tells us that $u_0 <_\lp{u} u <_\lp{u} u_1$.
    Since $\lambda_\lp{v}(h(u)) = \lambda_\lp{v}(v) \in \Act$, we get $\lambda_\lp{u}(u) \in\Act$, and the condition is satisfied.

    \item
    Suppose $u_0 \in S_\lp{u}$ with $\lambda_\lp{u}(u_0) \in \Act$; hence, we know that $h(u_0) \in S_\lp{v}$ with $\lambda_\lp{u}(h(u_0)) \in \Act$.
    Since $V$ obeys~\ref{item:succpred}, we obtain $v \in S_\lp{v}$ such that $\lambda_\lp{v}(v) \in \State$, and $v$ is the predecessor of $h(u_0)$ in $S_\lp{v}$.
    Let $*_{\min}$ be the unique state-labelled $\leq_\lp{u}$-minimum of $S_\lp{u}$, which exists by the premise.
    We then know that $*_{\min} <_\lp{u} u_0$, and hence $h(*_{\min}) <_\lp{v} h(u_0)$.
    Since $v$ is the predecessor of $h(u_0)$, it follows that $h(*_{\min}) \leq_\lp{v} v \leq_\lp{v} h(u_0)$; by convexity of $h$ we obtain $u_1 \in S_\lp{u}$ such that $h(u_1) = v$.
    It is easily seen that $u_1$ is state-labelled.
    To see that $u_1$ is the predecessor of $u_0$ in $S_\lp{u}$, first note that since $h(u_1) <_\lp{v} h(u_0)$ and $h$ is order-reflecting, it follows that $u_1 <_\lp{u} u_0$.
    Moreover, if $u_1' <_\lp{u} u_0$, then $h(u_1') <_\lp{v} h(u_0)$, and hence $h(u_1') \leq_\lp{v} h(u_1)$, meaning $u_1' \leq_\lp{u} u_1$.
    \qedhere
\end{enumerate}
\end{proof}

\fromguardedtocausal*
\begin{proof}
  We first prove the implication from left to right.
Note that guarded pomsets are series-parallel by construction.
To show that guarded pomsets satisfy~\ref{item:minmax}--\ref{item:prepath}, we proceed by induction on the construction of $\G$.
In the base, where $U = \alpha$ for $\alpha \in \State$ or $U = \alpha \cdot \ltr{a} \cdot \alpha[\ltr{a}]$ for $\alpha \in \State$ and $\ltr{a} \in \Act$, the requirements hold immediately.
In the inductive step, there are two cases, and for each of these cases we need to verify all seven properties.
\begin{itemize}
    \item
    If $U = V \cdot \alpha \cdot W$ such that $[\lp{v}] = V \cdot \alpha$ and $[\lp{w}] = \alpha \cdot W$ are guarded, then $V \cdot \alpha$ and $\alpha \cdot W$ satisfy all seven properties by induction.
    Without loss of generality, we can assume that $S_\lp{u} = S_\lp{v} \cup S_\lp{w}$ and
    ${\leq_\lp{u}} = {\leq_{\lp{v}}} \cup {\leq_{\lp{w}}} \cup S_\lp{v} \times S_\lp{w}$ and
    $\lambda_\lp{u} = \lambda_\lp{v} \cup \lambda_\lp{w}$, with $S_\lp{v} \cap S_\lp{w} = \{ * \}$ such that $*$ is the node labelled by $\alpha$ in $U$.
    In particular, this means that for all $v \in S_\lp{v}$ and $w \in S_\lp{w}$ we have that $v \leq_\lp{v} * \leq_\lp{w} w$.
    \begin{enumerate}[label={(A\arabic*)},leftmargin=1cm]
        \item
        The minimum node of $[\lp{v}]$ is also the minimum of $[\lp{u}]$, and similarly the maximum node of $[\lp{w}]$ is the maximum node of $[\lp{u}]$; since these nodes are state-labelled by induction, the condition holds.

        \item
        Let $u_0, u_1 \in S_\lp{u}$ be such that $\lambda_\lp{u}(u_0), \lambda_\lp{u}(u_1) \in \State$ and $u_0 <_\lp{u} u_1$.
        If $u_0, u_1 \in S_\lp{v}$, then $u_0 <_\lp{v} u_1$ and (since $[\lp{v}]$ satisfies~\ref{item:alternating}) we find $u_2 \in S_\lp{v}$ such that $u_0 <_\lp{v} u_2 <_\lp{v} u_1$ and $\lambda_\lp{v}(u_2) \in \Act$.
        From this, it follows that $u_2 \in S_\lp{u}$ such that $u_0 <_\lp{u} u_2 <_\lp{u} u_1$ and $\lambda_\lp{u}(u_2) \in \Act$.
        The case where $u_0, u_1 \in S_\lp{w}$ is similar.

        We are left with the case where $u_0 \in S_\lp{v} \setminus \{ * \}$ and $u_1 \in S_\lp{w} \setminus \{ * \}$.
        Since $[\lp{v}]$ satisfies~\ref{item:alternating} and $u_0 <_\lp{v} *$, we obtain $u_2 \in S_\lp{v}$ such that $u_0 <_\lp{v} u_2 <_\lp{v} *$ and $\lambda_\lp{v}(u_2) \in \Act$.
        Because $* <_\lp{w} u_1$, it then follows that $u_0 <_\lp{u} u_2 <_\lp{u} * <_\lp{u} u_1$ and $\lambda_\lp{u}(u_2) \in \Act$.

        \item
        Let $u_0 \in S_\lp{u}$ such that $\lambda_\lp{u}(u_0) \in \Act$.
        If $u_0 \in S_\lp{v}$, then by induction there exists a $u_1 \in S_\lp{v}$ such that $\lambda_\lp{v}(u_1) \in \State$ and $u_1$ is the predecessor of $u_0$ in $[\lp{v}]$.
        Since $\lambda_\lp{u}(u_1) \in \State$, it remains to prove that $u_1$ is the predecessor of $u_0$ in $[\lp{u}]$ as well. To see this, note that $u_1 <_\lp{u} u_0$, and if $u_1' \in S_\lp{u}$ such that $u_1' <_\lp{u} u_0$, then $u_1' \in S_\lp{v}$ by
        definition of $[\lp{u}]$, and hence also $u_1' <_\lp{v} u_0$,
        meaning that $u_1' \leq_\lp{v} u_1$, and thus $u_1' \leq_\lp{u} u_1$.
To show that $u_0$ has a state-node as successor as well, we can argue similarly.
        The case where $u_0 \in S_\lp{w}$ is analogous.

        \item
        If we take $u\in S_\lp{w}$, we know immediately that for every $v\in\dom(\lambda_\lp{u}(u))$ there is a path from $u$ to $*_{\max}$ by the induction hypothesis, where $*_{\max}$ is the maximum node of $[\lp{u}]$. Thus we only consider the case where $u\in S_\lp{v}\setminus\{*\}$.
        From the induction hypothesis we obtain a path $p_v$ for $v$ from $u$ to $*$ in $V$, and a path $q_v$ for $v$ from $*$ to $*_{\max}$ in $W$.
        Note that $*_{\max}$ in $W$ is the same node as $*_{\max}$ in $U$. Using \cref{lemma:pathconcatenation}, we obtain a path for $v$ from $u$ to $*_{\max}$.

        \item
        This item follows immediately from the induction hypothesis.

        \item This item also follows immediately from the induction hypothesis.

        \item
        If we take $u\in S_\lp{v}$, we know immediately that the condition is satisfied by the induction hypothesis. Thus we only consider the case where $u\in S_\lp{w}\setminus\{*\}$. Take $v\in\dom(\lambda_\lp{u}(u))$.
        From the induction hypothesis we obtain a path $p_v$ for $v$ from $s_0\in S_\lp{w}$ to $u$ in $[\lp{w}]$ such that $s_0=*$, or $s_0$ is the successor of an assignment-node with label $v\leftarrow k$ with $k\in \var\cup \val$. In the latter case, we are done immediately. In the case where $s_0=*$, we know this means that $v\in\dom(\lambda_\lp{v}(*))$.
        From our induction hypothesis we then obtain a path $q_v$ from $s_1\in S_\lp{v}$ such that $s_1=*_{\min}$, or $s_1$ is the successor of an assignment-node with label $v\leftarrow k$ with $k\in \var\cup \val$. The union of $q_v$ and $p_v$ is then a path for $v$ from $s_1$ to $u$ in $U$ such that $s_1=*_{\min}$, or $s_1$ is
        the successor of an assignment-node with label $v\leftarrow k$ with $k\in \var\cup \val$ (\cref{lemma:pathconcatenation}).
    \end{enumerate}

    \item
    If $U = \alpha\oplus\gamma \cdot (V \parallel W) \cdot \beta\oplus\delta$ such that $[\lp{v}] =
    \alpha \cdot V \cdot \beta$ and $[\lp{w}] = \gamma \cdot W \cdot \delta$ are guarded, then $[\lp{v}]$
    and $[\lp{w}]$ satisfy~\ref{item:minmax}--\ref{item:prepath} by induction.
    Without loss of generality, we assume that $S_\lp{u} = S_\lp{v} \cup S_\lp{w}$ and ${\leq_\lp{u}} =
    {\leq_\lp{v}} \cup {\leq_\lp{w}}$,
    with $S_\lp{v} \cap S_\lp{w} = \{ *_{\min}, *_{\max} \}$, with $*_{\min}$ the node with label
    $\alpha$ in $[\lp{v}]$ and label $\gamma$ in $[\lp{w}]$ and $*_{\max}$ the node with label $\beta$ in
    $[\lp{v}]$ and label $\delta$ in $[\lp{w}]$. We also have $\lambda_\lp{u}(*_{\min})=\alpha\oplus\gamma$
    and $\lambda_\lp{u}(*_{\max})=\beta\oplus\delta$ and for all nodes $x\in S_\lp{u}\setminus
    \{*_{\min},*_{\max}\}$ have $\lambda_\lp{u}(x) = \lambda_\lp{v}(x)$
    if $x\in S_\lp{v}$ and $\lambda_\lp{w}(x)$ if $x\in S_\lp{w}$.

    \begin{enumerate}[label={(A\arabic*)},leftmargin=1cm]
        \item
        To see that $*_{\min}$ is the minimum of $[\lp{u}]$, note that if $u \in S_\lp{u}$ then either $u \in S_\lp{v}$ or $u \in S_\lp{w}$.
        In the former case, we know that $*_{\min} \leq_\lp{v} s$, and hence $*_{\min} \leq_\lp{u} s$.
        The case where $u \in S_\lp{w}$ can be argued similarly.
        We can show that $*_{\max}$ is the maximum of $[\lp{u}]$ analogously.
        Lastly, we note that $\lambda_\lp{u}(*_{\min}), \lambda_\lp{u}(*_{\max}) \in \State$.

        \item
        Let $u_0, u_1 \in S_\lp{u}$ be such that $\lambda_\lp{u}(u_0), \lambda_\lp{u}(u_1) \in \State$ and $u_0 <_\lp{u} u_1$.
        By definition of $\leq_\lp{u}$, we have that either $u_0 <_\lp{v} u_1$ or $u_0 <_\lp{w} u_1$. W.l.o.g.\ we assume $u_0 <_\lp{v} u_1$. We obtain $u_2 \in S_\lp{v}$ such that $u_0
        <_\lp{v} u_2 <_\lp{v} u_1$ and $\lambda_\lp{v}(u_2) \in \Act$.
        It then follows that $u_2 \in S_\lp{u}$ such that $u_0 <_\lp{u} u_2 <_\lp{u} u_1$ and $\lambda_\lp{u}(u_2) \in \Act$.

        \item
        Let $u_0 \in S_\lp{u}$ be such that $\lambda_\lp{u}(u_0) \in \Act$.
        W.l.o.g.\ let $u_0 \in S_\lp{v}$. Hence, we find $u_1 \in S_\lp{v}$ such that $\lambda_\lp{v}(u_1) \in \State$ and $u_1$ is the predecessor of $u_0$ in $[\lp{v}]$ from the induction hypothesis.
        Since $u_1 \in S_\lp{u}$ and $\lambda_\lp{u}(u_1) \in \State$ automatically, it remains to show that $u_1$ is the predecessor of $u_0$ in $[\lp{u}]$ as well. This follows immediately from \cref{lemma:succandpred}. That $u_0$ has a state-node as its successor is proved analogously.

        \item Let $u \in S_\lp{u}$ be such that $\lambda_\lp{u}(u)\in \State$. Take $v\in\dom(\lambda_\lp{u}(u))$. W.l.o.g.\ assume that $u \in S_\lp{v}$. If $u\neq *_{\min}$ and $u\neq *_{\max}$,
        then we
        find a path for $v$ from $u$ to $*_{\max}$ in $[\lp{v}]$. This immediately is also a path for $v$ from $u_0$ to $*_{\max}$ in $[\lp{u}]$ as $\beta\subseteq\beta\oplus\delta$.
        If $u = *_{\max}$, there is a trivial path for $v$ from $*_{\max}$ to $*_{\max}$. If $u = *_{\min}$, we can assume without loss of generality that $v\in\dom(\alpha)$, and use our induction hypothesis to obtain a path for $v$ in a way similar to the first case.

        \item Let $u \in S_\lp{u}$ be such that $\lambda_\lp{u}(u)= v\leftarrow n$ for some $v\in \var$ and $n\in \val$. W.l.o.g.\ we assume $u \in S_\lp{v}$. We find $s \in S_\lp{v}$ such
        that $\lambda_\lp{v}(s) \in \State$, $s$ is the successor of $u$ in $[\lp{v}]$ and $\lambda_\lp{v}(s)(v)=n$. We immediately obtain $s \in S_\lp{u}$ and $\lambda_\lp{u}(s)(v)=n$.
        By \cref{lemma:succandpred}, using the identity function as witness to show that $[\lp{v}]$ is a convex subpomset of $[\lp{u}]$, we get that $s$ is the successor of $u$ in $[\lp{u}]$, and we are done.

        \item Let $u \in S_\lp{u}$ be such that $\lambda_\lp{u}(u)= v\leftarrow v'$ for some $v,v'\in \var$. If $u \in S_\lp{v}$, then we find $p,s \in S_\lp{v}$ such
        that $\lambda_\lp{v}(p),\lambda_\lp{v}(s) \in \State$,
        $p$ is the predecessor of $u$ in $[\lp{v}]$, $s$ is the successor of $u$ in $[\lp{v}]$, $v'\in\dom(\lambda_\lp{v}(p))$ and $\lambda_\lp{v}(s)(v)=\lambda_\lp{v}(s)(v')=\lambda_\lp{v}(p)(v')$.
        We know that $p,s \in S_\lp{u}$ and $\lambda_\lp{u}(s)(v)=\lambda_\lp{u}(s)(v')=\lambda_\lp{u}(p)(v')$ immediately, as $\alpha\subseteq\alpha\oplus\gamma$ and $\beta\subseteq\beta\oplus\delta$ and for all other nodes the labels are the same.
        From \cref{lemma:succandpred}, using the identity function as witness to show that $[\lp{v}]$ is a convex subpomset of $[\lp{u}]$, we get that $s$ is the successor of $u$ in $[\lp{u}]$, and $p$ is the predecessor of $u$ in $[\lp{u}]$, and we are done.
        The case where $u \in S_\lp{w}$ is similar.

        \item Let $u \in S_\lp{u}$ be s.t. $\lambda_\lp{u}(u)\in \State$. Take
        $v\in\dom(\lambda_\lp{u}(u))$. W.l.o.g.\ let $u \in S_\lp{v}$. If $u\neq *_{\min}$ and $u\neq
        *_{\max}$, we
        find a path for $v$ from $s\in S_\lp{v}$ to $u$ such that $s= *_{\min}$ or $s$ is the
        successor of an assignment-node $a\in S_\lp{v}$ with label $v\leftarrow k$ for $k\in \var\cup \val$.
        If $s= *_{\min}$, as $\alpha\subseteq\alpha\oplus\gamma$, the path is also a path from $s$ to
        $u$ in $[\lp{u}]$. If $s \neq *_{\min}$, then we know that $a\in S_\lp{u}$ with label
        $v\leftarrow k$ and $s\in S_\lp{u}$. From \cref{lemma:succandpred} we know that $s$ is the
        successor of $a$ in $S_\lp{u}$. We can also immediately infer that the path for $v$ from $s$
        to $u$ in $[\lp{v}]$ is also a path for $v$ from $s$ to $u$ in $[\lp{u}]$.
        The case where $u = *_{\min}$ is trivial. If $u =
        *_{\max}$, we can assume w.l.o.g.\ that $v\in\dom(\beta)$, and use our induction
        hypothesis to obtain a path for $v$ similar to the first case.
        \qedhere
    \end{enumerate}
    \end{itemize}
    For the implication from right to left, we proceed by induction on the size of $U$; this is well-founded, because $U$ is series-parallel and therefore finite.
Specifically, our induction hypothesis is that if $V$ satisfies~\ref{item:minmax}--\ref{item:prepath} and is strictly smaller than $U$, then $V$ is guarded.

Since $U$ satisfies~\ref{item:minmax}, there exist $*_{\min}, *_{\max} \in S_\lp{u}$, respectively the $\leq_\lp{u}$-minimum and $\leq_\lp{u}$-maximum of $S_\lp{u}$, such that $\lambda_\lp{u}(*_{\min}) = \alpha \in \State$ and $\lambda_\lp{u}(*_{\max}) = \beta \in \State$.
If $*_{\min} = *_{\max}$, we know that $U = \alpha = \beta$, which makes $U$ guarded.
Otherwise, we can write $U = \alpha \cdot V \cdot \beta$ for some pomset $V$, and note that $V$ is series-parallel because $U$ is --- after all, if $U$ does not contain an $\mathsf{N}$-shape, then neither does $V$.
This gives us four cases to consider.
\begin{itemize}
    \item
    If $V = 1$, we reach a contradiction, for then $U = \alpha \cdot \beta$, which fails~\ref{item:alternating}, as there is no action in between the state-labelled nodes $\alpha$ and $\beta$.
    We can therefore exclude this case.

    \item
    If $V = \ltr{a}$ for some $\ltr{a} \in \Act \cup \State$, we exclude the possibility that $\ltr{a} = \gamma \in \State$, for then $U$ would again fail~\ref{item:alternating}.
    Hence, $\ltr{a} \in \Act$. We verify that $\beta=\alpha[\ltr{a}]$. We have two cases. Suppose first that $a = v\leftarrow n$ for some $n\in \val$. Take $v'\in \var$. Then if $v=v'$, we know that $\beta(v')=n$, because of~\ref{item:effect}. So then $\beta(v')=\alpha[\ltr{a}](v')$.
    If $v\neq v'$, then we distinguish two cases again. If $v'\in\dom(\alpha)$, then as $U$ satisfies~\ref{item:path}, we obtain $\beta(v')=\alpha(v')$. Hence, $\beta(v')=\alpha[\ltr{a}](v')$. If
    $v'\notin\dom(\alpha)$, we do a proof by contradiction. Suppose that $v'\in\dom(\beta)$. Then because $U$
    satisfies~\ref{item:prepath}, as $\ltr{a}$ does not change the value of $v'$, we know that there is a path from $\alpha$ to $\beta$ for $v'$. Thus we have that $v'\in\dom(\alpha)$, which is a contradiction. Hence, $v'\notin\dom(\beta)$.
    We can conclude that $\beta=\alpha[\ltr{a}]$.

    For the second case, suppose that $\ltr{a} = v\leftarrow v'$ for some $v'\in \var$. From~\ref{item:annoyingassignments}, we establish that $v'\in\dom(\alpha)$, and $\beta(v)=\alpha(v')=\beta(v')$.
    Hence, $\alpha[\ltr{a}]$ exists. Take $v''\in \var$. If $v=v''$, we know that $\beta(v'')=\alpha(v')$, and thus $\beta(v'')=\alpha[\ltr{a}](v'')$.
    If $v\neq v''$, then we distinguish two cases. If $v''\in\dom(\alpha)$, as $U$ satisfies~\ref{item:path}, we obtain $\beta(v'')=\alpha(v'')$. Hence, $\beta(v'')=\alpha[\ltr{a}](v'')$. We can exclude the case where
    $v''\notin\dom(\alpha)$ similarly as above.
    Hence, $\beta=\alpha[\ltr{a}]$ and $\alpha[\ltr{a}]$ exists, which makes $U = \alpha \cdot \ltr{a} \cdot \beta$ a guarded pomset.

    \item
    Suppose that $V = V_0 \cdot V_1$ for some non-empty series-parallel pomsets $V_0$ and $V_1$.
    We write $[\lp{v}_0] = \alpha \cdot V_0$ and $[\lp{v}_1] = V_1 \cdot \beta$.
    Without loss of generality, we assume that $S_{\lp{v}_0}$ and $S_{\lp{v}_1}$ are disjoint, and all of the following hold:
    \begin{mathpar}
        S_\lp{u} = S_{\lp{v}_0} \cup S_{\lp{v}_1}
        \and
        \leq_\lp{u} = {\leq_{\lp{v}_0}} \cup {\leq_{\lp{v}_1}} \cup S_{\lp{v}_0} \times S_{\lp{v}_1}
        \and
        \lambda_\lp{u}(u) =
            \begin{cases}
            \lambda_{\lp{v}_0}(u) & u \in S_{\lp{v}_0} \\
            \lambda_{\lp{v}_1}(u) & u \in S_{\lp{v}_1} \\
            \end{cases}
    \end{mathpar}

    Let $T_0 \subseteq S_{\lp{v}_0}$ be the set of $\leq_{\lp{v}_0}$-maxima of $S_{\lp{v}_0}$, and let $T_1 \subseteq S_{\lp{v}_1}$ be the set of $\leq_{\lp{v}_1}$-minima of $S_{\lp{v}_1}$; note that these sets are non-empty, because $V_0$ and $V_1$ are non-empty.
    We proceed to make the following observations about the labels of nodes in $T_0$ and $T_1$:
    \begin{itemize}
        \item
        Suppose that $v_0 \in T_0$ and $v_1 \in T_1$ such that $\lambda_{\lp{v}_0}(v_0) \in \State$ and $\lambda_{\lp{v}_1}(v_1) \in \State$.
        In that case, we know that $v_0 <_\lp{u} v_1$. Since $U$ satisfies~\ref{item:alternating}, we obtain a $v_2 \in S_{\lp{u}}$ such that $v_0 <_\lp{u} v_2 <_\lp{u} v_1$ and $\lambda_\lp{u}(v_2) \in \Act$.
        We have now reached a contradiction, for if $v_2 \in S_{\lp{v}_0}$ then $v_0$ is not a $\leq_{\lp{v}_0}$-maximum of $S_{\lp{v}_0}$, and if $v_2 \in S_{\lp{v}_1}$ then $v_1$ is not a $\leq_{\lp{v}_1}$-minimum of $S_{\lp{v}_1}$.
        Thus we can conclude that \emph{at most} one of $T_0, T_1$ contains a node labelled by an atom.

        \item
        By the above and the fact that $T_0$ and $T_1$ are non-empty, it follows that at least one of $T_0, T_1$ contains a node labelled by a letter.
        For instance, suppose $v_0 \in T_0$ such that $\lambda_{\lp{v}_0}(v_0) \in \Act$.
        In that case, since $U$ satisfies~\ref{item:succpred}, we find that there exists a $v_1 \in S_\lp{u}$ such that $v_1$ is the successor of $v_0$, and $\lambda_\lp{u}(v_1) \in \State$.
        We know that $v_1 \not\in S_{\lp{v}_0}$, for then $v_0$ would not be $\leq_{\lp{v}_0}$-maximal in $S_{\lp{v}_0}$, and thus $v_1 \in S_{\lp{v}_1}$.
        Furthermore, if $v_1' \in S_{\lp{v}_1}$, then $v_0 <_\lp{u} v_1'$, and hence $v_1 \leq_\lp{u} v_1'$, meaning $v_1 \leq_{\lp{v}_1} v_1'$.
        This means that $v_1$ is the unique $\leq_{\lp{v}_1}$-minimum of $S_{\lp{v}_1}$, and hence $T_1 = \{ v_1 \}$.
        A similar analysis applies when $v_1 \in T_1$ such that $\lambda_{\lp{v}_1}(v_1) \in \Act$.
    \end{itemize}

    From the above, we learn that either $T_0$ is a singleton whose only element is state-labelled, and all elements of $T_1$ are action-labelled, or vice versa.
    For the remainder, we assume the former (the dual is argued similarly).
    Let $T_0 = \{ * \}$, and choose $\gamma = \lambda_\lp{u}(*)$; we claim that $[\lp{v}_0]$ and $\gamma \cdot [\lp{v}_1]$ satisfy~\ref{item:minmax}--\ref{item:prepath} and are smaller than $U$.

    \begin{enumerate}[label={(A\arabic*)},leftmargin=1cm]
      \setcounter{enumi}{2}
      \item
      For $[\lp{v}_0]$, we first recall that $[\lp{v}_0] = \alpha \cdot V_0$, and hence $[\lp{v}_0]$ contains a unique state-labelled minimum.
      Since $T_0 = \{ * \}$ and $\lambda_{\lp{v}_0}(*) = \gamma \in \State$, we know that $[\lp{v}_0]$ also contains a unique state-labelled maximum.
      Moreover, since $U = [\lp{v}_0] \cdot [\lp{v}_1]$, it is straightforward to show that $[\lp{v}_0]$ is a convex subpomset of $U$; by \cref{lemma:subpomset},
      it then follows that $[\lp{v}_0]$ satisfies~\ref{item:alternating} and~\ref{item:succpred}.

      \item
      Take $u\in S_{\lp{v}_0}$ such that $\lambda_{\lp{v}_0}(u)\in\State$ and $v\in\dom(\lambda_{\lp{v}_0}(u))$.
      We immediately obtain $u\in S_{\lp{u}}$, $\lambda_{\lp{u}}(u)\in\State$ and $v\in\dom(\lambda_{\lp{u}}(u))$.
      As $[\lp{u}]$ satisfies~\ref{item:path}, there exists a path $p_v$ for $v$ from $u$ to $*_{\max}$ in $[\lp{u}]$.
      Since $U = [\lp{v}_0] \cdot [\lp{v}_1]$ and $*$ is the local $\leq_{\lp{v}_0}$-maximum of $[\lp{v}_0]$, we know that $u\leq_\lp{u} * \leq_\lp{u} *_{\max}$. Suppose there exists $w\in S_\lp{u}$ such that $u\leq_\lp{u} w$.
      If $w\in S_{\lp{v}_0}$, then $w\leq_\lp{u} *$. If $w\in S_{\lp{v}_1}$, then $*\leq_\lp{u} w$.
      Hence, via \cref{lemma:bottleneck} we can conclude that $*$ is on path $p_v$.
      From \cref{lemma:subpath} we then know that there exists a path $q_v$ for $v$ from $u$ to $*$ in $[\lp{u}]$. As these are all nodes that also occur in $S_{\lp{v}_0}$, $q_v$ is also a path for $v$ from $u$ to $*$ in $[\lp{v}_0]$.
      As $*$ is the unique state-labelled maximum of $[\lp{v}_0]$, this proves that $[\lp{v}_0]$ satisfies~\ref{item:path}.

      \item
      Take $u\in S_{\lp{v}_0}$ such that $\lambda_{\lp{v}_0}(u)= v\leftarrow n$ for some $v\in \var$ and $n\in \val$. Hence, $u\in S_{\lp{u}}$ such that $\lambda_{\lp{u}}(u)= v\leftarrow n$.
      As $[\lp{u}]$ satisfies~\ref{item:effect}, we obtain a node $s\in S_\lp{u}$ such that $s$ is the successor of $u$ and $\lambda_{\lp{u}}(s)(v)=n$. Because $[\lp{v}_0]$ has a maximum state-labelled node, we know that the
      successor of $u$ in $U=[\lp{v}_0] \cdot [\lp{v}_1]$ is in fact also a node in $S_{\lp{v}_0}$.
      Thus we have $s\in S_{\lp{v}_0}$ such that $\lambda_{\lp{v}_0}(s)(v)=n$.
      From the fact that $[\lp{v}_0]$ is a convex subpomset of $[\lp{u}]$ using the identity function and \cref{lemma:succandpredtwo}, we can conclude that $s$ is a successor of $u$ in $[\lp{v}_0]$.
      Hence, $[\lp{v}_0]$ satisfies~\ref{item:effect}.

      \item
      Take $u\in S_{\lp{v}_0}$ such that $\lambda_{\lp{v}_0}(u)= v\leftarrow v'$ for some $v,v'\in \var$. Hence, $u\in S_{\lp{u}}$ such that $\lambda_{\lp{u}}(u)= v\leftarrow v'$.
      Because $[\lp{u}]$ satisfies~\ref{item:annoyingassignments}, we obtain nodes $p,s\in S_\lp{u}$ such that $p$ is the
      predecessor of $u$, $v'\in\dom(\lambda_\lp{u}(p))$, $s$ is the successor of $u$ and
      $\lambda_\lp{u}(s)(v)=\lambda_\lp{u}(s)(v')=\lambda_\lp{u}(p)(v')$.
      We know immediately that $p\in S_{\lp{v}_0}$ with the same label as in $[\lp{u}]$. Because $[\lp{v}_0]$ has a maximum state-labelled node, we know that the successor of $u$ in $U=[\lp{v}_0] \cdot [\lp{v}_1]$ is also a node in $S_{\lp{v}_0}$.
      Hence, $s\in S_{\lp{v}_0}$ such that $\lambda_{\lp{v}_0}(s)(v)=\lambda_{\lp{v}_0}(s)(v') = \lambda_{\lp{v}_0}(p)(v')$.
      From \cref{lemma:succandpredtwo}, using the identity function as a witness to show that $[\lp{v}_0]$ is a convex subpomset of $[\lp{u}]$, we conclude that $s$ and $p$ are the successor and predecessor of $u$ in $[\lp{v}_0]$.
      Thus $[\lp{v}_0]$ satisfies~\ref{item:annoyingassignments}.

      \item
      If $u\in S_{\lp{v}_0}$ such that $\lambda_{\lp{v}_0}(u)\in\State$ and $v\in\dom(\lambda_{\lp{v}_0}(u))$, we immediately obtain $u\in S_{\lp{u}}$, $\lambda_{\lp{u}}(u)\in\State$ and
      $v\in\dom(\lambda_{\lp{u}}(u))$.
      As $[\lp{u}]$ satisfies~\ref{item:prepath}, there exists a path $p_v$ for $v$ from $s\in S_\lp{u}$ to $u$ such that either $s=*_{\min}$ or $s$ is the successor of an assignment-node $a$ with label $v\leftarrow k$ with $k\in \var\cup \val$.
      If $s=*_{\min}$, then, as the minimal node of $[\lp{u}]$ is the same as the minimal node of $[\lp{v}_0]$ and all nodes on path $p_v$ also occur in $S_{\lp{v}_0}$, we can conclude that $p_v$ is a path from the unique minimum of $[\lp{v}_0]$ to $u$ for $v$ in $[\lp{v}_0]$.
      If $s\neq *_{\min}$, then we know that $a \leq_\lp{u} s \leq_\lp{u} u$, and thus that $s,a\in S_{\lp{v}_0}$. This means that the path $p_v$ exists entirely out of nodes that are also in $S_{\lp{v}_0}$. Hence we have a path in $[\lp{v}_0]$ for $v$
      from $s$ to $u$. From \cref{lemma:succandpredtwo}, we know that $s$ is the successor of $a$ in $[\lp{v}_0]$. Hence, $[\lp{v}_0]$ satisfies~\ref{item:prepath}.
    \end{enumerate}

      Furthermore, since $U = [\lp{v}_0] \cdot [\lp{v}_1]$ and $[\lp{v}_1]$ is non-empty, we know that $[\lp{v}_0]$ is smaller than $U$.

      \begin{enumerate}[label={(A\arabic*)},leftmargin=1cm]
      \setcounter{enumi}{2}
        \item
      For $\gamma \cdot [\lp{v}_1]$, we first recall that $[\lp{v}_1] = V_1 \cdot \beta$, and hence $\gamma \cdot [\lp{v}_1]$ contains a unique state-labelled maximum.
      Furthermore, $\gamma \cdot [\lp{v}_1]$ contains a unique state-labelled minimum by construction, as well.
      Since $[\lp{v}_0]$ has a unique maximum labelled by $\gamma$, we can write $[\lp{v}_0] = W \cdot \gamma$ for some pomset $W$; moreover, $W$ must be non-empty, for otherwise $\alpha \cdot V_0 = [\lp{v}_0] = \gamma$, meaning $V_0$ is empty.
      Since $U = [\lp{v_0}] \cdot [\lp{v}_1] = W \cdot \gamma \cdot [\lp{v}_1]$, we find that $\gamma \cdot [\lp{v}_1]$ is a convex subpomset of $U$.
      By \cref{lemma:subpomset}, it then follows that $\gamma \cdot [\lp{v}_1]$ satisfies~\ref{item:alternating} and~\ref{item:succpred}.

      \item
      Now consider $u\in S_{\lp{v}_1}\cup\{*\}$ such that $\lambda_{\lp{v}_1}(u)\in\State$ and $v\in\dom(\lambda_{\lp{v}_1}(u))$.
      We immediately obtain $u\in S_{\lp{u}}$, $\lambda_{\lp{u}}(u)\in\State$ and $v\in\dom(\lambda_{\lp{u}}(u))$.
      As $[\lp{u}]$ satisfies~\ref{item:path}, there exists a path $p_v$ for $v$ from $u$ to $*_{\max}$ in $[\lp{u}]$.
      Since $U = [\lp{v}_0] \cdot [\lp{v}_1]$, we know the nodes of $p_v$ all also exist in $S_{\lp{v}_1}$, and the unique state-labelled maximum of $U$, $*_{\max}$, is also the unique state-labelled maximum of $\gamma \cdot [\lp{v}_1]$.
      We immediately obtain that $p_v$ is a path for $v$ from $u$ to $*_{\max}$ in $\gamma \cdot [\lp{v}_1]$. Hence, $\gamma \cdot [\lp{v}_1]$ satisfies~\ref{item:path}.

      \item
      Take $u\in S_{\lp{v}_1}$ s.t. $\lambda_{\lp{v}_1}(u)= v\leftarrow n$ for some $v\in \var$ and $n\in \val$. Thus $u\in S_{\lp{u}}$ such that $\lambda_{\lp{u}}(u)= v\leftarrow n$.
      Because $[\lp{u}]$ satisfies~\ref{item:effect}, we obtain a node $s\in S_\lp{u}$ such that $s$ is the successor
      of $u$ and $\lambda_{\lp{u}}(s)(v)=n$. We know that $U=[\lp{v}_0] \cdot [\lp{v}_1]$. Hence, $s\in S_{\lp{v}_1}$.
    We can conclude that $s\in S_{\lp{v}_1}$ such that $\lambda_{\lp{v}_1}(s)(v)=n$.
      From \cref{lemma:succandpredtwo}, we infer that $s$ is also the successor of $u$ in $[\lp{v}_1]$.
      Hence, $\gamma \cdot [\lp{v}_1]$ satisfies~\ref{item:effect}.

      \item
      Take $u\in S_{\lp{v}_1}$ s.t. $\lambda_{\lp{v}_1}(u)= v\leftarrow v'$ for some $v,v'\in \var$. Thus $u\in S_{\lp{u}}$ such that $\lambda_{\lp{u}}(u)= v\leftarrow v'$.
      Because $[\lp{u}]$ satisfies~\ref{item:annoyingassignments}, we obtain nodes $p,s\in S_\lp{u}$ such that $p$ and $s$ are respectively the
      predecessor and successor of $u$, $v'\in\dom(\lambda_\lp{u}(p))$, and
      $\lambda_\lp{u}(s)(v)=\lambda_\lp{u}(s)(v')=\lambda_\lp{u}(p)(v')$.
      Immediately we obtain $s\in S_{\lp{v}_1}$. Because $\gamma \cdot [\lp{v}_1]$ has a minimum state-labelled node, we know that the predecessor of $u$ in $U=[\lp{v}_0] \cdot [\lp{v}_1]$ is also a node in $S_{\lp{v}_1}\cup\{*\}$.
      Hence, $p,s\in S_{\lp{v}_1}\cup\{*\}$ s.t.  $\lambda_{\lp{v}_1}(s)(v)=\lambda_{\lp{v}_1}(s)(v') = \lambda_{\lp{v}_1}(p)(v')$.
      From \cref{lemma:succandpredtwo}, we infer that $s$ and $p$ are the successor and predecessor of $u$ in $\gamma\cdot[\lp{v}_1]$.
      Thus $\gamma\cdot [\lp{v}_1]$ satisfies~\ref{item:annoyingassignments}.

      \item
      If $u\in S_{\lp{v}_1}\cup\{*\}$ s.t. $\lambda_{\lp{v}_1}(u)\in\State$ and $v\in\dom(\lambda_{\lp{v}_1}(u))$,
      we know that $u\in S_{\lp{u}}$, $\lambda_{\lp{u}}(u)\in\State$ and
      $v\in\dom(\lambda_{\lp{u}}(u))$.
      As $[\lp{u}]$ satisfies~\ref{item:prepath}, there is a path $p_v$ for $v$ from $s\in S_\lp{u}$ to $u$ such that either $s=*_{\min}$ or $s$ is the successor of a node $a$ with label $v\leftarrow k$ for $k\in \var\cup \val$.
      If $s=*_{\min}$, we know that $*_{\min}\leq_\lp{u} * \leq_\lp{u} u$. Suppose there exists $w\in S_\lp{u}$ such that $*_{\min}\leq_\lp{u} w$.
      If $w\in S_{\lp{v}_0}$, then since $U = [\lp{v}_0] \cdot [\lp{v}_1]$ and $*$ is the local $\leq_{\lp{v}_0}$-maximum of $[\lp{v}_0]$, we obtain $w\leq_\lp{u} *$.
      If $w\in S_{\lp{v}_1}$, $*\leq_\lp{u} w$.
      Then we apply \cref{lemma:bottleneck} to conclude that $*$ is on path $p_v$.
      From \cref{lemma:subpath} we obtain a path $t_v$ for $v$ from $*$ to $u$ in $[\lp{u}]$. As these are all nodes that also occur in $S_{\lp{v}_1}\cup\{*\}$, $t_v$ is also a path for $v$ from $*$ to $u$ in $\gamma \cdot [\lp{v}_1]$.
      As $*$ is the unique state-labelled minimum of $\gamma \cdot [\lp{v}_1]$, this proves that $\gamma \cdot [\lp{v}_1]$ satisfies~\ref{item:prepath}.

      If $s\neq*_{\min}$, we have two cases. If $s\in S_{\lp{v}_1}$, because the minimal nodes of $[\lp{v}_1]$ are assignment-labelled, we know that $a$ is also a node in $[\lp{v}_1]$ and the path $p_v$ only contains nodes
      that are in $S_{\lp{v}_1}$.
      Thus we have a path for $v$ from $s$ to $u$, with $s$ the successor of assignment-node $a$ (\cref{lemma:succandpredtwo}) and the label of $a$ is $v\leftarrow k$. If
      $s\in S_{\lp{v}_0}$, then also $a\in S_{\lp{v}_0}$ and $s\leq_{\lp{u}} * \leq_{\lp{u}}u$.
      Suppose there exists $w\in S_\lp{u}$ such that $s \leq_\lp{u} w$. If $w\in S_{\lp{v}_0}$, then $w\leq_\lp{u} *$.
      If $w\in S_{\lp{v}_1}$, as $U = [\lp{v}_0] \cdot [\lp{v}_1]$ and $*$ is the local $\leq_{\lp{v}_0}$-maximum of $[\lp{v}_0]$, we obtain $*\leq_\lp{u} w$.
      We apply \cref{lemma:bottleneck} to conclude that $*$ is on path $p_v$.
      From \cref{lemma:subpath} we obtain a path $t_v$ for $v$ from $*$ to $u$ in $[\lp{u}]$. As these are all nodes that also occur in $S_{\lp{v}_1}\cup\{*\}$, $t_v$ is also a path for $v$ from $*$ to $u$ in $\gamma \cdot [\lp{v}_1]$.
      As $*$ is the unique state-labelled minimum of $\gamma \cdot [\lp{v}_1]$, this proves that also in this case $\gamma \cdot [\lp{v}_1]$ satisfies~\ref{item:prepath}.
    \end{enumerate}

    Since $U = W \cdot \gamma \cdot [\lp{v}_1]$ and $W$ is non-empty, we know that $\gamma \cdot [\lp{v}_1]$ is smaller than $U$. Finally, since $[\lp{v}_0] = W \cdot
    \gamma$ and $\gamma \cdot [\lp{v}_1]$ satisfy~\ref{item:minmax}--\ref{item:prepath} and are strictly smaller than $U$, we can conclude by the induction hypothesis that both are guarded.
    This implies that $W \cdot \gamma \cdot [\lp{v}_1] = [\lp{v}_0] \cdot [\lp{v}_1] = U$ is guarded by definition.

    \item
    Suppose that $V = V_0 \parallel V_1$ for some non-empty series-parallel pomsets $V_0$ and $V_1$.
    As $U$ satisfies~\ref{item:minmax}--\ref{item:prepath}, we know that for each $v\in\dom(\alpha)$ there exists a path $p_v$ from $*_{\min}$ (note that $\lambda_\lp{u}(*_{\min})=\alpha$) to $*_{\max}$. As every node on $p_v$ is
    related via $\leq_{\lp{u}}$, we know $p_v$ exists out of nodes from either only $V_0$ or just $V_1$.
    This leads to the following definition of $\alpha_0$, $\beta_0$, $\alpha_1$, and $\beta_1$. For $v\in \var$:
    \[
    \alpha_0(v) =
    \begin{cases}
    \alpha(v) & \exists \text{ a path }p_v\text{ from }*_{\min}\text{ to }*_{\max}\text{ in }U\text{ that only uses nodes in }V_0\\
    \text{undefined} & \text{otherwise}
    \end{cases}
    \]
    and
    \[
    \beta_0(v) =
    \begin{cases}
    \beta(v) & \exists u_0\in S_{\lp{v}_0}\text{ s.t. }\lambda_{\lp{v}_0}(u_0)=v\leftarrow k\\
    \beta(v) & \exists \text{ a path }p_v\text{ from }*_{\min}\text{ to }*_{\max}\text{ in }U\text{ that only uses nodes in }V_0\\
    \text{undefined} & \text{otherwise}
    \end{cases}
    \]

    We define $\alpha_1$ and $\beta_1$ analogously. We claim that $\alpha_0\oplus\alpha_1$ is defined and equal to $\alpha$. For $v\in\dom(\alpha)$, there is a path for $v$ from $*_{\min}$ to $*_{\max}$ in $U$ by~\ref{item:path}. This
    path runs either through $V_0$ or $V_1$. Hence, $\alpha_0(v)=\alpha(v)$ or $\alpha_1(v)=\alpha(v)$ (or both) holds, which implies $\alpha_0\oplus\alpha_1(v)=\alpha(v)$.
    For $v\in\dom(\alpha_0\oplus\alpha_1)$, we know without loss of generality that $v\in\dom(\alpha_0)$ and then by construction we know that $\alpha_0\oplus\alpha_1(v)=\alpha_0(v)=\alpha(v)$.

    Similarly, we show that $\beta_0\oplus\beta_1$ is defined and equal to $\beta$. For $v\in\dom(\beta)$, as $U$ satisfies~\ref{item:prepath}, we have two cases. In the first case, there exists a path for $v$ from $*_{\min}$ to $*_{\max}$ that runs either entirely through $V_0$ or entirely through $V_1$. Hence, $\beta_0(v)=\beta(v)$
    or $\beta_1(v)=\alpha(v)$ (or both) holds, which implies $\beta_0\oplus\beta_1(v)=\beta(v)$.
    In the other case, there exists $s\in S_{\lp{v}_0}$ such that $\lambda_{\lp{v}_0}(s)=v\leftarrow k$ or $s\in S_{\lp{v}_1}$ such that $\lambda_{\lp{v}_1}(s)=v\leftarrow k$.
    W.l.o.g.\ we assume the former.
    Thus $\beta_0(v)=\beta(v)$.
    For $v\in\dom(\beta_0\oplus\beta_1)$, we know w.l.o.g.\ that $v\in\dom(\beta_0)$ and then by construction we know that $\beta_0\oplus\beta_1(v)=\beta_0(v)=\beta(v)$.

    We write $[\lp{v}_0]= \alpha_0 \cdot V_0 \cdot \beta_0$ and $[\lp{v}_1]= \alpha_1 \cdot V_1 \cdot \beta_1$. Without loss of generality, we assume that $S_\lp{u} = S_\lp{v_0} \cup S_\lp{v_1}$ and ${\leq_\lp{u}} =
    {\leq_\lp{v_0}} \cup {\leq_\lp{v_1}}$,
    with $S_\lp{v_0} \cap S_\lp{v_1} = \{ *_{\min}, *_{\max} \}$, with $*_{\min}$ the node with label
    $\alpha_0$ in $[\lp{v_0}]$ and label $\alpha_1$ in $[\lp{v_1}]$ and $*_{\max}$ the node with label $\beta_0$ in
    $[\lp{v_0}]$ and label $\beta_1$ in $[\lp{v_1}]$. We also have $\lambda_\lp{u}(*_{\min})=\alpha$
    and $\lambda_\lp{u}(*_{\max})=\beta$ and for all nodes $v\in S_\lp{u}\setminus
    \{*_{\min},*_{\max}\}$ have $\lambda_\lp{u}(v) = \lambda_\lp{v_0}(v)$
    if $v\in S_\lp{v_0}$ and $\lambda_\lp{v_1}(v)$ if $v\in S_\lp{v_1}$.

    We now argue that $\alpha_0 \cdot V_0 \cdot \beta_0$ and $\alpha_1 \cdot V_1 \cdot \beta_1$ satisfy~\ref{item:minmax}--\ref{item:prepath}.
    We only show the argument for $\alpha_0 \cdot V_0 \cdot \beta_0$, as the proof for $\alpha_1 \cdot V_1 \cdot \beta_1$ is identical.
    \begin{enumerate}[label={(A\arabic*)},leftmargin=1cm]
      \setcounter{enumi}{2}
      \item
    It is immediate that $\alpha_0 \cdot V_0 \cdot \beta_0$ is a convex subpomset of $U$, and hence by \cref{lemma:subpomset},
    we can conclude that $\alpha_0 \cdot V_0 \cdot \beta_0$ satisfies~\ref{item:alternating} and~\ref{item:succpred}.

    \item
    If $u\in S_{\lp{v}_0}$ s.t. $\lambda_{\lp{v}_0}(u)\in\State$ and $v\in\dom(\lambda_{\lp{v}_0}(u))$,
    we get $u\in S_\lp{u}$ s.t. $v\in\dom(\lambda_\lp{u}(u))$.
    As $U$ satisfies~\ref{item:path}, we obtain a path $p_v$ for $v$ from $u$ to $*_{\max}$.
    We distinguish three cases. If $u\in S_{\lp{v}_0}\setminus\{*_{\min},*_{\max}\}$, then $p_v$ only uses nodes in $V_0$ and from~\ref{item:prepath} we obtain $\beta_0(v)=\beta(v)$
    (there exists a path for $v$ from $*_{\min}$ to $u$ which combined with $p_v$ forms a path from $*_{\min}$ to $*_{\max}$ for $v$ that uses only nodes in $V_0$, or there exists a node $s\in S_{\lp{v}_0}$ such that $\lambda_{\lp{v}_0}(s)=v\leftarrow k$).
    This makes $p_v$ a path for $v$ from $u$ to $*_{\max}$ in $[\lp{v}_0]$.
    If $u = *_{\min}$, then by construction of $\alpha_0$ we obtain a path $q_v$ for $v$ from $*_{\min}$ to $*_{\max}$ in $U$ that only uses nodes in $V_0$ (note that $v\in\dom(\lambda_{\lp{v}_0}(u))=\dom(\alpha_0)$). By definition of $\beta_0$ we have then that $\beta_0(v)=\beta(v)$.
    The path $q_v$ is then immediately a path from $*_{\min}$ to $*_{\max}$ for $v$ in $[\lp{v}_0]$.
    If $u=*_{\max}$, we get a trivial path from $u$ to $*_{\max}$ for $v$. Hence, $\alpha_0 \cdot V_0 \cdot \beta_0$ satisfies~\ref{item:path}.

    \item
    If $u\in S_{\lp{v}_0}$ s.t. $\lambda_{\lp{v}_0}(u)=v\leftarrow n$ and $n\in \val$, then we get $u\in S_\lp{u}$ s.t. $\lambda_{\lp{u}}(u)=v\leftarrow n$.
    As $U$ satisfies~\ref{item:effect}, we know that the successor of $u$ is $s$ such that $\lambda_{\lp{u}}(s)(v)=n$.
    We distinguish two cases. The first case is $s\in S_{\lp{v}_0}\setminus\{*_{\min},*_{\max}\}$ with $\lambda_{\lp{v}_0}(s)(v)= n$. From \cref{lemma:succandpredtwo}, we infer that $s$ is the successor of $u$ in $[\lp{v}_0]$ as well.
    In the other case we have $s=*_{\max}$. By definition we have $\beta_0(v)=\beta(v)=n$. To see that $\beta_0$ is the successor of $u$ in $[\lp{v}_0]$, we use \cref{lemma:succandpredtwo}.
  Hence, $\alpha_0 \cdot V_0 \cdot \beta_0$ satisfies~\ref{item:effect}.

    \item
    If $u\in S_{\lp{v}_0}$ such that $\lambda_{\lp{v}_0}(u)=v\leftarrow v'$ for $v,v'\in \var$, we get $u\in S_\lp{u}$ such that $\lambda_{\lp{u}}(u)=v\leftarrow v'$.
    As $U$ satisfies~\ref{item:annoyingassignments}, we obtain $p,s\in S_\lp{u}$ such that $p$ and $s$ are resp.\ the predecessor and successor of $u$ and $\lambda_{\lp{u}}(s)(v)=\lambda_{\lp{u}}(s)(v')=\lambda_{\lp{u}}(p)(v')$.
    \begin{itemize}
      \item Let $p\in S_{\lp{v}_0}\setminus\{*_{\min},*_{\max}\}$. From \cref{lemma:succandpredtwo}, we infer that $p$
    is the predecessor of $u$ in $[\lp{v}_0]$ as well. We then have two cases. Either $s\in
    S_{\lp{v}_0}\setminus\{*_{\min},*_{\max}\}$ with
    $\lambda_{\lp{v}_0}(s)(v')=\lambda_{\lp{v}_0}(s)(v)=\lambda_{\lp{v}_0}(p)(v')$, in which case
    via \cref{lemma:succandpredtwo} we are done immediately, or $s=*_{\max}$.
    In the latter case, by construction of $\beta_0$ we have $\beta_0(v)=\beta(v)$. To see that $\beta_0$ is the successor of $u$ in $[\lp{v}_0]$, we use \cref{lemma:succandpredtwo}. From $U$ satisfying~\ref{item:prepath}, we infer that there exists
    a path $p_{v'}$ from $w\in S_\lp{u}$ to $p$ such that either $w=*_{\min}$ or $w$ is the successor of
    an assignment-node $a$ with label $v'\leftarrow k$ for $k\in \var\cup \val$. In the first case, we obtain a path from $s$ to $*_{\max}$ which combined with $p_{v'}$ forms a path from $*_{\min}$ to
    $*_{\max}$ for $v'$, and all these nodes are in $[\lp{v}_0]$. Hence, by construction of $\beta_0$ we obtain $\beta_0(v')=\beta(v')$. In the second case, we know node $a\in S_\lp{v_0}$, and thus $\beta_0(v')=\beta(v')$.
  Hence,  $\beta_0(v')=\beta(v')=\beta(v)=\beta_0(v)=\lambda_{\lp{u}}(p)(v')=\lambda_{\lp{v}_0}(p)(v')$, and we are done.
  \item
    If $p=*_{\min}$, we have two cases. First, we consider the case where $s \in S_{\lp{v}_0}\setminus\{*_{\min},*_{\max}\}$. From \cref{lemma:succandpredtwo}, we know that
    $*_{\min}$ and $s$ are resp.\ the predecessor and successor of $u$ in $[\lp{v_0}]$. We also have
    $\lambda_{\lp{v_0}}(s)(v)=\lambda_{\lp{v_0}}(s)(v')=\lambda_{\lp{u}}(p)(v')$. From $U$ satisfying~\ref{item:path}, we know there is a
    path $p_{v'}$ from $s$ to $*_{\max}$ for $v'$. This path only uses nodes in $V_0$. Then, we know that
    $p_{v'}$ together with the node $u$ forms a path for $v'$ from $*_{\min}$ to $*_{\max}$ using nodes in
    $V_0$. By construction, we have $\alpha_0(v')=\alpha(v')=\lambda_{\lp{u}}(p)(v')$.
    In the other case, we have $s=*_{\max}$. By construction we have $\alpha_0(v')=\alpha(v')$, and $\beta_0(v')=\beta(v')$. We also get immediately that $\beta_0(v)=\beta(v)$.
    Thus we have $\lambda_{\lp{v}_0}(s)(v)=\beta_0(v)=\beta(v)=\beta(v')=\beta_0(v')=\lambda_{\lp{v}_0}(s)(v')$ and $\lambda_{\lp{v}_0}(s)(v')=\beta(v')=\alpha(v')=\alpha_0(v')=\lambda_{\lp{v}_0}(p)(v')$.
    This concludes the proof that $\alpha_0 \cdot V_0 \cdot \beta_0$ satisfies~\ref{item:annoyingassignments}.
  \end{itemize}

    \item
    If $u\in S_{\lp{v}_0}$ s.t. $\lambda_{\lp{v}_0}(u)\in\State$ and $v\in\dom(\lambda_{\lp{v}_0}(u))$, $u\in S_\lp{u}$ such that $v\in\dom(\lambda_\lp{u}(u))$.
    As $U$ satisfies~\ref{item:prepath}, there exists a path $p_v$ for $v$ from $s$ to $u$ such that $s=*_{\min}$ or $s$ is the successor of an assignment-node with label $v\leftarrow k$ with $k\in \var\cup \val$.
    In the former case, we distinguish three cases.
    If $u\in S_{\lp{v}_0}\setminus\{*_{\min},*_{\max}\}$, then $p_v$ runs entirely through $V_0$. A $U$ satisfies~\ref{item:path}, there exists a path $t_v$ from $u$ to $*_{\max}$ using only nodes in $V_0$. Combining $p_v$ and $t_v$ we obtain a path from
    $*_{\min}$ to $*_{\max}$ for
    $v$ through $V_0$. Then by construction $\alpha_0(v)=\alpha(v)$.
    This makes $p_v$ a path for $v$ from $\alpha_0$ to $u$ in $V_0$.
    If $u=*_{\min}$, the case is trivial.
    If $u=*_{\max}$, then $\alpha_0(v)=\alpha(v)$ and $\beta_0(v)=\beta(v)$.
    This makes $p_v$ a path for $v$ from $*_{\min}$ to $u$ in $[\lp{v_0}]$.
    In the second case, we distinguish two cases.
    If $u\in S_{\lp{v}_0}\setminus\{*_{\min},*_{\max}\}$, then $p_v$ uses only nodes in $V_0$ and we are done immediately.
    If $u=*_{\max}$ then by definition of $\beta_0$ we have either that there exists a node $w$ in $S_{\lp{v}_0}$ such that $\lambda_{\lp{v}_0}(w)=v\leftarrow k$ or there exists a path $t_v$ for $v$ from $*_{\min}$ to $*_{\max}$ using only nodes in $V_0$.
    In the former case, using the fact that we know that $\alpha_0 \cdot V_0 \cdot \beta_0$ already satisfies~\ref{item:effect}, we know the successor of $w$, $y$, is such that
    $\lambda_{\lp{v}_0}(y)(v)=k$. By~\ref{item:path}, this gives us a path $s_v$ for $v$ from $y$ to
    $\beta_0$ in $[\lp{v}_0]$.
    If there exists a path $t_v$ from $*_{\min}$ to $*_{\max}$ using only nodes in $V_0$, we know that $\alpha_0(v)=\alpha(v)$ and $\beta_0(v)=\beta(v)$, and thus $t_v$ is a path for $v$ from $\alpha_0$ to $\beta_0$ in $[\lp{v}_0]$.
    Hence, $\alpha_0 \cdot [\lp{v}_0]\cdot \beta_0$ satisfies~\ref{item:prepath}.
  \end{enumerate}

    This makes $\alpha_0 \cdot V_0 \cdot \beta_0$ and $\alpha_1 \cdot V_1 \cdot \beta_1$ satisfy~\ref{item:minmax}--\ref{item:prepath}, and they are strictly smaller than $U$, and hence by the induction hypothesis we know that they are guarded.
    This makes $U = \alpha_0\oplus\alpha_1 \cdot (V_0 \parallel V_1) \cdot \beta_0\oplus\beta_1$ a guarded pomset by definition.
    \qedhere
\end{itemize}
\end{proof}

\fi

\section{Proofs about the litmus test}

\ifpnotguarded*
\begin{proof}
 We prove by contradiction; assume that $P(U)$ and that $U$ is guarded. Via \cref{lemma:fromguardedtocausal} we conclude that $U$ satisfies~\ref{item:minmax}--\ref{item:prepath}. From $P(U)$ we infer that there
 exists $u_1,u_2,w\in S_\lp{u}$ such that $\lambda_\lp{u}(u_1)=(x\leftarrow 1)$, $u_1\leq w$ and
 $\lambda_\lp{u}(w)(r_0)=0=\lambda_\lp{u}(w)(r_1)$.
From~\ref{item:alternating} and~\ref{item:effect}, we infer that $u_1$ has a unique
successor node $s_1\in S_\lp{u}$ such that $s_1$ is state-labelled, and $\lambda_\lp{u}(s_1)(x)=1$. From~\ref{item:path} there exists a path for $x$ from $s_1$ to $w$. Hence, if
$\lambda_\lp{u}(w)(x)\neq 1$, there must be at least one assignment between $s_1$ and $w$ altering the
value of $x$, as the path must explain how the value of $x$ changed from $1$ to $0$. Hence, there exists a node $u_3\in S_\lp{u}$ such that $s_1\leq_\lp{u} u_3 \leq_\lp{u} w$ and
$\lambda_\lp{u}(u_3)=(x\leftarrow n)$ for $n\in \var\cup \val$. However, from property $P$, we know that all such assignments occur before $u_1$, and thereby strictly before $s_1$.
From this we can conclude that $\lambda_\lp{u}(w)(x) = 1$. Similarly, we obtain $\lambda_\lp{u}(w)(y) =  1$.

From~\ref{item:alternating} and~\ref{item:annoyingassignments}, we know that $v_1$ has a unique successor node $t_1$, such that $\lambda_{\lp{u}}(t_1)(r_0)=\lambda_{\lp{u}}(t_1)(y)$.
Then from~\ref{item:path}, there must be a path for $r_0$ from $t_1$ to $w$. With similar reasoning as for $x$ above, we obtain $\lambda_\lp{u}(t_1)(r_0)=\lambda_\lp{u}(w)(r_0)=0$. Similarly, we obtain a successor node $t_2$ of $v_2$ such that $\lambda_{\lp{u}}(t_2)(r_1)=\lambda_{\lp{u}}(t_2)(x)=0$.

As we have that $\lambda_\lp{u}(t_1)(y)=0$ and $\lambda_\lp{u}(w)(y)=1$ and $t_1\leq_\lp{u} w$, we can
conclude from~\ref{item:path} that there must be a path from $t_1$ to $w$ for $y$ such that
this path contains at least one assignment that alters the value for $y$. Thus, there exists a node
$u_3$ such that $t_1\leq_\lp{u} u_3 \leq_\lp{w} w$ and $u_3$ has a label that changes the value of $y$. Similarly, we obtain a node $u_4$ such that $t_2\leq_\lp{u} u_4 \leq_\lp{w} w$ and $u_4$ changes the value of
$x$. From property $P$, we obtain $u_3\leq_\lp{u} u_2$ and $u_4\leq_\lp{u} u_1$.
Then, making use of the fact that $t_1$ and $t_2$ are the successors of $v_1$ and $v_2$ respectively, we can derive:
$
v_2 \leq_\lp{u} t_2 \leq_\lp{u} u_4 \leq_\lp{u} u_1 \leq_\lp{u} v_1 \leq_\lp{u} t_1 \leq_\lp{u} u_3 \leq_\lp{u} u_2 \leq_\lp{u} v_2
$
Then, by antisymmetry, all these nodes are equivalent. As they cannot be, we have a contradiction. Hence, $U$ is not a guarded pomset. Hence, $U$ is not a guarded pomset.
\end{proof}

\ifarxiv%

To prove \cref{lemma:closure-vs-contraction}, we must first consider a series of auxiliary lemmas.

\begin{lemma}\label{lemma:preceq-seq}
Let $U_0,U_1,V\in\SP$. If $U_0\cdot U_1\preceq V$, then there exist $V_0,V_1\in\SP$ such that
\begin{mathpar}
V=V_0\cdot V_1
\and
U_0\preceq V_0
\and
U_1\preceq V_1
\end{mathpar}
\end{lemma}
\begin{proof}
  We write $V=[\lp{v}]$, $U_0=[\lp{u_0}]$ and $U_1=[\lp{u_1}]$.
  For $i \in \{0, 1\}$, we choose $V_i = [\lp{v}_i]$ by
  \begin{mathpar}
  S_{\lp{v_i}}= \{v\in S_{\lp{v}} \pipe h(v)\in S_\lp{u_i}\}
  \and
  {\leq_\lp{v_i}} = {\leq_{\lp{v}}} \cap S_\lp{v_i}^2
  \and
  \lambda_{\lp{v_i}}(v)=\lambda_{\lp{v}}(v)
  \end{mathpar}
  where $h: S_\lp{v} \to S_{\lp{u}_0 \cdot \lp{u}_1}$ witnesses that $U_0 \cdot U_1 \preceq V$.
  We first verify whether $V=V_0\cdot V_1$.
  \begin{itemize}
    \item
    For the carrier, take $x \in S_\lp{v}$.
    This establishes that $h(x) \in S_{\lp{u}_0}$ or $h(x) \in S_{\lp{u}_1}$.
    Hence, $x \in S_{\lp{v_0}}$ or $x \in S_{\lp{v_1}}$.
    The converse inclusion holds by construction.

    \item
    To check that ${\leq_{\lp{v_0} \cdot \lp{v_1}}} = {\leq_\lp{v}}$, first suppose that $x\leq_\lp{v} y$.
    We have three cases to distinguish.
    \begin{itemize}
      \item
      If $h(x),h(y)\in S_\lp{u_0}$, then $x,y\in S_\lp{v_0}$ and we immediately obtain that $x\leq_{\lp{v_0}}y$ and thus that $x\leq_{\lp{v_0\cdot v_1}}y$.
      Same for $h(x),h(y)\in S_\lp{u_1}$.

      \item
      On the other hand, if $h(x)\in S_\lp{u_1}$ and $h(y)\in S_\lp{u_0}$, then $h(y) \leq_{\lp{u}_0 \cdot \lp{u}_1} h(x)$.
      Since $h(x) \leq_{\lp{u}_0 \cdot \lp{u}_1} h(y)$ already, this implies that $h(x) = h(y)$.
      But then $h(x) \in S_{\lp{u}_0}$, which contradicts that $S_{\lp{u}_0}$ is disjoint from $S_{\lp{u}_1}$.
      We can therefore disregard this case.

      \item
      This leaves the last possibility where $h(x)\in S_\lp{u_0}$ and $h(y)\in S_\lp{u_1}$.
      This establishes that $x\in S_\lp{v_0}$ and $y\in S_\lp{v_1}$, and we obtain that $x\leq_{\lp{v_0\cdot v_1}}y$.
    \end{itemize}

    \noindent
    Conversely, if $x \leq_\lp{v_0\cdot v_1} y$, then we distinguish two possibilities.
    \begin{itemize}
      \item
      If $x,y\in S_{\lp{v_0}}$, then $x\leq_\lp{v_0} y$ so $x\leq_\lp{v} y$ immediately, and similarly when $x,y\in S_{\lp{v_1}}$.

      \item
      On the other hand, if $x\in S_{\lp{v_0}}$ and $y\in S_\lp{v_1}$, then $h(x)\leq_\lp{u_0\cdot u_1} h(y)$.
      This means by properties of $\preceq$ that either $x\leq_\lp{v} y$, in which case we are done, or $y\leq_\lp{v} x$.
      In the latter case, we obtain that $h(y)\leq_\lp{u_0\cdot u_1} h(x)$, which results in $h(x) = h(y)$.
      This contradicts the fact that $S_{\lp{u}_0}$ is disjoint from $S_{\lp{u}_1}$, which means that we can disregard this case.
  \end{itemize}

  \item
  For the labels it follows immediately that $[\lp{v_0\cdot v_1}]$ gives the same labels as $[\lp{v}]$.
\end{itemize}

\noindent
The next thing to show is that $U_0 \preceq V_0$ and $U_1 \preceq V_1$.
Both cases are similar, so we only prove $U_0 \preceq V_0$.
We take as a witness the function $h_0: S_{\lp{v}_0} \to S_{\lp{u}_0}$ given by $h_0(x) = h(x)$.
\begin{itemize}
  \item
  First we prove that this is a surjective function from $V_0$ to $U_0$.
  For $v\in S_\lp{v_0}$, we know that $h_0(v) = h(v) \in S_\lp{u_0}$ by construction.
  For $u\in S_\lp{u_0}$, we know there exists $v\in S_{\lp{v}}$ such that $h(v)=u$ by surjectivity of $h$ on $U_0\cdot U_1$.
  Suppose that $v\in S_{\lp{v_1}}$.
  By construction of $S_\lp{v_1}$ this means that $h(v) \in S_\lp{u_1}$.
  This is a contradiction, thus we must have that $v\in S_{\lp{v_0}}$, and we conclude that $h$ is also surjective when restricted to $S_\lp{v_0}$.

  \item
  If $x\leq_\lp{v_0} y$, then $x\leq_\lp{v} y$ thus $h(x) \leq_\lp{u_0\cdot u_1} h(y)$.
  Since both $h(x),h(y)\in S_{\lp{u_0}}$, we obtain that $h(x) \leq_\lp{u_0} h(y)$.

  \item
  If $h(x)\leq_\lp{u_0}h(y)$, then $x,y\in S_\lp{v_0}$. If $\lambda_\lp{v_0}(x)\in\Act$ or $\lambda_\lp{v_0}(y)\in\Act$, then $\lambda_\lp{v}(x)\in\Act$ or $\lambda_\lp{v}(y)\in\Act$.
  From $h(x)\leq_\lp{u_0}h(y)$ we obtain that $h(x)\leq_\lp{u_0\cdot u_1}h(y)$.
  Thus it follows that $x \leq_\lp{v} y$.
  Since $h(x),h(y)\in S_\lp{u_0}$, we know that $x,y\in S_\lp{v_0}$, and we can conlude $x \leq_\lp{v_0} y$.

  \item
  The last requirement (if $h(x)\leq_\lp{u_0}h(y)$ and $\lambda_\lp{v_0}(x),\lambda_\lp{v_0}(y)\in\State$, then $x\leq_\lp{v_0} y$ or $y \leq_\lp{v_0}x$) is checked similarly.
  \qedhere
\end{itemize}
\end{proof}

\begin{lemma}\label{lemma:preceq-parallel}
Let $U_0,U_1,V\in\SP$. If $U_0\parallel U_1\preceq V$, then there exist $V_0,V_1\in\SP$ such that
\begin{mathpar}
V = V_0\parallel V_1
\and
U_0\preceq V_0
\and
U_1\preceq V_1
\end{mathpar}
\end{lemma}
\begin{proof}
  We write $V=[\lp{v}]$, $U_0=[\lp{u_0}]$ and $U_1=[\lp{u_1}]$.
  For $i \in \{0, 1\}$, we choose $V_i = [\lp{v}_i]$ by
  \begin{mathpar}
  S_{\lp{v_i}}= \{v\in S_{\lp{v}} \pipe h(v)\in S_\lp{u_i}\}
  \and
  {\leq_\lp{v_i}} = {\leq_{\lp{v}}} \cap S_\lp{v_i}^2
  \and
  \lambda_{\lp{v_i}}(v)=\lambda_{\lp{v}}(v)
  \end{mathpar}
  where $h: S_\lp{v} \to S_{\lp{u}_0 \cdot \lp{u}_1}$ witnesses that $U_0 \cdot U_1 \preceq V$.
  We first verify whether $V=V_0\parallel V_1$.
  \begin{itemize}
    \item
    Take $x\in S_\lp{v}$.
    Then we know that $h(x)$ in $U_0$ or in $U_1$ from $U_1\parallel U_2\preceq V$. Hence, we can conclude that $x\in S_{\lp{v_0}}$ or $x\in S_{\lp{v_1}}$.
    The converse inclusion holds by construction.

    \item
    To check that ${\leq_\lp{v_0\parallel v_1}} = {\leq_\lp{v}}$ we check both directions.
    If $x\leq_\lp{v} y$, then we know that $h(x) \leq_{\lp{u_0} \parallel \lp{u_1}} h(y)$.
    In that case $h(x) \leq_{\lp{u}_i} h(y)$ for $i \in \{0,1\}$.
    In that case we know that $x, y \in S_{\lp{v}_i}$ and we immediately obtain that $x \leq_{\lp{v_i}} y$ and thus that $x\leq_{\lp{v_i\parallel v_i}}y$.
    The converse claim, i.e., that whenever $x \leq_{\lp{v}_0 \parallel \lp{v}_1} y$ also $x \leq_{\lp{v}} y$ holds by construction.

    \item
   For the labels it follows immediately that $[\lp{v_0\parallel v_1}]$ gives the same labels as $[\lp{v}]$.
  \end{itemize}

  \noindent
  It remains to show $U_0\preceq V_0$ and $U_1\preceq V_1$; this is analagous to the proof of \cref{lemma:preceq-seq}.
\end{proof}

\begin{lemma}\label{lemma:preceq-smallest-preorder}
  Let $\preceq^{\SP}$ be $\preceq$ restricted to $\SP$.
  Then $\preceq^{\SP}$ is the smallest precongruence (preorder monotone w.r.t.\ the operators) such that for all $\alpha \in \State$ we have that $\alpha \preceq^{\SP} \alpha \cdot \alpha$.
\end{lemma}
\begin{proof}
  We first show that $\preceq^{\SP}$ is a preorder monotone w.r.t.\ the operators. Thus we need to prove that $\preceq^{\SP}$ is reflexive, transitive and monotone. Reflexivity follows immediately by using the identity function as a witness, and transivity by
sing function composition.

  For monotonicity we check whether for $U_0\preceq^{\SP}V_0$ and $U_1\preceq^{\SP}V_1$ we have $U_0\cdot U_1\preceq^{\SP}V_0\cdot V_1$ and $U_0\parallel U_1\preceq^{\SP}V_0\parallel V_1$.
  For $i \in \{0, 1\}$, let $U_i = [\lp{u}_i]$ and $V_i = [\lp{v}_i]$, and let $h_i: S_{\lp{v}_i} \to S_{\lp{u}_i}$ be the function witnessing that $U_i \preceq V_i$.
  We choose $h$ as the union of $h_0$ and $h_1$.
  To see that $h$ witnesses that $U_0 \cdot U_1 \preceq^{\SP} V_0 \cdot V_1$, we have to check the conditions of \cref{definition:contraction-pomsets}.
  \begin{itemize}
    \item
    On the one hand, if $v\leq_{\lp{v_0\cdot v_1}}v'$ because $v\in S_{\lp{v_0}}$ and $v'\in S_{\lp{v_1}}$, we immediately obtain that $h(v)=h_0(v)\leq_{\lp{u_0\cdot u_1}}h_1(v')=h(v')$, as we know that $h_0(v)\in S_{\lp{u_0}}$ and $h_1(v')\in S_{\lp{u_1}}$.

    On the other hand, if $v\leq_{\lp{v_0\cdot v_1}}v'$ because $v,v'\in S_{\lp{v_i}}$ for $i\in\{0,1\}$, we can establish immediately that $h_i(v)\leq_\lp{u_i} h_i(v')$ and thus $h(v)\leq_\lp{u_0\cdot u_1} h(v')$.

    \item
    Next, suppose that $x, y \in S_{\lp{v}_0} \cup S_{\lp{v}_1}$ such that $h(x) \leq_{\lp{u}_0 \cdot \lp{u}_1} h(y)$ and $\lambda_{\lp{v}_0 \cdot \lp{v}_1}(x), \lambda_{\lp{v}_0 \cdot \lp{v}_1}(y) \in \State$.
    We should show that $x \leq_{\lp{v}_0 \cdot \lp{v}_1} y$ or $y \leq_{\lp{v}_0 \cdot \lp{v}_1} x$.
    We distinguish three cases:
    \begin{itemize}
      \item
      If $x \in S_{\lp{v_0}}$ and $y \in S_\lp{v_1}$, then we know immediately that $x \leq_{\lp{v_0\cdot v_1}} y$ and we are done.

      \item
      If $x \in S_\lp{v_1}$ and $y \in S_\lp{v_0}$, we have $y \leq_{\lp{v_0\cdot v_1}} x$ and we are also done.

      \item
      If $x, y \in S_{\lp{v_i}}$ for some $i \in \{0, 1\}$, then we know that $h(x),h(y)\in S_{\lp{u_i}}$, and thus that $h_i(v)\leq_{\lp{u_i}}h_i(v')$.
      As $\lambda_{\lp{v}_i}(x) = \lambda_{\lp{v}_0 \cdot \lp{v}_1}(x)$ and $\lambda_{\lp{v}_i}(y) = \lambda_{\lp{v}_0 \cdot \lp{v}_1}(y)$, we can establish that $x \leq_\lp{v_i} y$ or $y \leq_\lp{v_i} x$.
      Hence, $x \leq_{\lp{v}_0 \cdot \lp{v}_1} y$ or $y \leq_{\lp{v}_0 \cdot \lp{v}_1} x$.
    \end{itemize}

    \item
    Now uppose that $x, y \in S_{\lp{v}_0} \cup S_{\lp{v}_1}$ such that $h(x) \leq_{\lp{u}_0 \cdot \lp{u}_1} h(y)$ and $\lambda_{\lp{v}_0 \cdot \lp{v}_1}(x) \in \Act$ or $\lambda_{\lp{v}_0 \cdot \lp{v}_1}(y) \in \Act$.
    By an argument similar to the previous case, we can argue $x \leq_{\lp{v}_0 \cdot \lp{v}_1} y$.

    \item
    The labelling requirements are satisfied immediately.
  \end{itemize}
  A similar argument can be used to show that $h$ witnesses $U_0 \parallel U_1 \preceq^{\SP} V_0 \parallel V_1$.

  The last thing to check is whether for $\alpha \in \State$ we have that $\alpha \preceq^{\SP} \alpha \cdot \alpha$.
  This is witnessed by the unique function from the carrier of the latter to that of the former.

  \smallskip
  Next thing to show is that $\preceq^{\SP}$ is the smallest preorder for which all of the conditions hold.
  Suppose we have another preorder $\triangleleft$ that fulfills the conditions.
  We now need to show that if $U\preceq^{\SP} V$, then $U\triangleleft V$.
  We do this by induction on $U$, as $U$ is a series-parallel pomset.
  In the base we have three cases to consider.
  \begin{itemize}
    \item
    If $U = 1$, then $U\preceq^{\SP} V$ means that $V=1$, and thus $U \triangleleft V$.

    \item
    If $U = \ltr{a}$ for $\ltr{a} \in \Act$, $U \preceq^{\SP} V$ can only hold when $V = a$, and thus $U \triangleleft V$ again.

    \item
    If $U = \alpha$ for some $\alpha \in \State$, then $U \preceq^{\SP} V$ implies that $V = \alpha^n$ for some $n \geq 1$.
    A straightforward argument by induction on $n$ then shows that $U \triangleleft V$.
  \end{itemize}

  \noindent
  In the inductive step we have two cases.
  \begin{itemize}
    \item
    If $U=U_0\cdot U_1$ and $U\preceq^{\SP} V$.
    We know from \cref{lemma:preceq-seq} that $V=V_0\cdot V_1$ and $U_0\preceq^{\SP} V_0$ and $U_1\preceq^{\SP} V_1$.
    From the induction hypothesis we obtain that $U_0\triangleleft V_0$ and $U_1\triangleleft V_1$.
    As $\triangleleft$ is monotone w.r.t the operators, we obtain that $U\triangleleft V$.

    \item
    If $U=U_0\parallel U_1$ and $U\preceq^{\SP} V$.
    We know from \cref{lemma:preceq-parallel} that $V=V_0\parallel V_1$ and $U_0\preceq^{\SP} V_0$ and $U_1\preceq^{\SP} V_1$.
    From the induction hypothesis we obtain that $U_0\triangleleft V_0$ and $U_1\triangleleft V_1$.
    As $\triangleleft$ is monotone w.r.t.\ the operators, we obtain that $U\triangleleft V$.\qedhere
  \end{itemize}
\end{proof}

\noindent
Using these lemmas, we can now prove \cref{lemma:closure-vs-contraction}.

\lemmaclosurevscontraction*
\begin{proof}
  For the implication from left to right, we write $L\preceq K$ if for every $U\in L$, there exists $V\in K$ such that $U\preceq V$.
  We then reformulate the statement as
  \[
      \forall A\subseteq \closure[\hcontr]L,\text{ }A\preceq L.
  \]
  We proceed by induction on the construction of $\closure[\hcontr]L$.
  \begin{itemize}
    \item
    The base case is $A=L$, this one is trivial as $\preceq$ is reflexive.

    \item
    For the inductive step we have $A=C[\sembka{\alpha}]$ for $\alpha\in\State$ and $C[\sembka{\alpha\cdot\alpha}]\in\closure[\hcontr]L$.
    From the induction hypothesis we obtain $C[\sembka{\alpha\cdot\alpha}]\preceq L$.
    Now, if $U \in C[\sembka{\alpha}]$, then $U = C[\alpha]$.
    Since $C[\alpha] \preceq C[\alpha \cdot \alpha] \in C[\sembka{\alpha \cdot \alpha}]$, it follows that $C[\sembka{\alpha}] \preceq C[\sembka{\alpha \cdot \alpha}]$; by transitivity, we then have that $A = C[\sembka{\alpha}] \preceq L$.
  \end{itemize}

  \noindent
  For the direction from right to left, we first prove that if $C\in\PC$ and $U,V\in\SP$ such that $U\preceq V$ and $C[V]\in \closure[\hcontr]L$, then $C[U]\in \closure[\hcontr]L$, by induction on the construction of $\preceq$ as characterised in \cref{lemma:preceq-smallest-preorder}.
  In the base, there are two cases.
  \begin{itemize}
    \item
    If $U\preceq V$, because $U=V$, we find that $C[U]=C[V]\in \closure[\hcontr]L$ immediately.

    \item
    If $U\preceq V$ because $U = \alpha$ and $V = \alpha \cdot \alpha$ for some $\alpha \in \State$, then $C[\sem{\alpha \cdot \alpha}] \subseteq \closure[\hcontr]{L}$, which means that $C[U] \in C[\sem{\alpha}] \subseteq \closure[\hcontr]{L}$.
  \end{itemize}
  For the inductive step, there are three cases to consider.
  \begin{itemize}
    \item
    If $U \preceq V$ because $U=U_0\cdot U_1$ and $V=V_0\cdot V_1$ with $U_0\preceq^{\SP}V_0$ and $U_1\preceq^{\SP}V_1$, then first choose $C'=C[*\cdot V_1]$.
    Thus $C'[V_0]=C[V_0\cdot V_1]\in \closure[\hcontr]L$.
    From the induction hypothesis we obtain that $C'[U_0]\in \closure[\hcontr]L$ and $C'[U_0]=C[U_0\cdot V_1]$.
    Then take $C''=C[U_0\cdot *]$. Thus $C''[V_1]=C[U_0\cdot V_1]\in \closure[\hcontr]L$.
    Again from the induction hypothesis we get that $C''[U_1]\in\closure[\hcontr]L$, and $C''[U_1]=C[U_0\cdot U_1]=C[U]$.

    \item
    If $U\preceq V$ because $U=U_0\parallel U_1$ and $V=V_0\parallel V_1$ with $U_0\preceq^{SP}V_0$ and $U_1\preceq^{SP}V_1$, the proof proceeds as above.

    \item
    If $U\preceq V$ because there exists a $W\in\Pom$ and $U\preceq W$ and $W\preceq V$, then by the induction hypothesis we find that $C[W]\in\closure[\hcontr]L$.
    By applying the induction hypothesis again we can conclude that $C[U]\in\closure[\hcontr]L$.
  \end{itemize}
  Then if $V \in L$ s.t. $U \preceq V$, we can choose $C=*$ to find that $C[V]=V\in\closure[\hcontr]L $, and thus $U=C[U]\in\closure[\hcontr]L$, using the previously derived fact.
\end{proof}

\fi

\pclosure*
\begin{proof}
First, note that $\sempocka{e} = \closure[\hexch\cup\hcontr]{\semwpocka{e}} = \closurep[\hcontr]{\closure[\hexch]{\semwpocka{e}}}$ by~\cref{lemma:factorise-exch}.
Thus, if $V \in \sempocka{e}$, we apply \cref{lemma:closure-vs-contraction}, to infer that there exists a pomset $W\in \closure[\hexch]{\semwpocka{e}}$ such that
$V\preceq W$.
Next we can apply \cref{lemma:exch-closure-vs-subsumption}, to obtain a pomset $U\in \semwpocka{e}$ such that $W\sqsubseteq U$. We know that $U$ has property $P$. We first show that $W$ also has property $P$, and then that the same holds for $V$. From the definition of $\sqsubseteq$ we get that there exists a bijective pomset morphism $h$ from $W$ to $U$.
Thus we have $U=[\lp{u}]$ and $W=[\lp{w}]$ and a bijective function $h\colon S_\lp{u}\to S_\lp{w}$ such that $\lambda_\lp{w}\circ h =\lambda_\lp{u}$ and
if $u\leq_\lp{u} u'$ then $h(u)\leq_\lp{w} h(u')$. Now we need to verify the two properties of \cref{def:propq}.
\begin{enumerate}
    \item As $\lambda_\lp{w}(h(u_1))=\lambda_\lp{u}(u_1)$, we get
    $\lambda_\lp{w}(h(u_1))=(x\leftarrow 1)$. The same for the other existential statements of
    \cref{item:existence-in-u}. For the ordering: from $u_1 \leq_\lp{u} v_1 \leq_\lp{u} w$ we immediately obtain
    that $h(u_1)\leq_\lp{w} h(v_1)\leq_\lp{w} h(w)$ and similarly for
    $h(u_2)\leq_\lp{w}h(v_2)\leq_\lp{w}h(w)$.

    \item
    Take a $z$ such that $\lambda_\lp{w}(z)=(x\leftarrow n)$ for $ n\in \val\cup \var$. As $h$ is surjective, we know there exists a node $s\in S_\lp{u}$ such that $h(s)=z$ and $\lambda_\lp{w}(h(s))=\lambda_\lp{u}(s)$. As $P(U)$, we get that $s\leq_\lp{u} u_1$.
Hence, $h(s)\leq_\lp{w} h(u_1)$ and thus $z\leq_\lp{w} h(u_1)$. An analogue argument can be given for the other conditions in property~\eqref{item:relative-existence}.
 \end{enumerate}
This demonstrates that $W$ has property $P$. We know that $V\preceq W$, and we will show this implies that $V$ also has property $P$. From the definition of $\preceq$ we know that there exists a pomset morphism $h$ from $V$ to $W$. The argument to verify the two  properties of \cref{def:propq} is exactly the same as above. Hence we can conclude that $V$ has property $P$.
\end{proof}

\end{document}